\documentclass[11pt,english,numberwithinsect]{article}

\usepackage[numbers,sort&compress]{natbib}

\usepackage[left=1in,right=1in,top=1in,bottom=1in]{geometry}

\usepackage{amsthm}
\usepackage{amsmath}
\usepackage{amssymb}

\usepackage{algpseudocode}
\usepackage{algorithm}

\usepackage{graphicx}

\usepackage{tikz}

\usepackage{subcaption}

\usepackage{color}

\usepackage{wrapfig}

\usepackage{xspace}

\usepackage[bookmarks]{hyperref}

\usepackage{varioref} 

\usepackage{footnote} 

\usepackage{enumitem} 

\mathchardef\mhyphen="2D
\newcommand{\concat}{\mathop{\circ}}
\newcommand{\bigconcat}{\mathop{\bigcirc}}
\newcommand{\Next}{\textup{Next}}

\newcommand{\OVH}{\textup{OVH}\xspace}
\newcommand{\UOVH}{\textup{UOVH}\xspace}
\newcommand{\x}{{\textsc{x}}}
\newcommand{\y}{{\textsc{y}}}

\newcommand{\DTW}{\textup{DTW}}

\newcommand{\eps}{\ensuremath{\varepsilon}}

\newcommand{\rev}{\mathrm{rev}}

\newcommand{\occ}{\#}
\newcommand{\params}{{\mathcal{P}}}
\newcommand{\Valpha}{{\boldsymbol{\alpha}}}

\newcommand{\sA}{{\cal A}}
\newcommand{\sB}{{\cal B}}
\newcommand{\lift}{\mathop{\uparrow}}
\newcommand{\cross}{\mathrm{cr}}

\newcommand{\R}{{\mathbb{R}}}

\newcommand{\figref}[1]{Figure~\ref{fig:#1}}

\newcommand{\tabref}[1]{Table~\ref{tab:#1}}
\newcommand{\tabvref}[1]{Table~\vref{tab:#1}}
\newcommand{\defref}[1]{Definition~\ref{def:#1}}

\newcommand{\thmref}[1]{Theorem~\ref{thm:#1}}
\newcommand{\thmrefs}[2]{Theorems~\ref{thm:#1} and~\ref{thm:#2}}

\newcommand{\lemref}[1]{Lemma~\ref{lem:#1}}
\newcommand{\lemrefs}[2]{Lemmas~\ref{lem:#1} and~\ref{lem:#2}}

\newcommand{\lemrefsss}[4]{Lemmas~\ref{lem:#1}, \ref{lem:#2}, \ref{lem:#3}, and~\ref{lem:#4}}

\newcommand{\obsref}[1]{Observation~\ref{obs:#1}}
\newcommand{\secref}[1]{Section~\ref{sec:#1}}

\newcommand{\assumpref}[1]{Assumption~\ref{assump:#1}}

\newcommand{\itemref}[1]{\ref{itm:#1}}
\newcommand{\itemrefs}[2]{\ref{itm:#1} and~\ref{itm:#2}}

\newtheorem{thm}{Theorem}[section]
\newtheorem{lem}[thm]{Lemma}

\newtheorem{obs}[thm]{Observation}

\newtheorem{defn}[thm]{Definition}

\newtheorem{assump}[thm]{Assumption}

\newtheorem{hypo}[thm]{Hypothesis}

\begin{document}
\global\long\def\SC{\mathbf{\Pi}}

\global\long\def\CG{\mathrm{CG}}

\global\long\def\VG{\mathrm{VG}}
\global\long\def\NVG{\mathrm{NVG}}

\global\long\def\CandA{\mathrm{GA}}

\global\long\def\zeroes{\mathrm{zeroes}}

\global\long\def\guard{\mathrm{G}}

\global\long\def\dist{\mathrm{dist}}

\global\long\def\Vinner{V_{\mathrm{in}}}

\global\long\def\Vouter{V_{\mathrm{out}}}

\global\long\def\Xz{X^{\mathbf{0}}}

\global\long\def\Xo{X^{\mathbf{1}}}

\global\long\def\Yz{Y^{\mathbf{0}}}

\global\long\def\Yo{Y^{\mathbf{1}}}

\global\long\def\Cz{C^{\mathbf{0}}}

\global\long\def\Co{C^{\mathbf{1}}}

\global\long\def\Gz{G^{\mathbf{0}}}

\global\long\def\Go{G^{\mathbf{1}}}

\global\long\def\inputs{{\cal I}}

\global\long\def\algn{{\mathbf{\Lambda}}}
\global\long\def\algnMult{{\algn}^{\mathrm{multi}}}
\global\long\def\algnTwo{{\algn}^{\mathrm{1,2}}}

\global\long\def\type{\mathrm{type}}

\global\long\def\poly{\mathrm{poly}}

\global\long\def\bigOh{{\cal O}}

\global\long\def\zleft{\mathbf{0}_{\x}}

\global\long\def\zright{\mathbf{0}_{\y}}

\global\long\def\oleft{\mathbf{1}_{\x}}

\global\long\def\oright{\mathbf{1}_{\y}}

\global\long\def\dtw{\DTW}
\global\long\def\cost{d}

\newcommand{\px}{\tilde{x}}
\newcommand{\py}{\tilde{y}}
\newcommand{\pw}{\tilde{w}}
\newcommand{\pz}{\tilde{z}}

\newcommand{\barx}{\bar{x}}
\newcommand{\bary}{\bar{y}}

\global\long\def\LCS{\textup{LCS}\xspace}
\global\long\def\OV{\textsc{OV}\xspace}
\global\long\def\UOV{\textup{UOV}\xspace}
\global\long\def\SETH{\textup{SETH}\xspace}
\global\long\def\kSAT{\textup{$k$-SAT}}

\global\long\def\Oh{\bigOh}
\global\long\def\tOh{\tilde{\bigOh}}

\newcommand{\sugg}[2]{{\color{green} #1} {\color{red}\tiny  #2}}

\newcommand{\dopara}[1]{}

\newcommand{\Size}{0.42cm}
\newcommand{\tSize}{0.40cm}
\tikzset{xSquare/.style={
    inner sep=0pt,
    text width=\tSize, 
    minimum size=\Size,
    }
}
\tikzset{NormSquare/.style={
    inner sep=0pt,
    text width=\tSize, 
    minimum size=\Size,
    align=center}
}
\tikzset{MatchSquare/.style={
    inner sep=0pt,
    text width=\tSize, 
    minimum size=\Size,
    align=center,
    fill=blue!50,
    draw=black,
    }
}
\tikzset{DomSquare/.style={
    inner sep=0pt,
    text width=\tSize, 
    minimum size=\Size,
    align=center,
    fill=orange, semithick,
    draw=black}
}

\newcommand{\LV}[1]{#1} \newcommand{\SV}[1]{}

\newcommand{\hideimages}[1]{#1}

\newcommand{\hideproofs}[1]{#1}

\date{}

\title{Multivariate Fine-Grained Complexity \\ of Longest Common Subsequence\thanks{Part of this work was done while the authors were visiting the Simons Institute for the Theory of Computing at the University of California, Berkeley.}
}
\author{Karl Bringmann\footnote{Max Planck Institute for Informatics, Saarland Informatics Campus, Saarbr\"ucken, \texttt{kbringma@mpi-inf.mpg.de} 
} \and Marvin K\"unnemann\thanks{Max Planck Institute for Informatics, Saarland Informatics Campus, Saarbr\"ucken, \texttt{marvin@mpi-inf.mpg.de}}}
\maketitle

\medskip

\begin{abstract}
We revisit the classic combinatorial pattern matching problem of finding a longest common subsequence (LCS). For strings $x$ and~$y$ of length $n$, a textbook algorithm solves LCS in time $O(n^2)$, but although much effort has been spent, no $\Oh(n^{2-\varepsilon})$-time algorithm is known. Recent work indeed shows that such an algorithm would refute the Strong Exponential Time Hypothesis (SETH) [Abboud, Backurs, Vassilevska Williams FOCS'15; Bringmann, K\"unnemann FOCS'15].

Despite the quadratic-time barrier, for over 40 years an enduring scientific interest continued to produce fast algorithms for LCS and its variations. Particular attention was put into identifying and exploiting input parameters that yield strongly subquadratic time algorithms for special cases of interest, e.g., differential file comparison. This line of research was successfully pursued until 1990, at which time significant improvements came to a halt.  
In this paper, using the lens of fine-grained complexity, our goal is to (1) justify the lack of further improvements and (2) determine whether some special cases of LCS admit faster algorithms than currently known.

To this end, we provide a systematic study of the \emph{multivariate complexity} of LCS, taking into account all parameters previously discussed in the literature: the input size $n:=\max\{|x|,|y|\}$, the length of the shorter string $m:=\min\{|x|,|y|\}$, the length $L$ of an LCS of $x$ and $y$, the numbers of deletions $\delta := m-L$ and $\Delta := n-L$, the alphabet size, as well as the numbers of matching pairs $M$ and dominant pairs $d$.
For any class of instances defined by fixing each parameter individually to a polynomial in terms of the input size, we prove a SETH-based lower bound matching one of three known algorithms (up to lower order factors of the form~$n^{o(1)}$). Specifically, we determine the optimal running time for LCS under SETH as $(n+\min\{d, \delta \Delta, \delta m\})^{1\pm o(1)}$. Polynomial improvements over this running time must necessarily refute SETH or exploit novel input parameters. We establish the same lower bound for any constant alphabet of size at least~3.
For binary alphabet, we show a SETH-based lower bound of $(n+\min\{d, \delta \Delta, \delta M/n\})^{1-o(1)}$ and, motivated by difficulties to improve this lower bound, we design an $\Oh(n + \delta M/n)$-time algorithm, yielding again a matching bound.

We feel that our systematic approach yields a comprehensive perspective on the well-studied multivariate complexity of LCS, and we hope to inspire similar studies of multivariate complexity landscapes for further polynomial-time problems.
\end{abstract}

\thispagestyle{empty}
\clearpage
\setcounter{page}{1}

\begin{footnotesize}

\tableofcontents
\end{footnotesize}
\thispagestyle{empty}
\clearpage
\setcounter{page}{1}

\section{Introduction}

String comparison is one of the central tasks in combinatorial pattern matching, with various applications such as spelling correction~\cite{Morgan70,WagnerF74}, DNA sequence comparison~\cite{AltschulGMML90}, and differential file comparison~\cite{HuntMcI75,MillerM85}. Perhaps the best-known measure of string similarity is the length of the longest common subsequence (LCS). A textbook dynamic programming algorithm computes the LCS of given strings $x,y$ of length $n$ in time $\Oh(n^2)$, and in the worst case only an improvement by logarithmic factors is known~\cite{MasekP80}. In fact, recent results show that improvements by polynomial factors would refute the Strong Exponential Time Hypothesis (SETH)~\cite{AbboudBVW15,BringmannK15} (see \secref{hypotheses} for a definition). 

Despite the quadratic-time barrier, the literature on LCS has been steadily growing, with a changing focus on different aspects of the problem over time (see \secref{relatedWork} for an overview). Spurred by an interest in practical applications, a particular focus has been the design of LCS algorithms for strings that exhibit certain structural properties. This is most prominently witnessed by the UNIX \texttt{diff} utility, which quickly compares large, similar files by solving an underlying LCS problem. 
A practically satisfying solution to this special case was enabled by theoretical advances exploiting the fact that in such instances the LCS differs from the input strings at only few positions~(e.g., \cite{Myers86, MillerM85}). In fact, since Wagner and Fischer introduced the LCS problem in 1974~\cite{WagnerF74}, identifying and exploiting structural parameters to obtain faster algorithms has been a decades-long effort~\cite{Apostolico86,ApostolicoG87,EppsteinGGI92,Hirschberg77,HuntS77,IliopoulosR09,Myers86,NakatsuKY82,WuMMM90}. 

Parameters that are studied in the literature are, besides the input size $n:=\max\{|x|,|y|\}$, the length $m:=\min\{|x|,|y|\}$ of the shorter string, the size of the alphabet~$\Sigma$ that $x$ and $y$ are defined on, the length $L$ of a longest common subsequence of $x$ and $y$, the number $\Delta = n - L$ of deleted symbols in the longer string, the number $\delta = m - L$ of deleted symbols in the shorter string, the number of \emph{matching} pairs $M$, and the number of \emph{dominant} pairs $d$ (see \secref{defparams} for definitions). Among the fastest currently known algorithms are an $\tOh(n+\delta L)$-algorithm due to Hirschberg~\cite{Hirschberg77}, an $\tOh(n+\delta \Delta)$-algorithm due to Wu, Manbers, Myers, and Miller~\cite{WuMMM90}, and an $\tOh(n+d)$-algorithm due to Apostolico~\cite{Apostolico86} (with log-factor improvements by Eppstein, Galil, Giancarlo, and Italiano~\cite{EppsteinGGI92}). In the remainder, we refer to such algorithms, whose running time is stated in more parameters than just the problem size $n$, as \emph{multivariate algorithms}. See \tabvref{survey} for a non-exhaustive survey containing the asymptotically fastest multivariate LCS algorithms.

The main question we aim to answer in this work is: \emph{Are there significantly faster multivariate LCS algorithms than currently known?} E.g., can ideas underlying the fastest known algorithms be combined to design an algorithm that is much faster than all of them?

\subsection{Our Approach and Informal Results}

We systematically study special cases of LCS that arise from polynomial restrictions of any of the previously studied input parameters. Informally, we define a \emph{parameter setting} (or polynomial restriction of the parameters) as the subset of all LCS instances where each input parameter is individually bound to a polynomial relation with the input size $n$, i.e., for each parameter $p$ we fix a constant $\alpha_p$ and restrict the instances such that $p$ attains a value $\Theta(n^{\alpha_p})$. 
An algorithm for a specific parameter setting of LCS receives as input two strings $x,y$ guaranteed to satisfy the parameter setting and outputs (the length of) an LCS of $x$ and $y$.
We call a parameter setting trivial if it is satisfied by at most a finite number of instances; this happens if the restrictions on different parameters are contradictory. For each non-trivial parameter setting, we construct a family of hard instances via a reduction from satisfiability, thus obtaining a conditional lower bound. This greatly extends the construction of hard instances for the $n^{2-o(1)}$ lower bound~\cite{AbboudBVW15,BringmannK15}.

\paragraph{Results for large alphabets.} 
Since we only consider exact algorithms, any algorithm for LCS takes time $\Omega(n)$. Beyond this trivial bound, for any non-trivial parameter setting we obtain a SETH-based lower bound of
$$ \min\big\{ d, \, \delta \Delta, \, \delta m\big\}^{1-o(1)}. $$
Note that this bound is matched by the known algorithms with running times $\tOh(n + d)$, $\tOh(n + \delta L)$\footnote{Note that $L \le m$. At first sight it might seem as if the $\tOh(n + \delta L)$ algorithm could be faster than our lower bound, however, for $L \ge m/2$ we have $\delta L = \Theta(\delta m)$, which appears in our lower bound, and for $L \le m/2$ we have $\delta = m-L = \Theta(m)$ and thus $\delta L = \Theta(Lm)$ which is $\Omega(d)$ by \LV{\lemref{dUBs}}\SV{Table~\ref{tab:paramRelations}}, and $d$ appears in our lower bound.}, and $\tOh(n + \delta \Delta)$. Thus, our lower bound very well explains the lack of progress since the discovery of these three algorithms (apart from lower-order factors).

\paragraph{Results for constant alphabet size.}
For the alphabet size $|\Sigma|$, we do not only consider the case of a polynomial relation with $n$, but also the important special cases of $|\Sigma|$ being any fixed constant. 
We show that our conditional lower bound for polynomial alphabet size also holds for any constant $|\Sigma| \ge 3$. 
For $|\Sigma| = 2$, we instead obtain a SETH-based lower bound of 
$$ \min\big\{ d, \, \delta \Delta, \, \delta M / n\big\}^{1-o(1)}. $$
This lower bound is weaker than the lower bound for $|\Sigma|\ge 3$ (as the term $\delta M/n$ is at most $\delta m$ by the trivial bound $M\le mn$; see \secref{defparams} for the definition of $M$). Surprisingly, a stronger lower bound is impossible (assuming SETH): Motivated by the difficulties to obtain the same lower bound as for $|\Sigma|\ge 3$, we discovered an algorithm with running time $\Oh(n + \delta M /n)$ for $|\Sigma|=2$, thus matching our conditional lower bound. To the best of our knowledge, this algorithm provides the first polynomial improvement for a special case of LCS since 1990, so while its practical relevance is unclear, we succeeded in uncovering a tractable special case. Interestingly, our algorithm and lower bounds show that the multivariate fine-grained complexity of LCS differs polynomially between $|\Sigma| = 2$ and $|\Sigma| \ge 3$. So far, the running time of the fastest known algorithms for varying alphabet size differed at most by a logarithmic factor in $|\Sigma|$.

We find it surprising that the hardness assumption SETH is not only sufficient to prove a worst-case quadratic lower bound for LCS, but extends to the \emph{complete spectrum} of multivariate algorithms using the previously used 7 parameters, thus proving an optimal running time bound which was implicitly discovered by the computer science community within the first 25 years of research on LCS (except for the case of $\Sigma = \{0,1\}$, for which we provide a missing algorithm).

\subsection{Related Work on LCS}
\label{sec:relatedWork}

\hideimages{
\begin{savenotes}
\begin{table}
\begin{tabular}{ll}
\textbf{Reference} & \textbf{Running Time}\\
\hline
Wagner and Fischer~\cite{WagnerF74} & $\Oh(mn)$ \\
Hunt and Szymanski~\cite{HuntS77} & $\Oh((n+M)\log n)$ \\
Hirschberg~\cite{Hirschberg77} & $\Oh(n\log n + Ln)$ \\
Hirschberg~\cite{Hirschberg77} & $\Oh(n\log n + L\delta \log n)$ \\
Masek and Paterson~\cite{MasekP80} & $\Oh(n + nm/\log^2 n)$ \, assuming $|\Sigma|= \Oh(1)$  \\
& $\Oh\big(n + nm \cdot \big(\frac{\log \log n}{\log n}\big)^2\big)$ \footnote{See~\cite{BilleFC08} for how to extend the Masek-Paterson algorithm to non-constant alphabets.} \\
Nakatsu, Kambayashi and Yajima~\cite{NakatsuKY82} & $\Oh(n\delta)$
 \\ 
Apostolico~\cite{Apostolico86} & $\Oh(n\log n + d\log(mn/d))$ \\
Myers~\cite{Myers86} & $\Oh(n\log n+\Delta^2)$ \\
Apostolico and Guerra~\cite{ApostolicoG87} & $\Oh(n\log n + Lm\min\{\log m, \log(n/m)\}) $ \\
Wu, Manbers, Myers and Miller~\cite{WuMMM90} &  $\Oh(n\log n+\delta \Delta)$ \footnote{Wu et al.\ state their running time as $\Oh(n\delta)$ in the worst case and $\Oh(n + \delta \Delta)$ in expectation for random strings. However, Myers worst-case variation trick~\cite[Section 4c]{Myers86} applies and yields the claimed time bound $\Oh(n \log n + \delta \Delta)$. The additional $\Oh(n \log n)$ comes from building a suffix tree.} \\
Eppstein, Galil, Giancarlo and Italiano~\cite{EppsteinGGI92} &  $\Oh(n\log n+d\log \log \min\{d, nm/d\} )$ \\
Iliopoulos and Rahman~\cite{IliopoulosR09} & $\Oh(n + M \log \log n)$ \\
\end{tabular}
\caption{Short survey of LCS algorithms. See \secref{defparams} for definitions of the parameters. When stating the running times, every factor possibly attaining non-positive values (such as $\delta,\log(n/m)$, etc.) is to be read as $\max\{\cdot,1\}$. For simplicity, $\log(\Sigma)$-factors have been bounded from above by $\log n$ (see \cite{PatersonD94} for details on the case of constant alphabet size). 
}
\label{tab:survey}
\end{table}
\end{savenotes}
}

\tabvref{survey} gives a non-comprehensive overview of progress on multivariate LCS, including the asymptotically fastest known algorithms. Note that the most recent polynomial factor improvement for multivariate LCS was found in 1990~\cite{WuMMM90}. Further progress on multivariate LCS was confined to log-factor improvements (e.g.,~\cite{EppsteinGGI92,IliopoulosR09}). Therefore, the majority of later works on LCS focused on transferring the early successes and techniques to more complicated problems, such as 
longest common increasing subsequence
\cite{moosa2013computing,kutz2011faster,yang2005fast,chan2007efficient},
tree LCS
\cite{mozes2009fast},
and many more generalizations and variants of LCS, see, e.g.,
\cite{kuboi2017,castelli2017,tiskin2006longest,jiang2004longest,
alber2002towards,landau2003sparse,keller2009longest,gotthilf2010restricted,
gotthilf2007approximating,pevzner1992matrix,iliopoulos2007algorithms,
iliopoulos2008algorithms,alam2012substring,farhana2012doubly,lemstrom2004bit,
wang1993set,landau2007two,amir2010weighted,amir2008generalized,tiskin2008semi,
blin2012parameterized,deorowicz2012quadratic,becerra2016multiobjective,tiskin2010fast,
benson2016lcsk,jacobson1992heaviest}.
One branch of generalizations considered the LCS of more than two strings
(e.g.,~\cite{blin2012hardness,AbboudBVW15}), with variations such as string consensus
(e.g.,~\cite{amir2013hardness,amir2016configurations})
and more
(e.g.,~\cite{cygan2016polynomial,landau2003sparse,bansal2010longest,gotthilf2010restricted,
gotthilf2008constrained,amir2010weighted}). 
Since natural language texts are well compressible,  
researchers also considered solving LCS directly on compressed strings, using either 
run-length encoding (e.g.,~\cite{kuboi2017,bunke1995improved,apostolico1997matching,crochemore2003subquadratic})
or straight-line programs and other Lempel-Ziv-like compression schemes
(e.g.,~\cite{tiskin2009faster,HLLW13,gawrychowski2012faster,lifshits2007processing}).
Further research directions include approximation algorithms for LCS and its variants (e.g.,~\cite{landau2009lcs,gotthilf2007approximating,gotthilf2008constrained}),
as well as the LCS length of random strings 
\cite{baeza1999bounding,lueker2009improved}.
For brevity, here we ignore the equally vast literature on the closely related edit distance. 
Furthermore, we solely regard the time complexity of computing the length of an LCS and hence omit all results concerning space usage or finding an LCS. See, e.g.,~\cite{PatersonD94,BergrothHR00} for these and other aspects of LCS (including empirical evaluations).

\subsection{(Multivariate) Hardness in \textsf{P}}

After the early success of 3SUM-hardness in computational geometry~\cite{AnkaOvermars}, recent years have brought a wealth of novel conditional lower bounds for polynomial time problems, see, e.g., \cite{AbboudVW14,RodittyVW13,Williams05,VassilevskaWW10,BackursI15,abboud2016approximation, KopelowitzPP16,AbboudHVWW16,
Bringmann14,BringmannM15,AbboudBVW15,BringmannK15,bringmann2016dichotomy,backurs2016regular} and the recent survey~\cite{VassilevskaWilliams15}. In particular, our work extends the recent successful line of research proving SETH-based lower bounds for a number problems with efficient dynamic programming solutions such as Fr\'echet distance~\cite{Bringmann14,BringmannM15}, edit distance~\cite{BackursI15,BringmannK15}, LCS and dynamic time warping~\cite{AbboudBVW15,BringmannK15}. 
Beyond worst-case conditional lower bounds of the form $n^{c-o(1)}$, recently also more detailed lower bounds targeting additional input restrictions have gained interest. Such results come in different flavors, as follows.

\emph{Input parameters, polynomial dependence.}
Consider one or more parameters in addition to the input size $n$, where the optimal time complexity of the studied problem depends polynomially on $n$ and the parameters. This is the situation in this paper as well as several previous studies, e.g.,~\cite{abboud2016approximation, KopelowitzPP16, Bringmann14}. To the best of our knowledge, our work is the first of this kind to study combinations of more than two parameters that adhere to a complex set of parameter relations -- for previous results, typically the set of non-trivial parameter settings was obvious and simultaneously controlling all parameters was less complex.

\emph{Input parameters, superpolynomial dependence.} 
Related to the above setting, parameters have been studied where the time complexity depends polynomially on $n$ and \emph{superpolynomially on the parameters}. If the studied problem is \textsf{NP}-hard then this is known as fixed-parameter tractability (\textsf{FPT}). However, here we focus on problems in \textsf{P}, in which case this situation is known as ``\textsf{FPT} in \textsf{P}''. Hardness results in this area were initiated by Abboud, Vassilevska Williams, and Wang~\cite{abboud2016approximation}.

\emph{A finite/discrete number of special cases.} Some input restrictions yield a discrete or even finite set of special cases. For example, Backurs and Indyk~\cite{backurs2016regular} and later Bringmann et al.~\cite{bringmann2016dichotomy} studied special cases of regular expression pattern matching by restricting the input to certain ``types'' of regular expressions. The set of types is discrete and infinite, however, there are only finitely many tractable types, and finitely many minimal hardness results. Their approach is similarly systematic to ours, as they classify the complexity of pattern matching for any type of regular expressions. The major difference is that our parameters are ``continuous'', specifically our parameter exponents~$\alpha_p$ are continuous, and thus our algorithms and lower bounds trace a continuous tradeoff.

While in all of the above settings the design of fast multivariate algorithms is well established, 
tools for proving matching conditional lower bounds have been developed only recently. 
In particular, the systematic approach to multivariate lower bounds pursued in this paper provides an effective complement to multivariate algorithmic studies in \textsf{P}, 
since it establishes (near-)optimality and may uncover tractable special cases for which improved algorithms can be found.%

\paragraph{Beyond SETH.}
Motivated in part to find barriers even for polylogarithmic improvements on LCS, a surprising result of Abboud et al.~\cite{AbboudHVWW16} strengthens the conditional quadratic-time hardness of LCS substantially. More precisely, they show that a strongly subquadratic-time algorithm for LCS would even refute a natural, weaker variant of SETH on \emph{branching programs}. In \secref{bpseth}, we survey their result and show that the conditional lower bounds we derive in this paper also hold under this weaker assumption.

\section{Preliminaries}

We write $[n]:= \{1,\dots,n\}$. For a string $x$, we denote its length by $|x|$, the symbol at its $i$-th position by $x[i]$, and the substring from position $i$ to position $j$ by $x[i..j]$.
If string $x$ is defined over alphabet~$\Sigma$, we denote the number of occurrences of symbol $\sigma\in \Sigma$ in $x$ by $\occ_\sigma(x)$. In running time bounds we write $\Sigma$ instead of $|\Sigma|$ for readability. For two strings $x,y$, we denote their concatenation by $x \concat y = xy$ and define, for any $\ell \ge 0$, the $\ell$-fold repetition $x^\ell := \bigconcat_{i=1}^\ell x$.
For any strings $x,y$ we let $\LCS(x,y)$ be any longest common subsequence of $x$ and $y$, i.e., a string $z = z[1..L]$ of maximum length $L$ such that there are $i_1 < \ldots < i_L$ with $x[i_k] = z[k]$ for all $1 \le k \le L$ and there are $j_1 < \ldots < j_L$ with $y[j_k] = z[k]$ for all $1 \le k \le L$. 
For a string $x$ of length $n$, let $\rev(x) := x[n] \, x[n-1] \dots x[1]$ denote its reverse.

\subsection{Parameter Definitions}
\label{sec:defparams}

We survey parameters that have been used in the analysis of the LCS problem (see also~\cite{PatersonD94,BergrothHR00}). Let $x,y$ be any strings. By possibly swapping $x$ and $y$, we can assume that $x$ is the longer of the two strings, so that $n = n(x,y) := |x|$ is  the input size (up to a factor of two). Then $m = m(x,y) := |y|$ is the length of the shorter of the two strings.
Another natural parameter is the solution size, i.e., the length of any LCS, $L = L(x,y):=|\LCS(x,y)|$. 

Since any symbol not contained in $x$ or in $y$ cannot be contained in a LCS, we can ensure the following using a (near-)linear-time preprocessing. 

\begin{assump}\label{assump:symbols}
Every symbol $\sigma \in \Sigma$ occurs at least once in $x$ and in $y$, i.e., $\occ_\sigma(x),\occ_\sigma(y) \ge 1$. 
\end{assump}

Consider the alphabet induced by $x$ and $y$ after ensuring \assumpref{symbols}, namely $\Sigma = \{x[i] \mid 1 \le i \le |x|\} \cap \{y[j] \mid 1 \le j \le |y|\}$. Its size $\Sigma(x,y) := |\Sigma|$ is a natural parameter. 

Beyond these standard parameters $n,m,L,|\Sigma|$ (applicable for any optimization problem on strings), popular structural parameters measure the \emph{similarity} and \emph{sparsity} of the strings. These notions are more specific to LCS and are especially relevant in practical applications such as, e.g., the \texttt{diff} file comparison utility, where symbols in $x$ and $y$ correspond to lines in the input files. 

\hideimages{

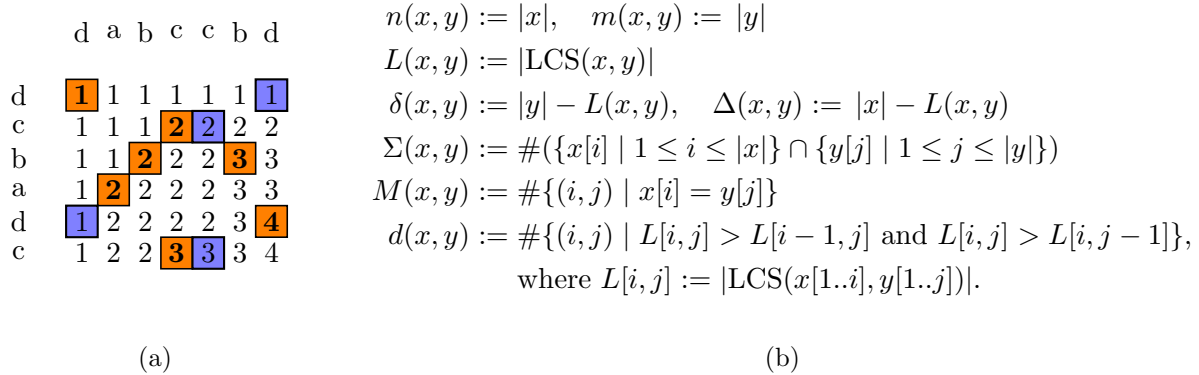
\begin{figure}
\begin{subfigure}[b]{0.25\textwidth}
  \centering
  \begin{tikzpicture}[draw=black, thick,x=\Size,y=\Size]
\node [NormSquare] at (2.5,-0.5) {d};
\node [NormSquare] at (3.5,-0.5) {a};
\node [NormSquare] at (4.5,-0.5) {b};
\node [NormSquare] at (5.5,-0.5) {c};
\node [NormSquare] at (6.5,-0.5) {c};
\node [NormSquare] at (7.5,-0.5) {b};
\node [NormSquare] at (8.5,-0.5) {d};
\node [NormSquare] at (0.5,-2.5) {d};
\node [NormSquare] at (0.5,-3.5) {c};
\node [NormSquare] at (0.5,-4.5) {b};
\node [NormSquare] at (0.5,-5.5) {a};
\node [NormSquare] at (0.5,-6.5) {d};
\node [NormSquare] at (0.5,-7.5) {c};
\node [DomSquare] at (2.5,-2.5) {\textbf{1}};
\node [NormSquare] at (3.5,-2.5) {1};
\node [NormSquare] at (4.5,-2.5) {1};
\node [NormSquare] at (5.5,-2.5) {1};
\node [NormSquare] at (6.5,-2.5) {1};
\node [NormSquare] at (7.5,-2.5) {1};
\node [MatchSquare] at (8.5,-2.5) {1};
\node [NormSquare] at (2.5,-3.5) {1};
\node [NormSquare] at (3.5,-3.5) {1};
\node [NormSquare] at (4.5,-3.5) {1};
\node [DomSquare] at (5.5,-3.5) {\textbf{2}};
\node [MatchSquare] at (6.5,-3.5) {2};
\node [NormSquare] at (7.5,-3.5) {2};
\node [NormSquare] at (8.5,-3.5) {2};
\node [NormSquare] at (2.5,-4.5) {1};
\node [NormSquare] at (3.5,-4.5) {1};
\node [DomSquare] at (4.5,-4.5) {\textbf{2}};
\node [NormSquare] at (5.5,-4.5) {2};
\node [NormSquare] at (6.5,-4.5) {2};
\node [DomSquare] at (7.5,-4.5) {\textbf{3}};
\node [NormSquare] at (8.5,-4.5) {3};
\node [NormSquare] at (2.5,-5.5) {1};
\node [DomSquare] at (3.5,-5.5) {\textbf{2}};
\node [NormSquare] at (4.5,-5.5) {2};
\node [NormSquare] at (5.5,-5.5) {2};
\node [NormSquare] at (6.5,-5.5) {2};
\node [NormSquare] at (7.5,-5.5) {3};
\node [NormSquare] at (8.5,-5.5) {3};
\node [MatchSquare] at (2.5,-6.5) {1};
\node [NormSquare] at (3.5,-6.5) {2};
\node [NormSquare] at (4.5,-6.5) {2};
\node [NormSquare] at (5.5,-6.5) {2};
\node [NormSquare] at (6.5,-6.5) {2};
\node [NormSquare] at (7.5,-6.5) {3};
\node [DomSquare] at (8.5,-6.5) {\textbf{4}};
\node [NormSquare] at (2.5,-7.5) {1};
\node [NormSquare] at (3.5,-7.5) {2};
\node [NormSquare] at (4.5,-7.5) {2};
\node [DomSquare] at (5.5,-7.5) {\textbf{3}};
\node [MatchSquare] at (6.5,-7.5) {3};
\node [NormSquare] at (7.5,-7.5) {3};
\node [NormSquare] at (8.5,-7.5) {4};
\end{tikzpicture}
\vspace{2em}
  \caption{}
  \label{fig:classicLCSpic}
  \end{subfigure}
  \hfill
  \begin{subfigure}[b]{0.74\textwidth}
    \centering
\begin{align*}
n(x,y) :=\; & |x|, \quad m(x,y) :=\, |y|\\
L(x,y) :=\; &|\LCS(x,y)| \\
\delta(x,y) :=\; &|y|-L(x,y), \quad \Delta(x,y):=\, |x|-L(x,y)\\
\Sigma(x,y) :=\; & \#(\{x[i] \mid 1 \le i \le |x|\} \cap \{y[j] \mid 1 \le j \le |y|\}) \\
M(x,y) :=\; &\#\{ (i,j)  \mid x[i]  = y[j] \}\\
d(x,y) :=\; &\#\{ (i,j) \mid L[i,j] > L[i-1,j] \text{ and } L[i,j] > L[i,j-1]\},\\
 & \text{where } L[i,j] := |\LCS(x[1..i],y[1..j])|.
\end{align*}
\caption{}
\label{fig:paramSummary}
    \end{subfigure}
\caption{(a) Illustration of the $L$-table, matching pairs and dominant pairs. Entries marked in orange color and bold letters correspond to dominant pairs (which by definition are also matching pairs), while entries marked in blue are matching pairs only. (b) Summary of all input parameters.}
\end{figure}
}

\paragraph{Notions of similarity.}
To obtain an LCS, we have to delete $\Delta = \Delta(x,y) := n-L$ symbols from $x$ or $\delta = \delta(x,y) := m - L$ symbols from $y$. Hence for very similar strings, which is the typical kind of input for file comparisons, we expect $\delta$ and $\Delta$ to be small. This is exploited by algorithms running in time, e.g., $\tOh(n + \delta \Delta)$~\cite{WuMMM90} or $\tOh(n + \delta L)$~\cite{Hirschberg77}.

\paragraph{Notions of sparsity.}
Based on the observation that the dynamic programming table typically stores a large amount of redundant information (suggested, e.g., by the fact that an LCS itself can be reconstructed examining only $O(n)$ entries), algorithms have been studied that consider only the most relevant entries in the table. The simplest measure of such entries is the number of  \emph{matching pairs} $M = M(x,y) := \#\{ (i,j) \mid x[i] = y[j] \}$. Especially for inputs with a large alphabet, this parameter potentially significantly restricts the number of candidate pairs considered by LCS algorithms, e.g., for files where almost all lines occur only once. Moreover, in the special case where $x$ and $y$ are permutations of $\Sigma$ we have $M=n=m$, and thus algorithms in time $\tOh(n+M)$~\cite{HuntMcI75,HuntS77,IliopoulosR09} recover the near-linear time solution for LCS of permutations~\cite{szymanski1975special}.

One can refine this notion to obtain the \emph{dominant pairs}. A pair $(i,j)$ \emph{dominates} a pair $(i',j')$ if we have $i \le i'$ and $j \le j'$.
A \emph{$k$-dominant pair} is a pair $(i,j)$ such that $L(x[1..i],y[1..j]) = k$ and no other pair $(i',j')$ with $L(x[1..i'],y[1..j'])=k$ dominates $(i,j)$. By defining $L[i,j] := L(x[1..i],y[1..j])$ and using the well known recursion $L[i,j] = \max\{L[i-1,j],L[i,j-1],L[i-1,j-1] + \mathbf{1}_{x[i] = y[j]}\}$, we observe that $(i,j)$ is a $k$-dominant pair if and only if $L[i,j] = k$ and $L[i-1,j]=L[i,j-1]=k-1$. Denoting the set of all $k$-dominant pairs by $D_k$, the set of \emph{dominant pairs} of $x,y$ is $\bigcup_{k\ge 1} D_k$, and we let $d = d(x,y)$ denote the number of dominant pairs. Algorithms running in time $\tOh(n + d)$ exploit a small number of dominant pairs~\cite{Apostolico86,EppsteinGGI92}. \figref{classicLCSpic} illustrates matching and dominant pairs. 

While at first sight the definition of dominant pairs might not seem like the most natural parameter, it plays an important role in analyzing LCS: First, from the set of dominant pairs alone one can reconstruct the $L$-table that underlies the basic dynamic programming algorithm. Second, the parameter $d$ precisely describes the complexity of one of the fastest known (multivariate) algorithms for LCS. Finally, LCS with parameter $d$ is one of the first instances of the paradigm of \emph{sparse dynamic programming} (see, e.g., \cite{EppsteinGGI92}).

\medskip
On practical instances, exploiting similarity notions seems to typically outperform algorithms based on sparsity measures (see~\cite{MillerM85} for a classical comparison to an algorithm based on the number of matching pairs $M$~\cite{HuntMcI75,HuntS77}). 
To the best of our knowledge, \figref{paramSummary} summarizes all parameters which have been exploited to obtain multivariate algorithms for LCS.

\subsection{Hardness Hypotheses}
\label{sec:hypotheses}

\noindent
\textbf{%
Strong Exponential Time Hypothesis (\SETH):} \emph{
For any $\varepsilon > 0$ there is a $k \ge 3$ such that $k$-SAT on $n$ variables cannot be solved in time $\Oh((2-\varepsilon)^n)$.
}

\medskip
\SETH was introduced by Impagliazzo, Paturi, and Zane~\cite{ImpagliazzoPZ01} and essentially asserts that satisfiability has no algorithms that are much faster than exhaustive search. It forms the basis of many conditional lower bounds for NP-hard as well as polynomial-time problems.

Effectively all known SETH-based lower bounds for polynomial-time problems use reductions via the \emph{Orthogonal Vectors problem} (\OV):
Given sets $\sA$, $\sB\subseteq \{0,1\}^D$ of size $|\sA| = n$, $|\sB| = m$, determine whether there exist $a \in \sA$, $b \in \sB$ with $\sum_{i=1}^D a[i] \cdot b[i] = 0$ (which we denote by $\langle a, b \rangle =0$).
Simple algorithms solve \OV in time $\Oh(2^D (n+m))$ and $\Oh(nm D)$. The fastest known algorithm for $D=c(n)\log n$ runs in time $n^{2-1/\Oh(\log c(n))}$ (when $n=m$) \cite{AbboudWY15}, which is only slightly subquadratic for $D \gg \log n$. This has led to the following reasonable hypothesis.

\medskip
\noindent
\textbf{%
Orthogonal Vectors Hypothesis (\OVH):} \emph{
\OV restricted to $n=|\sA|=|\sB|$ and $D=n^{o(1)}$ requires time $n^{2-o(1)}$.
} 

\medskip
A well-known reduction by Williams~\cite{Williams05} shows that \SETH implies \OVH. Thus, \OVH is the weaker assumption and any \OVH-based lower bound also implies a \SETH-based lower bound. The results in this paper do not only hold assuming \SETH, but even assuming the weaker \OVH.

\section{Formal Statement of Results}
\label{sec:formal}

Recall that $n$ is the input size and $\params:= \{m, L, \delta, \Delta, |\Sigma|, M, d\}$ is the set of parameters that were previously studied in the literature. We let $\params^* := \params \cup \{n\}$. A parameter setting fixes a polynomial relation between any parameter and $n$. To formalize this, we call a vector $\Valpha = (\alpha_p)_{p\in \params}$ with $\alpha_p \in \R_{\ge 0}$ a \emph{parameter setting}, and an LCS instance $x,y$ \emph{satisfies} the parameter setting $\Valpha$ if each parameter $p$ attains a value $p(x,y) = \Theta(n^{\alpha_p})$. This yields a subproblem of LCS consisting of all instances that satisfy the parameter setting. 
We sometimes use the notation $\alpha_n = 1$.

For our running time bounds, for each parameter $p \in \params$ except for $|\Sigma|$ we can assume $\alpha_p > 0$, since otherwise one of the known algorithms runs in time $\tOh(n)$ and there is nothing to show. Similarly, for $\alpha_d \le 1$ there is an $\tOh(n)$ algorithm and there is nothing to show.
For $\Sigma$, however, the case $\alpha_\Sigma = 0$, i.e., $|\Sigma| = \Theta(1)$, is an important special case. We study this case more closely by also considering parameter settings that fix $|\Sigma|$ to any specific constant greater than 1. 

\begin{defn}[Parameter Setting]
Fix $\gamma \ge 1$. Let $\Valpha = (\alpha_p)_{p\in \params}$ with $\alpha_p \in \R_{\ge 0}$. We define $\LCS^\gamma(\Valpha)$ as the problem of computing the length of an LCS of two given strings $x,y$ satisfying $n^{\alpha_p}/\gamma \le p(x,y) \le n^{\alpha_p} \cdot \gamma$ for every parameter $p \in \params$, where $n = |x|$, and $|x| \ge |y|$. We call $\Valpha$ and $\LCS^\gamma(\Valpha)$ \emph{parameter settings}. In some statements we simply write $\LCS(\Valpha)$ to abbreviate that there exists a $\gamma \ge 1$ such that the statement holds for $\LCS^\gamma(\Valpha)$. 

For any fixed alphabet $\Sigma$, constant $\gamma \ge 1$, and parameter setting $\Valpha$ with $\alpha_\Sigma = 0$, we also define the problem $\LCS^\gamma(\Valpha,\Sigma)$, where additionally the alphabet of $x,y$ is fixed to be $\Sigma$. We again call $(\Valpha,\Sigma)$ and $\LCS^\gamma(\Valpha,\Sigma)$ parameter settings.

We call a parameter setting $\Valpha$ or $(\Valpha,\Sigma)$ \emph{trivial} if for all $\gamma \ge 1$ the problem $\LCS^\gamma(\Valpha)$ or $\LCS^\gamma(\Valpha,\Sigma)$, respectively, has only finitely many instances. 
\end{defn}

As our goal is to prove hardness for any non-trivial parameter setting, for each parameter setting we either need to construct hard instances or verify that it is trivial. That is, in one way or the other we need a complete \emph{classification} of parameter settings into trivial and non-trivial ones. To this end, we need to understand all interactions among our parameters that hold up to constant factors, which is an interesting question on its own, as it yields insight into the structure of strings from the perspective of the LCS problem.
For our seven parameters, determining all interactions is a complex task. This is one of the major differences to previous multivariate fine-grained complexity results, where the number of parameters was one, or in rare cases two, limiting the interaction among parameters to a simple level.

\begin{thm}[Classification of non-trivial parameter settings] \label{thm:nontrivialclassification}
  A parameter setting $\Valpha$ or $(\Valpha,\Sigma)$ is non-trivial if and only if it satisfies all restrictions in \tabref{paramChoiceRestr}.
\end{thm}

\hideimages{

\begin{table}
\centering
\begin{tabular}{ll}
\textbf{Parameter} & \textbf{Restriction}\\
\hline
$m$ & $0 \le \alpha_m \le 1$ \\
\hline
$L$ & $0 \le \alpha_L \le \alpha_m$  \\
\hline
 $\delta$ & $\begin{cases} 0 \le \alpha_{\delta} \le \alpha_m & \text{if } \alpha_L = \alpha_m \\ \alpha_{\delta} = \alpha_m & \text{otherwise} \end{cases}$ \\
\hline
 $\Delta$ & $\begin{cases} \alpha_{\delta} \le \alpha_{\Delta} \le 1 & \text{if } \alpha_L = \alpha_m = 1 \\ \alpha_{\Delta}= 1 & \text{otherwise} \end{cases}$  \\
\hline
$|\Sigma|$ & $0 \le \alpha_{\Sigma} \le \alpha_m$  \\
\hline
 $d$ & $\max\{\alpha_L, \alpha_\Sigma\} \le \alpha_d \le \min\{2\alpha_L+\alpha_{\Sigma},\alpha_L + \alpha_m, \alpha_L + \alpha_\Delta\}$  \\
\hline
 $M$ & $\max\{1,\alpha_d,2\alpha_L-\alpha_\Sigma\} \le \alpha_M \le \alpha_L + 1$ \\
 & if $|\Sigma| = 2$: $\alpha_M \ge \max\{\alpha_L+\alpha_m, 1+\alpha_d - \alpha_L\}$ \\
 & if $|\Sigma| = 3$: $\alpha_M \ge \alpha_m+\alpha_d - \alpha_L$ \\
\hline
\end{tabular}
\caption{Complete set of restrictions for non-trivial parameter settings.}
\label{tab:paramChoiceRestr}
\end{table}

}

Note that the restrictions in \tabref{paramChoiceRestr} consist mostly of linear inequalities, and that for small alphabet sizes $|\Sigma| \in \{2,3\}$ additional parameter relations hold.
The proof of this and the following results will be outlined in \secref{roughproof}.
We are now ready to state our main lower bound.

\begin{thm}[Hardness for Large Alphabet]\label{thm:main1}
For any non-trivial parameter setting $\Valpha$, there is a constant $\gamma \ge 1$ such that $\LCS^\gamma(\Valpha)$ requires time
$\min\big\{ d, \, \delta \Delta, \, \delta m \big\}^{1-o(1)}$,
unless \OVH fails.
\end{thm}

In the case of constant alphabet size, the (conditional) complexity differs between $|\Sigma|=2$ and $|\Sigma|\ge 3$. Note that $|\Sigma|=1$ makes LCS trivial.
\begin{thm}[Hardness for Small Alphabet]\label{thm:main2}
For any non-trivial parameter setting $(\Valpha,\Sigma)$, there is a constant $\gamma \ge 1$ such that, unless \OVH fails, $\LCS^\gamma(\Valpha,\Sigma)$ requires time
\begin{itemize}
  \item $\min\big\{ d, \, \delta \Delta, \, \delta m \big\}^{1-o(1)}$ if $|\Sigma| \ge 3$,
  \item $\min\big\{ d, \, \delta \Delta, \, \delta M / n \big\}^{1-o(1)}$ if $|\Sigma| = 2$.
\end{itemize} 
\end{thm}

Finally, we prove the following algorithmic result, handling binary alphabets faster if $M$ and $\delta$  are sufficiently small. This yields matching upper and lower bounds also for $|\Sigma|=2$.

\begin{thm}[\secref{algo}]\label{thm:algo}
For $|\Sigma| = 2$, LCS can be solved in time $\Oh(n + \delta M / n)$.
\end{thm}

\section{Hardness Proof Overview}
\label{sec:roughproof}

In this section we present an overview of the proofs of our main results. 
We first focus on the large alphabet case, i.e., parameter settings $\Valpha$, and discuss small constant alphabets in \secref{overview_smallalphabet}.

\subsection{Classification of Non-trivial Parameter Settings}

The only-if-direction of \thmref{nontrivialclassification} follows from proving ineqalities among the parameters that hold for all strings, and then converting them to inequalities among the $\alpha_p$'s, as follows. 

\begin{table}
\centering
\begin{tabular}{lll}
\textbf{Relation} & \textbf{Restriction} & \textbf{Reference} \\
\hline
$L \le m \le n$ & & trivial \\
$L \le d\le M$ & & trivial \\
$\Delta \le n$ & & trivial \\
$\delta \le m$ & & trivial \\
$\delta \le \Delta$ & & trivial \\
\hline
$\delta = m-L$ & & by definition \\
$\Delta = n-L$ & & by definition \\
\hline
$|\Sigma| \le m$ & & \assumpref{symbols} \\
$n \le M$ & & \assumpref{symbols} \\
\hline
$d \le Lm$ &  & \lemref{dUBs} \\
$d \le L^2 |\Sigma|$ &  & \lemref{dUBs} \\
$d \le 2L (\Delta + 1)$ &  & \lemref{dLgap}\\
$|\Sigma| \le d$ &  & \lemref{dSigma}\\
$\frac{L^2}{|\Sigma|} \le M \le 2Ln$ & & \lemref{Mbounds} \\
\hline
$M \ge Lm/4$ &  if $|\Sigma| = 2$ & \lemref{MbinaryLm} \\
$M \ge nd/(5 L)$ &  if $|\Sigma| = 2$ & \lemref{Mbinarynd} \\
\hline
$M \ge md/(80 L)$ &  if $|\Sigma| = 3$ & \lemref{Mternarymd}
\end{tabular}
\caption{Relations between the parameters.}
\label{tab:paramRelations}
\end{table}

\begin{lem}[Parameter Relations, \secref{params}] \label{lem:paramsnecessary}
  For any strings $x,y$ the parameter values $\params^*$ satisfy the relations in \tabref{paramRelations}. Thus, any non-trivial parameter setting $\Valpha$ or $(\Valpha,\Sigma)$ satisfies \tabref{paramChoiceRestr}.
\end{lem}
\begin{proof}[Proof Sketch]
  The full proof is deferred to \secref{params}.
  Some parameter relations follow trivially from the parameter definitions, like $L \le m \le n$. Since by \assumpref{symbols} every symbol in $\Sigma$ appears in $x$ and $y$, we obtain parameter relations like $|\Sigma| \le m$. Other parameter relations need a non-trivial proof, like $M \ge m d / (80L)$ if $|\Sigma| = 3$. 
  
  From a relation like $L \le m$ we infer that if $\alpha_L > \alpha_m$ then for sufficiently large $n$ no strings $x,y$ have $L(x,y) = \Theta(n^{\alpha_L})$ and $m(x,y) = \Theta(n^{\alpha_m})$, and thus $\LCS^\gamma(\Valpha)$ is finite for any $\gamma>0$. This argument converts  \tabref{paramRelations} to \tabref{paramChoiceRestr}.
\end{proof}

For the if-direction of \thmref{nontrivialclassification}, the task is to show that any parameter setting satisfying \tabref{paramChoiceRestr} is non-trivial, i.e., to construct infinitely many strings in the parameter setting. We start with a construction that sets a single parameter $p$ as specified by $\Valpha$, and all others not too large.

\begin{lem}[Paddings, \secref{paddings}] \label{lem:paddingstotal}
  Let $\Valpha$ be a parameter setting satisfying \tabref{paramChoiceRestr}.
  For any parameter $p \in \params^*$ and any $n \ge 1$ we can construct strings $x_p,y_p$ such that (1) $p(x_p,y_p) = \Theta(n^{\alpha_p})$, and (2) for all $q \in \params^*$ we have $q(x_p,y_p) = O(n^{\alpha_q})$. 
  Moreover, given $n$ we can compute $x_p$, $y_p$, and $L(x_p,y_p)$ in time $\Oh(n)$. 
\end{lem}

Note that although for \thmref{nontrivialclassification} the \emph{existence} of infinitely many strings would suffice, we even show that they can be \emph{computed} very efficiently. We will use this additional fact in \secref{overview_monotonicity}.

\begin{proof}[Proof Sketch] 
  We defer the full proof to \secref{paddings} and here only sketch the proof for the parameter~$|\Sigma|$. 
    Let $w := 1 2 \ldots t$ be the concatenation of $t := \lceil n^{\alpha_\Sigma} \rceil$ unique symbols. We argue that the strings $w,w$ or the strings $w,\rev(w)$ prove \lemref{paddingstotal} for parameter $p = |\Sigma|$, depending on the parameter setting $\Valpha$. 
  Clearly, both pairs of strings realize an alphabet of size $t = \Theta(n^{\alpha_\Sigma})$, showing~(1). 
  By \tabref{paramChoiceRestr}, we have $\alpha_L = \alpha_m$ or $\alpha_\delta = \alpha_m$. In the first case, we use $L(w,w) = t = O(n^{\alpha_\Sigma})$ together with $\alpha_\Sigma \le \alpha_m = \alpha_L$, as well as $\delta(w,w) = \Delta(w,w) = 0 \le n^{\alpha_\delta} \le n^{\alpha_\Delta}$, to show (2) for the parameters $L,\delta,\Delta$. In the second case, we similarly have $L(w,\rev(w)) = 1 \le n^{\alpha_L}$ and $\delta(w,\rev(w)) = \Delta(w,\rev(w)) = t-1 = O(n^{\alpha_\Sigma})$ and $\alpha_\Sigma \le \alpha_m = \alpha_\delta \le \alpha_\Delta$.
  
  The remaining parameters are straight-forward. Let $(x,y) \in \{(w,w), (w,\rev(w))\}$. We have $n(x,y) = m(x,y) = t = O(n^{\alpha_\Sigma}) = O(n^{\alpha_m}) = O(n)$. Moreover, $d(x,y) \le M(x,y) = t = O(n^{\alpha_\Sigma}) = O(n^{\alpha_d}) = O(n^{\alpha_M})$.
  Clearly, the strings and their LCS length can be computed in time $\Oh(n)$.
\end{proof}

To combine the paddings for different parameters, we need the useful property that all studied parameters sum up if we concatenate strings over disjoint alphabets. 

\begin{lem}[Disjoint Alphabets] \label{lem:disjointalphabet}
  Let $\Sigma_1,\ldots,\Sigma_k$ be disjoint alphabets and let $x_i,y_i$ be strings over alphabet $\Sigma_i$ with $|x_i| \ge |y_i|$ for all $i$. Consider $x := x_1 \ldots x_k$ and $y := y_1 \ldots y_k$. Then for any parameter $p \in \params^*$, we have $p(x,y) = \sum_{i=1}^k p(x_i,y_i)$.
\end{lem}
\begin{proof}
  The statement is trivial for the string lengths $n,m$, alphabet size $|\Sigma|$, and number of matching pairs $M$. For the LCS length $L$ we observe that any common subsequence $z$ can be decomposed into $z_1 \ldots z_k$ with $z_i$ using only symbols from $\Sigma_i$, so that $|z_i| \le L(x_i,y_i)$ and thus $L(x,y) \le \sum_{i=1}^k L(x_i,y_i)$. Concatenating longest common subsequences of $x_i,y_i$, we obtain equality. Using $\delta = m-L$ and $\Delta = n - L$, the claim follows also for $\delta$ and $\Delta$. 
  
  Since every dominant pair is also a matching pair, every dominant pair of $x,y$ stems from prefixes $x_1 \ldots x_j x'$ and $y_1 \ldots y_j y'$, with $x'$ being a prefix of $x_{j+1}$ and $y'$ being a prefix of $y_{j+1}$ for some $j$. Since $L(x_1 \ldots x_j x', y_1 \ldots y_j y') = \sum_{i=1}^j L(x_i,y_i) + L(x',y')$, where the first summand does not depend on $x',y'$, the dominant pairs of $x,y$ of the form $x_1 \ldots x_j x', y_1 \ldots y_j y'$ are in one-to-one correspondence with the dominant pairs of $x_{j+1},y_{j+1}$. This yields the claim for parameter $d$.
\end{proof}

With these preparations we can finish our classification.

\begin{proof}[Proof of \thmref{nontrivialclassification} for large alphabet]
  One direction follows from \lemref{paramsnecessary}. For the other direction, let $\Valpha$ be a parameter setting satisfying \tabref{paramChoiceRestr}.
  For any $n \ge 1$ consider the instances $x_p,y_p$ constructed in \lemref{paddingstotal}, and let them use disjoint alphabets for different $p \in \params^*$. Then the concatenations $x := \bigconcat_{p \in \params^*} x_p$ and $y := \bigconcat_{p \in \params^*} x_p$ form an instance of $\LCS(\Valpha)$, since for any parameter $p \in \params^*$ we have $p(x_p,y_p) = \Theta(n^{\alpha_p})$, and for all other instances $x_{p'},y_{p'}$ the parameter $p$ is $\Oh(n^{\alpha_p})$, and thus $p(x,y) = \Theta(n^{\alpha_p})$ by the Disjoint Alphabets Lemma.
  Thus, we constructed instances of $\LCS(\Valpha)$ of size $\Theta(n)$ for any $n \ge 1$, so the parameter setting $\Valpha$ is non-trivial.
\end{proof} 

We highlight two major hurdles we had to be overcome to obtain this classification result:
\begin{itemize}
\item
Some of the parameter relations of \tabref{paramRelations} are scattered through the LCS literature, e.g., the inequality $d \le L m$ is mentioned in~\cite{ApostolicoG87}. In fact, proving any single one of these inequalities is not very hard -- the main issue was to find a \emph{complete} set of parameter relations. The authors had to perform many iterations of going back and forth between searching for new parameter relations (i.e., extending \lemref{paramsnecessary}) and constructing strings satisfying specific parameter relations (i.e., extending \lemref{paddingstotal}), until finally coming up with a complete list.

\item
The dependency of $d$ on the other parameters is quite complicated. Indeed, eight of the parameter relations of \tabref{paramRelations} involve dominant pairs. Apostolico~\cite{Apostolico86} introduced the parameter under the initial impression that ``it seems that whenever [$M$] gets too close to $mn$, then this forces $d$ to be linear in $m$''. While we show that this intuition is somewhat misleading by constructing instances with high values of both $M$ and $d$, it is a rather complex task to generate a desired number of dominant pairs while respecting given bounds on all other parameters. Intuitively, handling dominant pairs is hard since they involve restrictions on each pair of prefixes of $x$ and $y$.
For \lemref{paddingstotal}, we end up using the strings $(01)^{R+S}$, $0^R(01)^S$ as well as $((1 \circ \ldots \circ t) \circ (t' \circ \ldots \circ 1))^R (1 \circ \ldots \circ t)^{S-R}$, $(1 \circ \ldots \circ t)^S$ for different values of $R,S,t,t'$.
\end{itemize}

\subsection{Monotonicity of Time Complexity} \label{sec:overview_monotonicity}

It might be tempting to assume that the optimal running time for solving LCS is monotone in the problem size~$n$ and the parameters $\params$ (say up to constant factors, as long as all considered parameters settings are non-trivial). However, since the parameters have complex interactions (see \tabref{paramRelations}) it is far from obvious whether this intuition is correct. In fact, the intuition \emph{fails} for $|\Sigma| = 2$, where the running time $\Oh(n + \delta M / n)$ of our new algorithm is not monotone, and thus also the tight time bound $(n + \min\{d, \delta \Delta, \delta M/ n\})^{1 \pm o(1)}$ is not monotone. 

Nevertheless, we will prove monotonicity for any parameter setting $\Valpha$, i.e., when the alphabet size can be assumed to be at least a sufficiently large constant. To formalize monotonicity, we define problem $\LCS_\le(\Valpha)$ consisting of all instances of $\LCS$ with all parameters \emph{at most} as in $\LCS(\Valpha)$.

\begin{defn}[Downward Closure of Parameter Setting]
Fix $\gamma \ge 1$ and let $\Valpha$ be a parameter setting. We define the \emph{downward closure} $\LCS_\le^\gamma(\Valpha)$ as follows. An instance of this problem is a triple $(n,x,y)$, where $p(x,y) \le \gamma \cdot n^{\alpha_p}$ for any $p \in \params^*$, and the task is to compute the length of an LCS of $x,y$. In some statements, we simply write $\LCS_\le(\Valpha)$ to abbreviate that there exists a $\gamma \ge 1$ such that the statement holds for $\LCS_\le^\gamma(\Valpha)$. 

Similarly, for any fixed alphabet $\Sigma$ we consider the downward closure $\LCS_\le^\gamma(\Valpha,\Sigma)$ with instances $(n,x,y)$, where $x,y$ are strings over alphabet $\Sigma$ and $p(x,y) \le \gamma \cdot n^{\alpha_p}$ for any $p \in \params^*$.
\end{defn}

\begin{lem}[Monotonicity] \label{lem:reductiontoclosure}
  For any non-trivial parameter setting $\Valpha$ and $\beta \ge 1$, $\LCS^\gamma(\Valpha)$ has an $\Oh(n^\beta)$-time algorithm for all $\gamma$ if and only if $\LCS_\le^\gamma(\Valpha)$ has an $\Oh(n^\beta)$-time algorithm for all~$\gamma$.
\end{lem}
\begin{proof}[Proof of \lemref{reductiontoclosure}]
  The if-direction follows from the fact that if $(x,y)$ is an instance of $\LCS^\gamma(\Valpha)$ then $(|x|,x,y)$ is an instance of $\LCS_\le^\gamma(\Valpha)$.
  
  For the other direction, 
  let $(n,x,y)$ be an instance of $\LCS_\le(\Valpha)$.
  Since $\Valpha$ is non-trivial, it satisfies \tabref{paramChoiceRestr}, by \thmref{nontrivialclassification}.
  \lemref{paddingstotal} thus allows us to construct paddings $x_p,y_p$ for any $p \in \params^*$ such that (1) $p(x_p,y_p) = \Theta(n^{\alpha_p})$ and (2) $(n,x_p,y_p)$ is an instance of $\LCS_\le(\Valpha)$. We construct these paddings over disjoint alphabets for different parameters and consider the concatenations $x' := x \concat \bigconcat_{p \in \params} x_p$ and $y' := y \concat \bigconcat_{p \in \params} y_p$. Then (1), (2), and the Disjoint Alphabets Lemma imply that $p(x',y') = \Theta(n^{\alpha_p})$ for any $p \in \params^*$,
  so that $(x',y')$ is an instance of $\LCS(\Valpha)$.
  By assumption, we can thus compute $L(x',y')$ in time $O(n^\beta)$. By the Disjoint Alphabets Lemma, we have $L(x,y) = L(x',y') - \sum_{p \in \params} L(x_p,y_p)$, and each $L(x_p,y_p)$ can be computed in time $O(n)$ by \lemref{paddingstotal}, which yields $L(x,y)$ and thus solves the given instance $(n,x,y)$. We finish the proof by observing that the time to construct $x',y'$ is bounded by $\Oh(n)$.
\end{proof}

Note that this proof was surprisingly simple, considering that monotonicity fails for $|\Sigma|=2$.

\subsection{Hardness for Large Alphabet}

Since we established monotonicity for parameter settings $\Valpha$, it suffices to prove hardness for $\LCS_\le(\Valpha)$ instead of $\LCS(\Valpha)$.
This makes the task of constructing hard strings considerably easier, since we only have to satisfy upper bounds on the parameters. Note that our main result \thmref{main1} follows from \lemref{reductiontoclosure} and \thmref{hardness} below.

\begin{thm}[Hardness for Large Alphabet, \secref{hardness}] \label{thm:hardness}
  For any non-trivial parameter setting $\Valpha$, there exists $\gamma \ge 1$ such that $\LCS_\le^\gamma(\Valpha)$ requires time
$\min\big\{ d, \, \delta m, \, \delta \Delta \big\}^{1-o(1)}$,
unless \OVH fails.
\end{thm}
\begin{proof}[Proof Sketch]
  The full proof is deferred to \secref{hardness}.
  We provide different reductions for the cases $\alpha_\delta = \alpha_m$ and $\alpha_L = \alpha_m$. Intuitively, this case distinction is natural, since after this choice all remaining restrictions from \tabref{paramChoiceRestr} are of an easy form: they are linear inequalities.
  
  In the case $\alpha_\delta = \alpha_m$ the complexity $\min\{d, \delta \Delta, \delta m\}^{1 \pm o(1)}$ simplifies to $d^{1 \pm o(1)}$ (since $\alpha_d \le \alpha_L + \alpha_m \le 2 \alpha_m = \alpha_\delta + \alpha_m$ and similarly $\alpha_d \le \alpha_\delta + \alpha_m = 2 \alpha_\delta \le \alpha_\delta + \alpha_\Delta$, see \tabref{paramChoiceRestr}). This simplification makes this case much easier. For constant alphabet, instantiating the known reduction from OV to LCS~\cite{BringmannK15} such that $x$ chooses one of $\approx L$ vectors and $y$ chooses one of $\approx d/L$ vectors yields the claim. For larger alphabet, the right-hand side of the parameter relation $d \le L^2 |\Sigma|$ increases and allows for potentially more dominant pairs. In this case, the second set of vectors would increase to a size of $\approx d/L = \omega(L)$, and the length of an LCS of this construction becomes too large. We thus adapt the reduction by using the construction for constant alphabet multiple times over disjoint alphabets and concatenating the results (reversing the order in one). 
  
  The case $\alpha_L = \alpha_m$ is harder, since all three terms of the complexity $\min\{d, \delta \Delta, \delta m\}^{1 \pm o(1)}$ are relevant. The known reduction~\cite{BringmannK15} fails fundamentally in this case, roughly speaking since the resulting $\delta$ is always as large as the number of vectors encoded by any of $x$ and $y$. 
  Hence, we go back to the ``normalized vector gadgets'' from the known reduction~\cite{BringmannK15}, which encode vectors $a,b$ by strings $\NVG(a),\NVG(b)$ whose LCS length only depends on whether $a,b$ are orthognal. We then carefully embed these gadgets into strings that satisfy any given parameter setting. 
  A crucial trick is to pad each gadget to $\NVG'(a) := 0^{\alpha} 1^\beta (01)^\gamma \NVG(a) 1^\gamma$ for appropriate lengths $\alpha,\beta,\gamma$. It is easy to see that this constructions ensures the following: 
\begin{enumerate}[leftmargin=2\parindent]
  \item[\emph{(1vs1)}] The LCS length of $\NVG'(a),\NVG'(b)$ only depends on whether $a,b$ are orthognal, and 
  \item[\emph{(2vs1)}] $\NVG'(b)$ is a subsequence of $\NVG'(a) \circ \NVG'(a')$ for any vectors $a,a',b$.  
\end{enumerate}  
  In particular, for any vectors $a^{(1)},\ldots,a^{(2k-1)}$ and $b^{(1)},\ldots,b^{(k)}$, on the strings $x = \bigcirc_{i=1}^{2k-1} \NVG'(a^{(i)})$ and $y = \bigcirc_{j=1}^{k} \NVG'(b^{(j)})$ we show that any LCS consists of $k-1$ matchings of type 2vs1 and one matching of type 1vs1 (between $\NVG'(b^{(j)})$ and $\NVG'(a^{(2j-1)})$ for some $j$). Thus, the LCS length of $x$ and $y$ only depends on whether there exists a $j$ such that $a^{(2j-1)}, b^{(j)}$ are orthogonal. Moreover, since most of $y$ is matched by type 2vs1 and thus completely contained in $x$, the parameter $\delta(x,y)$ is extremely small compared to the lengths of $x$ and $y$ -- which is not achievable with the known reduction~\cite{BringmannK15}. Our proof uses an extension of the above construction, which allows us to have more than one matching of type 1vs1. 
  We think that this 1vs1/2vs1-construction is our main contribution to specific proof techniques and will find more applications.
\end{proof}

\subsection{Small Alphabet} \label{sec:overview_smallalphabet}

Proving our results for small constant alphabets poses additional challenges.
For instance, our proof of \lemref{reductiontoclosure} fails for parameter settings $(\Valpha,\Sigma)$ if $|\Sigma|$ is too small, since the padding over disjoint alphabets produces strings over alphabet size at least $|\params| = 7$. In particular, for $|\Sigma| = 2$ we may not use the Disjoint Alphabets Lemma at all, rendering \lemref{paddingstotal} completely useless. 
However, the classification \thmref{nontrivialclassification} still holds for parameter settings $(\Valpha,\Sigma)$. A proof is implicit in \secref{hardnessSmallSigma}, as we construct (infinitely many) hard instances for all parameter settings $(\Valpha,\Sigma)$ satisfying \tabref{paramChoiceRestr}. 

As mentioned above, the Monotonicity Lemma (\lemref{reductiontoclosure}) is wrong for $|\Sigma|=2$, since our new algorithm has a running time $\tOh(n + \delta M / n)$ which is not monotone. Hence, it is impossible to use general strings from $\LCS_\le(\Valpha,\Sigma)$ as a hard core for $\LCS(\Valpha,\Sigma)$. Instead, we use strings from a different, appropriately chosen parameter setting $\LCS_\le(\Valpha',\Sigma')$ as a hard core, see, e.g., \obsref{MnConstructionalphaprime}. 
Moreover, instead of padding with new strings $x_p,y_p$ for each parameter, we need an integrated construction where we control all parameters at once. 
This is a technically demanding task to which we devote a large part of this paper (\secref{hardnessSmallSigma}).
Since the cases $|\Sigma| = 2$, $|\Sigma| = 3$, and $|\Sigma| \ge 4$ adhere to different relations of \tabref{paramChoiceRestr}, these three cases have to be treated separately. Furthermore, as for large alphabet we consider cases $\alpha_\delta = \alpha_m$ and $\alpha_L = \alpha_m$. Hence, our reductions are necessarily rather involved and we need to very carefully fine-tune our constructions.

\section{Organization}
\label{sec:appendixIntro}

The remainder of the paper contains the proofs of Theorems~\ref{thm:nontrivialclassification}, \ref{thm:main1}, \ref{thm:main2}, and \ref{thm:algo}, following the outline given in \secref{roughproof}. Specifically, \secref{params} lists and proves our complete set of parameter relations (proving \lemref{paramsnecessary}). In \secref{basicFacts}, we prove basic facts and technical tools easing the constructions and proofs in later sections -- this includes a simple greedy prefix matching property as well as a surprising technique to reduce the number of dominant pairs of two given strings.
In \secref{paddings}, we show how to pad each parameter individually (proving \lemref{paddingstotal}). 
\secref{hardness} then constructs hard instances for large alphabet (for the downward closure of any parameter setting, proving \thmref{hardness} and thus \thmref{main1}). Finally, the much more intricate case of small constant alphabet sizes such as $|\Sigma|=2$ is handled in \secref{hardnessSmallSigma}, which takes up a large fraction of this paper (proving \thmref{main2}). We present our new algorithm in \secref{algo} (proving \thmref{algo}). Finally, \secref{bpseth} describes the necessary modifications to our hardness proofs to show the same conditional lower bounds also under the weaker variant of SETH used in~\cite{AbboudHVWW16}.

We remark that for some intermediate strings $x,y$ constructed in the proofs, the assumption $|x| \ge |y|$ may be violated; in this case we use the definitions given in \figref{paramSummary} (and thus we may have $n(x,y) < m(x,y)$ and $\Delta(x,y) < \delta(x,y)$). Since $L, M, d$, and $\Sigma$ are symmetric in the sense $L(x,y) = L(y,x)$, these parameters are independent of the assumption $|x| \ge |y|$.

For simplicity, we will always work with the following equivalent variant of \OVH. 

\medskip
\noindent
\textbf{Unbalanced Orthogonal Vectors Hypothesis (\UOVH):} \emph{
For any $\alpha,\beta \in (0,1]$, and computable functions $f(n)=n^{\alpha- o(1)},\, g(n) = n^{\beta - o(1)}$, the following problem requires time $n^{\alpha+\beta-o(1)}$: Given a number $n$, solve a given \OV instance with $D=n^{o(1)}$ and $|\sA|=f(n)$ and $|\sB|=g(n)$.
} 

\begin{lem}[Essentially folklore]
\label{lem:UOVH}
\UOVH is equivalent to \OVH.
\end{lem}
\begin{proof}
Clearly, \UOVH implies \OVH (using $\alpha=\beta=1, f(n) = g(n)=n$). For the other direction, assume that \UOVH fails and let $\alpha,\beta \in (0,1]$, $f(n)=n^{\alpha - o(1)}$, and $g(n)=n^{\beta - o(1)}$ be such that $\OV$ with $D=n^{o(1)}$ and $|\sA|=f(n)$ and $|\sB|=g(n)$ can be solved in time $\Oh(n^{\alpha+\beta-\varepsilon})$ for some constant $\varepsilon >0$. Consider an arbitrary \OV instance $\sA, \sB \subseteq \{0,1\}^D$ with $D=n^{o(1)}$. We partition $\sA$ into $s:=\lceil \frac{n}{f(n)}\rceil$ sets $\sA_1,\dots, \sA_s$ of size $f(n)$ and $\sB$ into $t:=\lceil \frac{n}{g(n)} \rceil$ sets $\sB_1,\dots,\sB_t$ of size $g(n)$ (note that the last set of such a partition might have strictly less elements, but can safely be filled up using all-ones vectors). By assumption, we can solve each \OV instance $\sA_i,\sB_j$ in time $\Oh(n^{\alpha+\beta-\varepsilon})$. Since there exist $a\in \sA, b\in \sB$ with $\langle a,b\rangle = 0$ if and only if there exist $a\in \sA_i, b\in \sB_j$ with $\langle a_i,b_j\rangle = 0$ for some $i \in [s],j\in [t]$, we can decide the instance $\sA,\sB$ by sequentially deciding the $s\cdot t=\Oh(n^{2-(\alpha+\beta)+o(1)})$ \OV instances $\sA_i,\sB_j$. This takes total time $\Oh(s \cdot t \cdot n^{\alpha+\beta -\varepsilon}) = \Oh(n^{2-\varepsilon'})$ for any $\varepsilon' < \varepsilon$, which contradicts \OVH and thus proves the claim.
\end{proof}

\section{Parameter Relations}
\label{sec:params}

In this section we prove relations among the studied parameters, summarized in Table~\ref{tab:paramRelations}.
Some of these parameter relations can be found at various places in the literature, however, our set of relations is complete in the sense that any parameter setting $\Valpha$ is non-trivial if and only if it satisfies all our relations, see \thmref{nontrivialclassification}.

Consider a relation like $d(x,y) \le L(x,y)\cdot m(x,y)$, given in \lemref{dUBs}(i) below. Fix exponents $\alpha_d, \alpha_L$, and $\alpha_m$, and consider all instances $x,y$ with $d(x,y) = \Theta(n^{\alpha_d}), L(x,y) = \Theta(n^{\alpha_L})$, and $m(x,y) = \Theta(n^{\alpha_m})$. Note that the relation may be satisfied for infinitely many instances if $\alpha_d \le \alpha_L + \alpha_m$. On the other hand, if $\alpha_d > \alpha_L + \alpha_m$ then the relation is satisfied for only finitely many instances. This argument translates Table~\ref{tab:paramRelations} into Table~\ref{tab:paramChoiceRestr} (using $\alpha_n = 1$), thus generating a complete list of restrictions for non-trivial parameter settings. 

Let $x,y$ be any strings. In the remainder of this section, for convenience, we write $p = p(x,y)$ for any parameter $p \in \params^*$.
Recall that by possibly swapping $x$ and $y$, we may assume $m = |y| \le |x| = n$. 
This assumption is explicit in our definition of parameter settings. For some other strings $x,y$ considered in this paper, this assumption may be violated. In this case, the parameter relations of \tabref{paramRelations} still hold after replacing $n$ by $\max\{n,m\}$ and $m$ by $\min\{n,m\}$, as well as $\Delta$ by $\max\{\Delta,\delta\}$ and $\delta$ by $\min\{\Delta,\delta\}$ (as the other parameters are symmetric).

Note that \assumpref{symbols} (i.e., every symbol in $\Sigma$ appears at least once in $x$ and $y$) implies $|\Sigma| \le m$ and ensures that any symbol of $x$ has at least one matching symbol in $y$, and thus $M \ge n$.

We next list trivial relations. 
The length of the LCS $L$ satisfies $L\le m$. The numbers of deleted positions satisfy $\Delta = n-L \le n$, $\delta = m - L \le m$, and $\delta \le \Delta$. 
Since any dominant pair is also a matching pair, we have $d \le M$. Moreover, $d \ge L$ since for any $1 \le k \le L$ there is at least one $k$-dominant pair.

To prepare the proofs of the remaining relations, recall that we defined $L[i,j] = L(x[1..i],y[1..j])$. Moreover, observe that $L(x,y) \le \sum_{\sigma\in \Sigma} \min\{\occ_\sigma(x),\occ_\sigma(y)\}$, which we typically exploit without explicit notice. Furthermore, we will need the following two simple facts. 

\begin{obs}\label{obs:occAtMostL}
For any $\sigma \in \Sigma$, we have $\occ_\sigma(x) \le L$ or $\occ_\sigma(y) \le L$.
\end{obs}
\begin{proof}
If some $\sigma\in \Sigma$ occurs at least $L+1$ times in both $x$ and $y$, then $\sigma^{L+1}$ is an LCS of $x$ and $y$ of length $L+1>L$, which is a contradiction.
\end{proof}

\begin{obs}\label{obs:atMostOned}
Fix $1\le k \le L$ and $1\le \bar{i} \le n$. Then there is at most one $k$-dominant pair $(\bar{i},j)$ with $1 \le j \le m$, namely the pair $(\bar{i}, j^*)$ with $j^* = \min \{j \mid L[\bar{i},j] = k\}$ if it exists. 
Symmetrically for every $1\le k\le n$ and $1\le \bar{j} \le m$, there is at most one $k$-dominant pair $(i,\bar{j})$ with $1 \le i \le n$.
\end{obs}
\begin{proof}
All pairs $(\bar{i},j)$ with $j\ne j^*$ and $L[\bar{i},j] = k$ satisfy $j\ge j^*$, so they are dominated by $(\bar{i},j^*)$. 
\end{proof}

We are set up to prove the more involved relations of Table~\ref{tab:paramRelations}. We remark that while the inequality $d\le Lm$ is well known since the first formal treatment of dominant pairs, the bound $d\le L^2 |\Sigma|$ seems to go unnoticed in the literature.

\begin{lem}\label{lem:dUBs}
It holds that (i) $d \le Lm$ and (ii) $d\le L^2 |\Sigma|$.
\end{lem}
\begin{proof}
(i) Let $1\le k \le L$. For any $1\le \bar j \le m$ there is at most one $k$-dominant pair $(i,\bar j)$ by Observation~\ref{obs:atMostOned}. This proves $|D_k| \le m$ and thus $d = \sum_{i=1}^L |D_k| \le L m$.

(ii) Let $\sigma \in \Sigma$. By Observation~\ref{obs:occAtMostL}, we may assume that $\occ_\sigma(x) \le L$ (the case $\occ_\sigma(y)\le L$ is symmetric). For any occurrence $i_\sigma$ of $\sigma$ in $x$ and any $1\le k \le L$, there can be at most one $k$-dominant pair $(i_\sigma,j)$ by \obsref{atMostOned}. Hence, $\sigma$ contributes at most $L$ $k$-dominant pairs. Summing over all $\sigma \in \Sigma$ and $k=1,\dots,L$ yields the claim.
\end{proof}

\begin{lem}\label{lem:dLgap}
We have $d\le 2L (\Delta + 1)$.
\end{lem}
\begin{proof}
Fix an LCS $z$ of $x$ and $y$. Since $z$ can be obtained by deleting at most $\Delta=n-L$ positions from $x$ or by deleting at most $\delta=m-L$ positions from $y$, $x[1..i]$ and $y[1..j]$ contain $z[1..i-\Delta]$ and $z[1..j-\delta]$, respectively, as a subsequence. Hence, we have $\min\{i-\Delta,j-\delta\} \le L[i,j] \le \min \{ i,j\}$.

Let $1\le k \le L$. By the previous property, if $L[i,j]=k$ then (i) $k \le i \le k + \Delta$ or (ii) $k \le j \le k+\delta$. Note that for each $\bar{i} \in \{k, \dots, k+\Delta\}$ we have (by \obsref{atMostOned}) at most one $k$-dominant pair $(\bar{i},j)$, and similarly, for each $\bar{j}\in \{k, \dots, k+\delta\}$ we have at most one $k$-dominant pair $(i,\bar{j})$. This proves $|D_k|\le \Delta + \delta+2 \le 2 (\Delta + 1)$, from which the claim follows.
\end{proof}

\begin{lem}\label{lem:dSigma}
We have $d \ge |\Sigma|$. 
\end{lem}
\begin{proof}
By \assumpref{symbols}, every symbol $\sigma \in \Sigma$ appears in $x$ and $y$. Let $i$ be minimal with $x[i] = \sigma$ and $j$ be minimal with $y[j] = \sigma$. We show that $(i,j)$ is a dominant pair of $x,y$, and thus $d \ge |\Sigma|$. Let $k = L[i,j] = L(x[1..i],y[1..j])$. Since $x[i] = y[j]$, we have $L[i-1,j-1] = k-1$. Moreover, since the last symbol in $x[1..i]$ does not appear in $y[1..j-1]$, it cannot be matched, and we obtain $L[i,j-1] = L[i-1,j-1] = k-1$. Similarly, $L[i-1,j] = k-1$. This proves that $(i,j)$ is a $k$-dominant pair of $x,y$, as desired.
\end{proof}

\begin{lem}\label{lem:Mbounds}
We have (i) $M \ge L^2 / |\Sigma|$ and (ii) $M \le 2Ln$.
\end{lem}
\begin{proof}
(i) Let $z$ be an LCS of $x$ and $y$. We have $M = \sum_{\sigma \in \Sigma} \occ_\sigma(x) \cdot \occ_\sigma(y) \ge \sum_{\sigma \in \Sigma} \occ_\sigma(z)^2$. By $\sum_{\sigma \in \Sigma} \occ_\sigma(z) = L$ and the arithmetic-quadratic mean inequality, the result follows.

(ii) Let $\Sigma_w := \{\sigma \in \Sigma \mid \occ_\sigma(w) \le L\}$ for $w\in \{x,y\}$. By \obsref{occAtMostL}, we have $\Sigma_x \cup \Sigma_y = \Sigma$. We can thus bound 
\[ M = \sum_{\sigma \in \Sigma} \occ_\sigma(x)\cdot \occ_\sigma(y) \le \sum_{\sigma \in \Sigma_y} L\cdot \occ_\sigma(x) + \sum_{\sigma \in \Sigma_x} L\cdot \occ_\sigma(y)  \le L(n+m) \le 2Ln. \qedhere\]
\end{proof}

\paragraph{Small alphabets.}
Not surprisingly, in the case of very small alphabets there are more relations among the parameters. The following relations are specific to $|\Sigma| = 2$ and $|\Sigma| = 3$.

\begin{lem}\label{lem:MbinaryLm}
Let $\Sigma = \{0,1\}$. Then $M\ge Lm/4$.
\end{lem}
\begin{proof}
Let $z$ be an LCS of $x$ and $y$. Without loss of generality we may assume $\occ_0(x) \ge n/2$ (by possibly exchanging 0 and 1). If $\occ_0(y) \ge L/2$, then $M \ge \occ_0(x)\cdot \occ_0(y) \ge Ln/4 \ge Lm/4$. Otherwise we have $\occ_0(y) < L/2\le m/2$, which implies $\occ_1(y) \ge m/2$. By $\occ_0(y) < L/2$, we must have that $\occ_1(x) \ge L/2$, since otherwise $L \le \min\{\occ_0(x),\occ_0(y)\}+\min\{\occ_1(x),\occ_1(y)\} \le \occ_0(y) +\occ_1(x) < L$, which is a contradiction. Hence $M\ge \occ_1(x)\cdot \occ_1(y) \ge Lm/4$, proving the claim.
\end{proof}

The following lemma (which is also applicable for large alphabets) asserts that if most positions in $x$ and $y$ are the same symbol, say $0$, then the number of dominant pairs is small.

\begin{lem}\label{lem:few1s}
Let 0 be a symbol in $\Sigma$ and set $\lambda := \sum_{\sigma \in \Sigma \setminus \{0\}} \min \{\occ_\sigma(x),\occ_\sigma(y)\}$. Then $d \le  5 \lambda L$.  In particular, for $\Sigma=\{0,1\}$, we have $d \le 5 L \cdot \occ_1(y)$.
\end{lem}
\begin{proof}
Let $1 \le k \le L$ and $\sigma \in \Sigma\setminus\{0\}$. By \obsref{atMostOned}, there are at most $\min \{ \occ_\sigma(x), \occ_\sigma(y)\}$ $k$-dominant pairs $(i,j)$ with $x[i] = y[j] = \sigma$. Hence in total, there are at most $\lambda \cdot L$ dominant pairs $(i,j)$ with $x[i] = y[j] \ne 0$.

To count the remaining dominant pairs, which are contributed by $0$, we use a similar argument to~\lemref{dLgap}. Let $1 \le k \le L$ and consider any pair $(i,j)$ with $x[i] = y[j] = 0$, say $x[i]$ is the $\ell_x$-th occurrence of 0 in $x$ and $y[j]$ is the $\ell_y$-th occurrence of 0 in $y$. If $(i,j)$ is a $k$-dominant pair then $k-\lambda \le \min \{ \ell_x, \ell_y\} \le k$. Indeed, if $\min \{\ell_x,\ell_y\} < k-\lambda$, then 
\[L[i,j] \le \sum_{\sigma \in \Sigma} \min \{\occ_\sigma(x[1..i]),\occ_\sigma(y[1..j])\} \le \min\{\ell_x,\ell_y\} + \lambda < k,\]
contradicting the definition of a $k$-dominant pair. Moreover, if $\min\{\ell_x,\ell_y\} > k$, then $0^{\min\{\ell_x,\ell_y\}}$ is a common subsequence of $x[1..i], y[1..j]$ of length strictly larger than $k$, which is again a contradiction to $(i,j)$ being a $k$-dominant pair.

Hence, we have $k - \lambda \le \ell_x \le k$ or $k - \lambda \le \ell_y \le k$. 
Since any choice of $\ell_x$ uniquely determines~$i$ by \obsref{atMostOned} (and symmetrically $\ell_y$ determines $j$), there are at most $2\lambda + 2$ $k$-dominant pairs with $x[i]=y[j] = 0$.
In total, we have at most $(3\lambda + 2)L\le 5 \lambda L$ dominant pairs (note that $\lambda \ge 1$ by Assumption~\ref{assump:symbols} and $|\Sigma| \ge 2$).
\end{proof}

\begin{lem}\label{lem:Mbinarynd}
If $\Sigma = \{0,1\}$ then $M \ge nd/(5L)$.
\end{lem}
\begin{proof}
Without loss of generality assume that $\min \{\occ_0(x),\occ_0(y)\} \ge \min \{\occ_1(x),\occ_1(y)\}$ (by possibly exchanging 0 and 1). Then $\lambda = \min\{\occ_1(x),\occ_1(y)\}$ satisfies $\occ_\sigma(y)\ge \lambda$ for all $\sigma\in \Sigma$. Thus,  $M = \sum_{\sigma \in \Sigma} \occ_\sigma(x)\cdot \occ_\sigma(y) \ge (\occ_0(x)+\occ_1(x))\cdot \lambda =\lambda n$. By \lemref{few1s}, we have $\lambda \ge d/(5L)$ and the claim follows.
\end{proof}

For ternary alphabets, the following weaker relation holds.

\begin{lem}\label{lem:Mternarymd}
If $\Sigma = \{0,1,2\}$ then $M \ge md/(80 L)$.
\end{lem}
\begin{proof}
Let $z$ be an LCS of $x$ and $y$. We permute the symbols such that $\sigma = 0$ maximizes $\occ_\sigma(z)$. Thus, we have $\occ_0(x) \ge \occ_0(z) \ge L/|\Sigma|=L/3$ and symmetrically $\occ_0(y)\ge L/3$. 

If $\occ_0(y) \ge m/2$ then we have $M \ge \occ_0(x)\cdot \occ_0(y) \ge Lm/6 \ge dm/(18L)$ by \lemref{dUBs}(ii). Similarly, if $\occ_0(x) \ge n/2$ then we obtain $M \ge Ln/6 \ge dn/(18L)\ge dm/(18L)$. Hence, it remains to consider the case $\occ_0(x) \le n/2$ and $\occ_0(y) \le  m/2$. Let $x',y'$ be the subsequences obtained by deleting all 0s from $x$ and $y$, respectively, and note that $|x'|,|y'| \ge m/2$. Since $x',y'$ have alphabet size 2, \lemref{MbinaryLm} is applicable and yields $M(x', y') \ge L(x',y')\cdot m(x',y') / 4$. Observe that there is a common subsequence of $x',y'$ of length at least $\lambda/2$, where $\lambda = \sum_{\sigma \in \Sigma \setminus \{0\}} \min \{\occ_\sigma(x),\occ_\sigma(y)\}$ (consider the longer subsequence of $1^{\min\{\occ_1(x),\occ_1(y)\}}$ and $2^{\min\{\occ_2(x),\occ_2(y)\}}$). Hence,
\[M(x,y)\ge M(x',y') \ge \frac{1}{4}\cdot L(x',y')\cdot m(x',y') \ge \frac{1}{4}\cdot \frac{\lambda}{2} \cdot \frac{m}{2}  = \frac{\lambda m}{16}.\]
The claim now follows from $\lambda \ge d/(5L)$ as proven in \lemref{few1s}.
\end{proof}

\section{Technical Tools and Constructions}
\label{sec:basicFacts}

To prepare later constructions and ease their analysis, this section collects several technical results. We start off with the simple fact that equal prefixes can be greedily matched.

\begin{lem}[Greedy Prefix Matching]\label{lem:greedy}
For any strings $w,x,y$, we have $L(w x, w y)=|w| + L(x,y)$ and $d(w x, w y) = |w| + d(x,y)$. 
\end{lem}
\begin{proof}
For the first statement, it suffices to prove the claim when $w=0$ is a single symbol (by induction and renaming of symbols). 
Consider any common subsequence $z$ of $0x, 0y$. If $z$ is of the form $0 z'$, then $z'$ is a common subsequence of $x,y$, so $|z| \le 1 + L(x,y)$. If the first symbol of $z$ is not $0$, then the first symbols of $0x,0y$ are not matched, so $z$ is a common subsequence of $x,y$ and we obtain $|z| \le L(x,y)$. In total, $L(0x,0y) \le 1 + L(x,y)$. The converse holds by prepending $0$ to any LCS of $x,y$. 

For the second statement, let $x' = wx$ and $y' = wy$. For $i\in [|w|]$, we have $L(x'[1..i],y'[1..i]) = i$. Hence $(i,i)$ is the unique $i$-dominant pair of $x',y'$ and no other dominant pairs $(i,j)$ with $i\le |w|$ or $j \le |w|$ exist. This yields $|w|$ dominant pairs. Consider now any $(i,j)$ with $i = |w| + \bar{i}$ and $j = |w| + \bar{j}$ where $\bar{i} \in [|x|], \bar{j} \in [|y|]$. By the first statement, $L(x'[1..i],y'[1..j]) = |w| + L(x[1..\bar{i}], y[1..\bar{j}])$. Thus $(i,j)$ is a $(|w|+k)$-dominant pair of $x'$ and $y'$ if and only if $(\bar{i},\bar{j})$ is a $k$-dominant pair of $x$ and $y$. This yields $d(x,y)$ additional dominant pairs, proving the claim.
\end{proof}

For bounding the number of dominant pairs from below we often use the following observation.

\begin{obs}\label{obs:prefix}
For any strings $a$, $x$, $b$, $y$, we have $d(ax,by) \ge d(a,b)$.
\end{obs}
\begin{proof}
This trivially follows from the fact that any prefixes $a'$, $b'$ of $a$, $b$ are also prefixes of $ax$,$by$.
\end{proof}

\subsection{Generating dominant pairs}

The dependency of $d$ on the other parameters is quite complicated, and hence it is a rather complex task to generate a desired number of dominant pairs while respecting given bounds on all other parameters. We present different constructions for this purpose in the following.

\begin{figure}
\centering
\small
\input{pics/dom1.tikz}
\caption{The $L$-table for the strings $a=(01)^{R+S}$ and $b=0^R (01)^S$ with $R=4, S=5$ (where the entry in row $j$ and column $i$ denotes $L(a[1..i],b[1..j])$). The indicated dominant pairs visualize the results of \lemref{genDomPairs}.}
\label{fig:domPairsI}
\end{figure}

The following lemma establishes the first such construction, illustrated in \figref{domPairsI}. We remark that statements \itemrefs{genDomPairsLmore}{genDomPairsdmore} are technical tools that we will only use for $|\Sigma| = \Oh(1)$ in \secref{hardnessSmallSigma}.

\begin{lem}[Generating dominant pairs]\label{lem:genDomPairs}
Let $R,S\ge 0$ and define $a := (01)^{R+S}$ and $b := 0^{R} (01)^{S}$. The following statements hold.
\begin{enumerate}[label=(\roman{*})]
 \item\label{itm:genDomPairsL} We have $L(a,b) = |b| = R+2S$.
 \item\label{itm:genDomPairsd} We have $R\cdot S \le d(a,b) \le \min\{2(R+1),5S\} \cdot (R+2S)=\Oh(R\cdot S)$. 
 \item\label{itm:genDomPairsLmore} For any $\alpha, \beta, \beta' \ge 0$, we have $L(a 1^\alpha, 0^\beta b 0^{\beta'}) = |b|=R + 2S$.
 \item\label{itm:genDomPairsdmore} For any $\alpha, \beta, \beta' \ge 0$, we have $R\cdot S\le d(a 1^\alpha, 0^\beta b 0^{\beta'}) \le 2(\max\{R+\alpha,\beta+\beta'\}+1)(R+2S)$.
\end{enumerate}
\end{lem}
\begin{proof}
All statements follow from the following fact. 
\begin{enumerate}[label=($\ast$)]
\item \label{itm:genDomPairsHelper} For any $0\le s \le S, s\le r\le R+s$ and $\beta\ge 0$, we have $L((01)^r, 0^\beta 0^{R} (01)^{s}) = r+s$. 
\end{enumerate}
To prove \ref{itm:genDomPairsHelper}, note that by \lemref{greedy} (reversing the strings)
we have \[L((01)^r,0^{\beta+R} (01)^s) = 2s + L((01)^{r-s}, 0^{\beta+R}) = 2s + \min\{\occ_0((01)^{r-s}),\beta+R\} = 2s + (r-s) = r + s.\]
Statement \ref{itm:genDomPairsL} now follows from setting $s=S$, $r=R+S$, and $\beta = 0$. 

To see \ref{itm:genDomPairsLmore}, note that $L(a 1^\alpha,0^\beta b 0^{\beta'}) \ge L(a,b) = |b|$ by \itemref{genDomPairsL}. For the upper bound, we compute 
\begin{eqnarray*}
L(a 1^\alpha,0^\beta b 0^{\beta'}) & \le & \min\{\occ_0(a1^\alpha),\occ_0(0^\beta b0^{\beta'})\}+\min\{\occ_1(a1^\alpha),\occ_1(0^\beta b0^{\beta'})\} \\
& = & \min\{R+S, R+S+\beta+\beta'\} + \min\{R+S+\alpha,S\} = R+2S=L(a,b).
\end{eqnarray*}

To prove \itemref{genDomPairsd}, note that $d(a,b) \le 5 \cdot \occ_1(b) \cdot L(a,b) = 5 S (R+2S)$ by \lemref{few1s}.  The bound $d(a,b) \le 2(R+1)\cdot (R+2S)$ follows from \itemref{genDomPairsdmore} by setting $\alpha=\beta=\beta'=0$, hence it remains to prove~\itemref{genDomPairsdmore}.

 For the lower bound, we use $d(a 1^\alpha, 0^\beta b 0^{\beta'}) \ge d(a, 0^\beta b)$ (by \obsref{prefix}) and consider $L'[r,s] := L((01)^r, 0^{\beta+R} (01)^s)$. We prove that for any $1\le s\le S$, $s < r\le R+s$, we have at least one $(r+s)$-dominant pair $(i,j)$ with $2(r-1) < i \le 2r$ and $\beta+R+2(s-1) < j \le \beta+R+2s$. Indeed, $L'[r,s] = r+s$ (by \ref{itm:genDomPairsL}) implies that an $(r+s)$-dominant pair $(i,j)$ with $i\le 2r$ and $j \le \beta + R+2s$ exists. If we had $i\le 2(r-1)$, then by monotonicity of $L(x[1..i],y[1..j])$ in $i$ and $j$ we would have $L'[r-1,s] \ge r+s$, contradicting $L'[r-1,s] = r+s-1$ (by \ref{itm:genDomPairsL}). Thus, we obtain $i > 2(r-1)$, and symmetrically we have $j > \beta+R+2(s-1)$. Hence, for every $1\le s\le S$, $s < r\le R+s$, we have at least one dominant pair which is not counted for any other choice of $r$ and $s$. Since for any $1\le s\le S$ there are $R$ choices for $s< r\le R+s$, we conclude that $d(a,b) \ge S\cdot R$. For the upper bound, note that $\Delta(a1^\alpha ,0^\beta b 0^{\beta'}) = \max\{R+\alpha,\beta+\beta'\}$ by \itemref{genDomPairsLmore}, and hence \lemref{dLgap} yields $d(a1^\alpha,0^\beta b0^{\beta'}) \le 2\cdot L(a1^\alpha ,0^\beta b 0^{\beta'}) \cdot (\Delta(a1^\alpha,0^\beta b0^{\beta'})+1) = 2(R+2S)(\max\{R+\alpha,\beta+\beta'\}+1)$.
\end{proof}

The previous construction uses alphabet size $|\Sigma|=2$, enforcing $M(a,b)\ge L(a,b)^2/2$. If the desired number of matching pairs is much smaller than $L^2$, which can only be the case if the desired alphabet size is large, we have to use a more complicated construction exploiting the larger alphabet. To make the analysis more convenient, we first observe a simple way to bound the number of dominant pairs from below: if we can find $k$ pairwise non-dominating index pairs $(i,j)$ that have the same LCS value (of the corresponding prefixes) and whose predecessors $(i-1,j-1)$ have a strictly smaller LCS value, then at least $k/2$ dominant pairs exist.
\begin{lem}\label{lem:countDomPairs}
For any strings $x,y$ set $L[i,j] := L(x[1..i],y[1..j])$. 
Suppose that there are index pairs $(i_1,j_1),\dots,(i_k,j_k)$ with $i_1 < i_2 < \dots < i_k$ and $j_1 > j_2 > \dots > j_k$ such that for some $\gamma$ and all $1\le \ell \le k$ we have $L[i_\ell,j_\ell] = \gamma$ and $L[i_\ell - 1, j_\ell -1] = \gamma-1$. Then the number of $\gamma$-dominant pairs of $x$ and $y$  is at least $k/2$.
\end{lem}
\begin{proof}
For each $k$, fix any $\gamma$-dominant pair $p_\ell=(i^*_\ell, j^*_\ell)$ that dominates $(i_\ell,j_\ell)$. Note we may have $p_\ell = (i_\ell,j_\ell)$ if $(i_\ell,j_\ell)$ is itself a dominant pair, and thus $p_\ell$ always exists. Set $i_0 := 1$. We argue that for every $\ell$, we have $i_{\ell-1} \le i^*_\ell \le i_\ell$.

Note that for all $\ell$, we have $L[i_\ell - 1,j_\ell-1]<\gamma$ and hence $L[i,j] < \gamma$ for all $i<i_\ell$, $j<j_\ell$. Thus, we have either (1) $i^*_\ell = i_\ell$ (and $j^*_\ell \le j_\ell$) or (2) $j^*_\ell = j_\ell$ (and $i^*_\ell \le i_\ell$). Case (1) trivially satisfies $i_{\ell-1} \le i^*_\ell \le i_\ell$. Thus, it remains to argue that in case (2) we have $i_{\ell-1} \le i^*_\ell$. Indeed, otherwise we have $i^*_\ell < i_{\ell-1}$ and $j^*_\ell = j_\ell < j_{\ell-1}$, which implies $L[i^*_\ell,j^*_\ell] < \gamma$ and $(i^*_\ell,j^*_\ell)$ is no $\gamma$-dominant pair. 

Note that the above property implies $p_\ell \ne p_{\ell+2}$ for all $1\le \ell \le k-2$, since $i^*_\ell \le i_\ell < i_{\ell+1} \le i^*_{\ell+2}$. Thus, the number of $\gamma$-dominant pairs is bounded from below by $|\{ p_\ell \mid 1\le \ell \le k\}| \ge k/2$.
\end{proof}

Note that the previous lemma would not hold without the condition $L[i_\ell-1,j_\ell-1]= \gamma-1$. 
We are set to analyze our next construction, which is illustrated in \figref{domPairsII}.

\begin{lem}[Generating dominant pairs, large alphabet]\label{lem:genDomPairs-largeSigma}
Let $t \ge 2$, $1 \le t' \le t$, and $S \ge R \ge 1$. Over alphabet $\Sigma = \{1,\dots,t\}$ we define the strings 
\begin{eqnarray*}
a & := & ((1\dots t) \concat (t' \dots 1))^{R} \concat (1\dots t)^{S-R}, \\
b & := & (1\dots t)^S.
\end{eqnarray*}
It holds that
\begin{enumerate}[label=(\roman{*})]
\item\label{itm:genDomPairs-largeSigma-L} $L(a,b) = |b| = St$,
\item\label{itm:genDomPairs-largeSigma-d} Assume that $S\ge R(t'+1)$. Then $(St)(Rt')/8\le d(a,b) \le 4(St)(Rt')$,
\item\label{itm:genDomPairs-largeSigma-M} $tS^2 \le M(a,b) \le t(S+R)S$.
\end{enumerate}
\end{lem}
\begin{figure}
\centering
\small
\input{pics/dom-largeSigma.tikz}
\caption{The $L$-table for strings $a,b$ of \lemref{genDomPairs-largeSigma} with $t=t'=3$, $R=2$, $S=7$.}
\label{fig:domPairsII}
\end{figure}

\begin{proof}
Note that \itemref{genDomPairs-largeSigma-L} trivially follows from the fact that $b$ is a subsequence of $a$. For \itemref{genDomPairs-largeSigma-M}, observe that for all $\sigma\in \Sigma$, we have $S\le \occ_\sigma(a) \le R+S$ and $\occ_\sigma(b)=S$, from which the claim follows immediately. The upper bound of \itemref{genDomPairs-largeSigma-d} follows from $d(a,b) \le 2\cdot L(a,b) \cdot (\Delta(a,b)+1) = 2 (St) (Rt'+1) \le 4(St)(Rt')$ (by \lemref{dLgap}). To prove the remaining lower bound, we establish the following fact.

\begin{enumerate}[label=($\ast$)]
\item\label{itm:genDomPairs-largeSigma-basic} Let $w:= 1\dots t$ and $w':= t' \dots 1$. Then $L((w w')^{R}, w^{R+k} ) = Rt + k$ for all $0\le k\le t'R$.
\end{enumerate}
Let us postpone the proof and show that \itemref{genDomPairs-largeSigma-d} follows from \itemref{genDomPairs-largeSigma-basic}. Define $v := w^{S-R}$. For $0 \le k \le K := \min\{S-R,Rt'\}$ and $0\le \ell \le (S-R-k)t$, we let $a(k,\ell) := (ww')^R v[1..\ell]$ and $b(k,\ell) := w^{R+k} v[1..\ell]$. Note that $a(k,\ell)$ and $b(k,\ell)$ are prefixes of $a$ and $b$, respectively. By greedy suffix matching (i.e., \lemref{greedy} applied to the reversed strings) and~\itemref{genDomPairs-largeSigma-basic}, we obtain
\[ L(a(k,\ell),b(k,\ell)) = \ell + L((ww')^R,w^{R+k})= Rt+k+\ell.\] 
Hence, any $0\le k \le K$, $1\le \ell \le (S-R-k)t$ give rise to an index pair $(i,j)$ with $L(a[1..i],b[1..j])=L(a(k,\ell),b(k,\ell)) > L(a(k,\ell-1),b(k,\ell-1))=L(a[1..i-1],b[1..j-1])$. Let $I_\gamma$ denote the set of all such index pairs $(i,j)$ that additionally satisfy $L(a[1..i],b[1..j])=\gamma$. Then for any $\gamma$, no $(i,j)\in I_\gamma$ dominates another $(i',j')\in I_\gamma$. Thus by \lemref{countDomPairs}, $d(a,b) \ge \sum_{\gamma} |I_\gamma|/2$. By counting all possible choices for $k$ and $\ell$, we obtain $\sum_\gamma |I_\gamma| = \sum_{0\le k \le K} t(S-R-k) \ge tK(S-R)/2$. This yields $d(a,b) \ge t \cdot  \min\{S-R, Rt'\} \cdot (S-R)/4$. For $S \ge R(t'+1)$, we have $S-R\ge S/2$ as well as $S-R \ge Rt'$ and the lower bound of \itemref{genDomPairs-largeSigma-d} follows.

To prove~\itemref{genDomPairs-largeSigma-basic}, let $a' := (ww')^R$ and $b' := w^{R+k}$. For the lower bound, it is easy to see that we can completely match $R$ of the copies of $w$ in $b'$ to copies of $w$ in $a'$, and at the same time match a single symbol in each of the remaining $k$ copies of $w$ in $b$ to a single symbol in some copy of $w'$ in $a'$ (since $k\le R|w'|=Rt'$). This yields $L(a',b') \ge R|w| + k = Rt+k$.

For the upper bound, we write $b'= \bigconcat_{j=1}^{R+k} b_j$ with $b_j := w$ and consider a partitioning $a' = \bigconcat_{j=1}^{R+k} a_j$ such that $L(a',b') = \sum_{j=1}^{R+k} L(a_j,b_j)$. For any $a_j$, let $w(a_j)$ denote the number of symbols that $a_j$ shares with any occurrence of $w$ (if, e.g.,  $a_j = x w' y$ for some prefix $x$ of $w$ and some suffix $y$ of $w$, then $w(a_j) = |x| + |y|$). We first show that
\begin{equation} L(a_j,b_j) \le \begin{cases} 1 & \text{if } w(a_j) = 0, \\ \min\{w(a_j),|w|\} & \text{otherwise.} \end{cases} \label{eq:LUBwave}\end{equation}
Note that trivially $L(a_j,b_j) \le |b_j| = |w|$. Hence for an upper bound, we may assume that $w(a_j) < |w|$, and in particular that $a_j$ is a subsequence of $a_j' = x w' y$ for some prefix $x = \sigma_x \dots t$ of $w$ and some suffix $y = 1 \dots \sigma_y$ of $w$ with $\sigma_y \le \sigma_x$, where $|x| + |y| = w(a_j)$. Note that any longest common subsequence $z$ of $a_j'$ and $b_j=w=1\dots t$ is an increasing subsequence of $a_j'$. Hence, if $z$ starts with a symbol $\sigma' \ge \sigma_y$, then $z$ is an increasing subsequence in $x \; t' \dots \sigma'$; easy inspection shows that in this case $|z| \le \max\{ |x|, 1\}$. If $z$ starts with a symbol $\sigma' \le \sigma_x$, then $z$ is an increasing subsequence in $\sigma' \dots 1 \; y$; again, one can see that $|z| \le \max\{|y|,1\}$ holds in this case. Thus, $L(a_j,b_j) \le L(a_j',b_j) = |z| \le \max\{|x|,|y|, 1\} \le \max\{|x|+|y|,1\} = \max\{w(a_j), 1\}$, concluding the proof of~\eqref{eq:LUBwave}.

Let $J = \{j \mid w(a_j) \ge 1\}$. We compute
\begin{eqnarray*}
 L(a',b') = \sum_{j=1}^{R+k} L(a_j,b_j) & \le & \left(\sum_{j\in J} \min\{w(a_j),|w|\}\right) + (R+k - |J|) \\
& \le & \min\left\{\sum_{j=1}^{R+k} w(a_j), |J|\cdot |w|\right\} + (R+k - |J|) \\
& \le & \min\{ R \cdot  |w|, |J|\cdot |w| \} + R+k - |J| \le R|w| + k = Rt+k,
\end{eqnarray*}
where the last inequality follows from the observation that $|J| = R$ maximizes the expression $\min\{R\cdot |w|,|J|\cdot|w|\} - |J|$. 
This finishes the proof of \itemref{genDomPairs-largeSigma-basic} and thus the lemma.
\end{proof}

\subsection{Block elimination and dominant pair reduction}

We collect some convenient tools for the analysis of later constructions. The first allows us to ``eliminate'' $0^\ell$-blocks when computing the LCS of strings of the form $x0^\ell y, 0^\ell z$, provided that $\ell$ is sufficiently large.

\begin{lem} \label{lem:zeroblocklcs}
  For any strings $x,y,z$ and $\ell \ge \occ_0(x) + |z|$ we have
  $L( x 0^\ell y, 0^\ell z ) = \ell + L( 0^{\occ_0(x)} y, z )$. 
\end{lem}
\begin{proof}
  Let $u := x 0^\ell y$ and $v := 0^\ell z$. 
  In case we match no symbol in the $0^\ell$-block of $v$ with a symbol in $0^\ell y$ in $u$, then at most $\min\{\occ_0(x),\ell\}$ symbols of the $0^\ell$-block of $v$ are matched. The remainder $z$ yields at most $|z|$ matched symbols.
  Otherwise, in case we match any symbol in the $0^\ell$-block of $v$ with a symbol in $0^\ell y$ in $u$, then no symbol $\sigma \ne 0$ of $x$ can be matched. Thus, in this case we may replace $x$ by $0^{\occ_0(x)}$. 
  Together this case distinction yields
  $L( u,v ) = \max\{ \min\{\occ_0(x),\ell\} + |z|, L(0^{\occ_0(x) + \ell} y, 0^\ell z) \}$. 
  Using \lemref{greedy}, we obtain
  $L(u,v) = \max\{ \occ_0(x) + |z|, \ell + L(0^{\occ_0(x)} y, z) \}$. The assumption $\ell \ge \occ_0(x) + |z|$ now yields the claim.
\end{proof}

The following lemma bounds the number of dominant pairs of strings of the form $x'=yx$, $y'=zy$ by $d(x',y') = \Oh(|z|\cdot |y'|)$. If $|x'|\ge |y'|$, this provides a bound of $\Oh(\delta(x',y') \cdot m(x',y'))$ instead of the general, weaker bound $\Oh(\Delta(x',y') \cdot m(x',y'))$ of \lemref{dLgap}.

\begin{lem} \label{lem:domPairRedbase}
For any strings $x,y,z$, let $x' = yx$, $y'=zy$. Then 
\[d(x',y') \le |y|\cdot (|z|+1) + d(x',z) \le |y|\cdot (|z| + 1) + |z|^2. \]
\end{lem}
\begin{proof}
For every prefix $\py=y'[1..j]$, we bound the number of dominant pairs $(i,j)$ of $x',y'$. Clearly, all prefixes $\py$ of $z$ (i.e., $j \le |z|$) contribute $d(x',z) \le L(x',z) \cdot m(x',z) \le |z|^2$ dominant pairs.

It remains to consider $\py = z \,y[1..\ell]$ (i.e., $j = |z| + \ell$) for $\ell \in [|y|]$. For $i < \ell$, the string $\px = x'[1..i] = y[1..i]$ is a subsequence of $\py$, i.e., $L(\px,\py) = i$, but the prefix $z \, y[1..i]$ of $\py$ already satisfies $L(\px, z \, y[1..i]) = i$. Hence, there are no dominant pairs with $i < \ell$.
Thus, consider $i \ge \ell$ and let $\px = x'[1..i]$. Clearly, $y[1..j]$ is a common subsequence of $\px,\py$. This yields $L(\px,\py) \ge j = |\py|-|z|$ and hence any such dominant pair $(i,j)$ satisfies $j - |z| \le L(x'[1..i],y'[1..j]) \le j$. By \obsref{atMostOned}, there are at most $|z|+1$ such dominant pairs for fixed $j$. This yields at most $|y|\cdot (|z|+1)$ dominant pairs $(i,j)$ with $|z| < j \le |y'|$, concluding the proof.
\end{proof}

\begin{figure}
\centering
\small
\input{pics/domRed.tikz}
\caption{Illustration of \lemref{dreduction}. The strings $x' = y2^\ell x$,$y'=2^\ell y$ are defined using $x=(01)^{R+S}$, $y=0^R (01)^S$, $R=4$, $S=5$, and $\ell=2$. The number of dominant pairs is significantly reduced compared to \figref{domPairsI}.}
\label{fig:domPairsRed}
\end{figure}

The above lemma gives rise to a surprising technique: Given strings $x,y$, we can build strings $x',y'$ such that $L(x',y')$ lets us recover $L(x,y)$, but the number of dominant pairs may be reduced significantly, namely to a value $d(x',y') = \Oh(\delta(x,y) \cdot n(x,y))$,  independently of $d(x,y)$. The effect of this technique is illustrated in \figref{domPairsRed}.

\begin{lem}[Dominant Pair Reduction] \label{lem:dreduction}
  Consider strings $x,y$ and a number $\ell > |y|-L(x,y)$. 
  \begin{enumerate}
  \item[(i)] If $2$ is a symbol not appearing in $x,y$, then $x' := y 2^\ell x$ and $y' := 2^\ell y$ satisfy
  $L(x',y') = L(x,y) + \ell$ and $d(x,y) \le 3 \ell \cdot |y|$. 
  \item[(ii)] For any symbols $0,1$ (that may appear in $x,y$) set $x'' := 0^k 1^k y 1^\ell 0^k 1^k x$ and $y'' := 1^\ell 0^k 1^k y$ with $k := 2|y| + |x| + 1$. Then $L(x'',y'') = L(x,y) + \ell + 2k$ and $d(x,y) \le \Oh(\ell(|x| + |y| + \ell))$.
  \end{enumerate}
\end{lem}
\begin{proof}
  (i) We clearly have $L(x',y') \ge L(2^\ell, 2^\ell) + L(x, y) \ge \ell + L(x,y)$. For the other direction, let $z$ be a common subsequence of $x',y'$. If $z$ contains no $2$, then by inspecting $y'$ we obtain $|z| \le |y| = L(x,y) + (|y|-L(x,y)) < L(x,y) + \ell$, so $z$ is no LCS. Otherwise, if $z$ contains a $2$, then no symbol from the copy of $y$ in $x'$ can be matched by $z$, implying $|z| \le L(2^\ell x, y') = \ell + L(x,y)$.

For the dominant pairs, we apply \lemref{domPairRedbase} to obtain $d(x',y') \le |y| (\ell+1) + d(x', 2^\ell)$. Note that  $d(x',2^\ell) = d(2^\ell,2^\ell) = \ell$, since we can delete all characters different from 2 in $x'$ without affecting the dominant pairs of $x',2^\ell$ and then apply \lemref{greedy}. Thus, $d(x',y') \le |y| (\ell+1) + \ell \le 3 \ell \cdot |y|$.

  (ii) The argument is slightly more complicated when the padding symbols may appear in $x,y$. Clearly, we have $L(x'',y'') \ge L(1^\ell 0^k 1^k x, 1^\ell 0^k 1^k y) \ge \ell + 2k + L(x,y)$. For the other direction, let $z$ be a common subsequence of $x'',y''$. If $z$ does not match any $0$ in the $0^k$-block of $y''$ with a symbol in a $0^k$-block in $x''$, then from the $0^k$-block of $y''$ we match at most $|y|+|x|$ symbols, and we obtain $|z| \le (|y| + |x|) + |1^\ell 1^k y| = |x| + 2|y| + \ell + k < \ell + 2k$, since $k > |x| + 2|y|$, so $z$ is no longest common subsequence. If $z$ matches a $0$ in the $0^k$-block of $y''$ with a symbol in the left $0^k$-block of $x''$, then no symbol in the $1^\ell$-block in $y''$ is matched by $z$, so we obtain $|z| \le |0^k 1^k y| = 2k + L(x,y) + (|y| - L(x,y)) < 2k + L(x,y) + \ell$, so $z$ is no longest common subsequence. It remains the case where $z$ matches some $0$ in the $0^k$-block of $y''$ with a symbol in the right $0^k$-block of $x''$. Then the part $1^k y$ of $y''$ has to be matched to a subsequence of $0^k 1^k x$ in $x''$. This yields $|z| \le \ell + k + L(0^k 1^k x, 1^k y)$. Since $k > |y|$ we can apply \lemref{zeroblocklcs} (swapping the roles of 0 and 1) to obtain $L(0^k 1^k x, 1^k y) = k + L(x,y)$, so as desired we have $|z| \le \ell + 2k + L(x,y)$.

For the dominant pairs, we apply \lemref{domPairRedbase} to obtain $d(x'',y'') \le |0^k 1^k y| \cdot  (\ell + 1) + \ell^2 = \Oh(\ell (|x|+|y|+\ell))$.
\end{proof}

\section{Paddings}
\label{sec:paddings}

In this section we construct paddings that allow to augment any strings from $\LCS_\le(\Valpha)$ to become strings in $\LCS(\Valpha)$. Specifically, we prove \lemref{paddingstotal}. So let $\Valpha$ be a parameter setting satisfying \tabref{paramChoiceRestr}, let $p \in \params^* = \{n,m, L, \delta, \Delta, |\Sigma|, M, d\}$ be a parameter, and let $n \ge 1$. 
We say that strings $x,y$ \emph{prove \lemref{paddingstotal} for parameter $p$} if $(n,x,y)$ is an instance of $\LCS_\le(\Valpha)$ with $p(x,y) = \Theta(n^{\alpha_p})$, and given $n$ we can compute $x=x(n)$, $y=y(n)$, and $L(x,y)$ in time $\Oh(n)$. Note that for the first requirement of being an instance of $\LCS_\le(\Valpha)$, we have to show that $p'(x,y) \le \Oh(n^{\alpha_{p'}})$ for any parameter $p' \in \params^*$. Recall that we write $p = n^{\alpha_p}$ for the target value of parameter $p$. 

\begin{lem} \label{lem:padL}
  Let $\Sigma'$ be an alphabet of size $\min\{|\Sigma|,L\}$. Then the strings $x := y := \bigconcat_{\sigma \in \Sigma'} \sigma^{\lfloor L/|\Sigma'| \rfloor}$ prove \lemref{paddingstotal} for parameter $L$.
\end{lem}
\begin{proof}
  Note that $\lfloor L/|\Sigma'| \rfloor = \Theta(L / |\Sigma'|)$, since $|\Sigma'| \le L$.
  Thus, indeed $L(x,y) = |x| = \Theta(L/|\Sigma'|) \cdot |\Sigma'| = \Theta(L)$. Moreover, $L(x,y)$ can be computed in time $\Oh(n)$, as well as $x$ and $y$. For the number of matching pairs we note that $M(x,y) \le |\Sigma'| \cdot (L/|\Sigma'|)^2$, which is $\max\{L, L^2 / |\Sigma|\}$ by choice of $|\Sigma'|$. This is $\Oh(M)$, using the parameter relations $M \ge n \ge m \ge L$ and $M \ge L^2 / |\Sigma|$ (see \tabref{paramRelations}).
  
  The remaining parameters are straight-forward. Using $m(x,y) = n(x,y) = L(x,y) = \Theta(L)$ and the parameter relations $L \le m \le n$ we obtain that $m(x,y) \le \Oh(m)$ and $n(x,y) \le \Oh(n)$. Moreover, $\delta(x,y) = \Delta(x,y) = 0 \le \delta \le \Delta$. The alphabet size $|\Sigma(x,y)| = |\Sigma'|$ is at most $|\Sigma|$ by choice of $|\Sigma'|$. By \lemref{greedy} we obtain $d(x,y) = |x| = \Theta(L) \le \Oh(d)$ using the parameter relation $L \le d$.
\end{proof}

\begin{lem} \label{lem:padDelta}
  The strings $x := 1^{\Delta+1}$ and $y := 1$ prove \lemref{paddingstotal} for parameter $\Delta$. The strings $x := 1$ and $y := 1^{\delta+1}$ prove \lemref{paddingstotal} for parameter $\delta$.
\end{lem}
\begin{proof}
  The analysis is straight-forward. Note that indeed $\Delta(x,y) = \Delta$, and that $L(x,y) = 1 \le \Oh(L)$. We have $n(x,y) = \Delta+1 \le \Oh(n)$ and $m(x,y)=1 \le \Oh(m)$. Clearly, $L(x,y)$, $x$, and $y$ can be computed in time $\Oh(n)$. Moreover, $\delta(x,y) = 0 \le \delta$ and the alphabet size is $1 \le \Oh(\Sigma)$. Finally, we have $M(x,y) = \Theta(\Delta) \le \Oh(n) \le \Oh(M)$ using the parameter relations $L \le n \le M$, and using the relation $d \le L m$ we obtain $d(x,y) \le 1 \le \Oh(d)$.
  
  The analysis for $\delta$ is symmetric; the same proof holds almost verbatim.
\end{proof}

\begin{lem}
  The strings constructed in \lemref{padL} or the strings constructed in \lemref{padDelta} prove \lemref{paddingstotal} for parameters $n$ and $m$. 
\end{lem}
\begin{proof}
  Since $n = L + \Delta$ we have $L = \Theta(n)$ or $\Delta = \Theta(n)$, i.e., $\alpha_L=1$ or $\alpha_\Delta=1$. In the first case, in \lemref{padL} we construct strings of length $\Theta(L) = \Theta(n)$, and thus these strings prove \lemref{paddingstotal} not only for parameter $L$ but also for parameter $n$. In the second case, the same argument holds for the first pair of strings constructed in \lemref{padDelta}. The parameter $m$ is symmetric.
\end{proof}

\begin{lem}
  Let $w := 1 2 \ldots |\Sigma|$ be the concatenation of $|\Sigma|$ unique symbols. The strings $w,w$ or the strings $w,\rev(w)$ prove \lemref{paddingstotal} for parameter $|\Sigma|$. 
\end{lem}
\begin{proof}
  Clearly, both pairs of strings realize an alphabet of size exactly $|\Sigma|$. Since $m = L + \delta$ we have $L = \Theta(m)$ or $\delta = \Theta(m)$. In the first case, we use $L(w,w) = |\Sigma| \le m = \Theta(L)$ and $\delta(w,w) = \Delta(w,w) = 0 \le \delta \le \Delta$. In the second case, we have $L(w,\rev(w)) = 1 \le \Oh(L)$ and $\delta(w,\rev(w)) = \Delta(w,\rev(w)) = |\Sigma|-1 \le m = \Theta(\delta) \le \Oh(\Delta)$.
  
  The remaining parameters are straight-forward. Let $(x,y) \in \{(w,w), (w,\rev(w))\}$. We have $n(x,y) = m(x,y) = |\Sigma| \le m \le n$. Moreover, $d(x,y) \le M(x,y) = |\Sigma| \le d \le M$ using the relations $|\Sigma| \le d \le M$.
  Clearly, the strings and their LCS length can be computed in time $\Oh(n)$.
\end{proof}

\subsection{Matching Pairs}

\begin{lem}
  If $\alpha_\Delta = 1$ then $x := 1^{\lfloor M/n \rfloor + \Delta}$ and $y := 1^{\lfloor M/n \rfloor}$ prove \lemref{paddingstotal} for parameter $M$. 
\end{lem}
\begin{proof}
  Note that $\lfloor M/n \rfloor = \Theta(M/n)$ by the parameter relation $M \ge n$.
  We have $M(x,y) = \Theta((M/n)^2 + \Delta M/n)$. By the parameter relations $M \le 2 L n \le 2 n^2$ the first summand is $\Oh(n \cdot M/n) = \Oh(M)$. Since $\alpha_\Delta = 1$, the second summand is $\Theta(M)$. Thus, we indeed have $M(x,y) = \Theta(M)$.
  
  The remainder is straight-forward. Clearly, $x$, $y$, and $L(x,y) = \lfloor M/n \rfloor$ can be computed in time $\Oh(n)$. Since $M \le 2 L n$ we also obtain $L(x,y) = m(x,y) = \lfloor M/n \rfloor \le \Oh(L) \le \Oh(m)$. Moreover, $n(x,y) = \lfloor M/n \rfloor + \Delta \le \Oh(n)$ by the relations $M/n \le 2L \le 2n$ and $\Delta \le n$. Note that $\Delta(x,y) = \Delta$ and $\delta(x,y) = 0 \le \delta$. The alphabet size is 1. By \lemref{greedy} we have $d(x,y) = \lfloor M/n \rfloor \le 2L \le 2 d$.
\end{proof}

\begin{lem}
  Assume $\alpha_\Delta < 1$ and let $\Sigma'$ be an alphabet of size $\min\{ \lceil m^2 /M \rceil, |\Sigma| \}$. Then $x := y := \bigconcat_{\sigma \in \Sigma'} \sigma^{\lfloor m/|\Sigma'| \rfloor}$ prove \lemref{paddingstotal} for parameter $M$. 
\end{lem}
\begin{proof}
  Observe that $\alpha_\Delta < 1$ implies $\alpha_L = \alpha_m = 1$, so that $n = \Theta(L) = \Theta(m)$ (see \tabref{paramChoiceRestr}). The number of matching pairs is $M(x,y) = |\Sigma'| \cdot \lfloor m/|\Sigma'| \rfloor^2$. By the parameter relation $m \ge |\Sigma|$ and $|\Sigma| \ge |\Sigma'|$ we have $\lfloor m/|\Sigma'| \rfloor = \Theta(m/|\Sigma'|)$, and by $M \le 2 L n = \Theta(m^2)$ we obtain $\lceil m^2/M \rceil = \Theta(m^2/M)$. Thus, $M(x,y) = \Theta(m^2 / |\Sigma'|) = \Theta(\max\{M,m^2 / |\Sigma|\})$ by choice of $|\Sigma'|$. Using $m = \Theta(L)$ and the parameter relation $M \ge L^2 / |\Sigma|$ we indeed obtain $M(x,y) = \Theta(M)$.
  
  The remainder is straight-forward. Since $x=y$ we have $L(x,y) = m(x,y) = n(x,y) = |\Sigma'| \cdot \lfloor m/|\Sigma'| \rfloor \le m = \Theta(L) = \Theta(n)$. Moreover, $\delta(x,y) = \Delta(x,y) = 0 \le \delta \le \Delta$. The alphabet size is $|\Sigma(x,y)| = |\Sigma'| \le |\Sigma|$ by choice of $|\Sigma'|$. By \lemref{greedy} we have $d(x,y) = L(x,y) \le \Oh(L) \le \Oh(d)$. Clearly, $x$, $y$, and $L(x,y)$ can be computed in time $\Oh(n)$.
\end{proof}

\subsection{Dominant Pairs}

We start with a simple construction that always works on constant-sized alphabets ($\alpha_\Sigma = 0$).

\begin{lem} \label{lem:paddsmallS}
Assume $\alpha_d \le 2\alpha_L \le \alpha_M$ and set $x:= (01)^{R+S}$ and $y:= 0^R (01)^{R+S}$ (as analyzed in \lemref{genDomPairs}), instantiated with $R=\lfloor \min\{\Delta,\sqrt{d}\} \rfloor$, $S=\lceil d/R \rceil$. Then $x,y$ prove \lemref{paddingstotal} for parameter~$d$. 
\end{lem}
\begin{proof}
Note that by definition $R\le \sqrt{d}$, and hence $S\ge d/R \ge R$. By \lemref{genDomPairs}\itemref{genDomPairsd}, we obtain $d(x,y)= \Theta(R\cdot S) = \Theta(R \cdot d/R) = \Theta(d)$. For the other parameters, note that $n(x,y) = 2(R+S) = \Oh(d/R) = \Oh(d/\Delta+\sqrt{d})$. By the relation $d\le 2L(\Delta+1)$, we have $d/\Delta \le \Oh(L)$, and by the assumption $\alpha_d\le 2\alpha_L$, we have $d \le \Oh(L^2)$ and hence $\sqrt{d} \le \Oh(L)$. Thus, $n(x,y) \le \Oh(L)$.

The remainder is straight-forward. By $L(x,y)\le m(x,y) \le n(x,y) \le \Oh(L) \le \Oh(m) \le \Oh(n)$ we have verified $n,m,L$. Consequently, also $M(x,y) \le n(x,y)^2 \le \Oh(L^2) \le \Oh(M)$ by the assumption $2\alpha_L \le \alpha_M$. Trivially, $|\Sigma(x,y)|=2=\Oh(|\Sigma|)$. Observing that $\delta(x,y) = 0 \le \Oh(\delta)$ and $\Delta(x,y) = R \le \Oh(\Delta)$ by \lemref{genDomPairs}\itemref{genDomPairsL} concludes the parameter verification. Since $x,y$ and $L(x,y)=R+2S$ (by \lemref{genDomPairs}\itemref{genDomPairsL}) can be computed in time $\Oh(n)$, the claim follows.
\end{proof}

The construction above creates a long LCS of length $L(x,y) = \Theta(m(x,y))$ which forces $d(x,y)=\Oh(L(x,y)^2)$. With super-constant alphabet sizes, one can construct larger numbers of dominant pairs (compared to $L(x,y)$) by exploiting the crossing gadgets defined in \defref{crossinggadget}.

\begin{lem}
  Assume $\alpha_d > 2 \alpha_L$ and set $v := (01)^{R+S}$ and $w := 0^{R} (01)^{S}$ with $R=S=L$. Construct $x := \cross^\x(v,\ldots,v)$ and $y := \cross^\y(w,\ldots,w)$ on $\lfloor d/L^2 \rfloor$ copies of $v$ and $w$. Then $x,y$ prove \lemref{paddingstotal} for parameter $d$. 
\end{lem}
\begin{proof}
  Note that $\lfloor d/L^2 \rfloor = \Theta(d/L^2)$ since we assume $\alpha_d > 2 \alpha_L$. By the Crossing Alphabets Lemma (\lemref{crossingalphabet}), we obtain $L(x,y) = L(v,w)=3L$, and in particular $L(x,y)$, $x$, and $y$ can be computed in time $\Oh(n)$. 

   Furthermore, the Crossing Alphabets Lemma also yields $d(x,y) = \Theta(d/L^2) \cdot d(v,w) = \Theta(d)$, where the bound $d(v,w)= \Theta(L^2)$ follows from \lemref{genDomPairs}\itemref{genDomPairsL}. Similarly, we observe that $\Delta(x,y) \le n(x,y) = \Theta(d/L^2) \cdot n(v,w) = \Theta(d/L)$, which is at most $\Oh(\Delta) \le \Oh(n)$ by the parameter relation $d \le 2 L (\Delta+1)$. Likewise, $m(x,y) \le n(x,y) = \Oh(d/L) \le \Oh(m)$ by the parameter relation $d \le Lm$. Moreover, $M(x,y) = \Oh(d/L^2) \cdot M(v,w) = \Oh(d) \le \Oh(M)$ by $d\le M$. Finally, the assumption $2\alpha_L < \alpha_d$ together with the parameter relation $d\le Lm$, i.e., $\alpha_d \le \alpha_L + \alpha_m$, forces $\alpha_L < \alpha_m$. Hence, $\alpha_\delta = \alpha_m$, i.e.,  $\delta = \Theta(m)$ (see \tabref{paramChoiceRestr}), and thus $\delta(x,y)\le m(x,y) \le \Oh(m)= \Oh(\delta)$.
   Since $v$ and $w$ have alphabet size 2 and we use $\lfloor d/L^2 \rfloor$ copies over disjoint alphabets, we have $|\Sigma(x,y)| = 2 \lfloor d/L^2 \rfloor \le \Oh(|\Sigma|)$ by the parameter relation $d \le L^2 |\Sigma|$, which concludes the proof. 
\end{proof}

For super-constant alphabet sizes, the number of matching pairs $M(x,y)$ can attain values much smaller than $L(x,y)^2$, which is an orthogonal situation to the lemma above. In this case, we use a different generalization of the first construction (that we already prepared in \secref{basicFacts}).

\begin{lem} \label{lem:paddlargeS}
Assume $\alpha_M < 2 \alpha_L$ and set $x := (1\dots t\;t'\dots 1)^R (1\dots t)^{S-R}$ and $y:= (1\dots t)^S$ (as analyzed in \lemref{genDomPairs-largeSigma}) instantiated with 
 \begin{align*}
t  & := \Big\lfloor \frac{L^2}{M} \Big\rfloor, & t' &:= \min\{r,t\} & R & := \Big\lceil \frac{r}{t} \Big\rceil & S & := 4 \Big\lceil \frac{d}{rt} \Big\rceil,
\end{align*}
  where $r := \min\{\Delta, \lfloor \sqrt{d/t} \rfloor \}$.
 Then $x,y$ prove \lemref{paddingstotal} for parameter~$d$.
\end{lem}
\begin{proof}
We first verify the conditions of \lemref{genDomPairs-largeSigma}.
Observe that by the assumption $\alpha_M < 2\alpha_L$ we indeed have $t= \Theta(L^2/M)$ and $t\ge 2$ (for sufficiently large $n$). From the parameter relation $M \ge L^2 / |\Sigma|$ we obtain $t \le |\Sigma|$, and from the parameter relation $d \ge |\Sigma|$ (and $\alpha_\Delta > 0$) this yields $r \ge 1$. Thus, $1 \le t' \le t$.
Moreover, $r = \Theta(\min\{\Delta, \sqrt{d/t}\})$. Observe that $r \le Rt' \le 2r$. Indeed, if $r \le t$ then $R = 1$ and $t' = r$, and if $r > t$ then $r/t \le R \le 2r/t$ and $t' = t$. In particular, we have
$$Rt' = \Theta\big(\min\{\Delta, \sqrt{d/t}\}\big) \quad \text{and} \quad S = \Theta\Big(\frac{d}{Rt'\cdot t}\Big).$$
Note that $R(t'+1) \le 2Rt' \le 4r \le S$, since $r \le \sqrt{d/t}$. In particular, this yields $1 \le R \le S$, so that all conditions of \lemref{genDomPairs-largeSigma} are satisfied. 

In the remainder we show that $x,y$ satisfy the parameter constraints. We have $n(x,y) \le (R+S)t = \Oh(St) = \Oh(d/(Rt')) = \Oh(d/\Delta + \sqrt{dt})$. Note that $d/\Delta = \Oh(L)$ by the parameter relation $d\le 2L(\Delta+1)$, and that $\sqrt{dt} = \Oh(\sqrt{dL^2/M}) = \Oh(L)$ by the parameter relation $d\le M$. Thus, $L(x,y) \le m(x,y)\le n(x,y) \le \Oh(L) \le \Oh(m) \le \Oh(n)$, which satisfies the parameters $L,m,n$.

For $d$, note that by \lemref{genDomPairs-largeSigma}\itemref{genDomPairs-largeSigma-d}, we have $d= \Theta((Rt')\cdot (St)) = \Theta((Rt') \cdot d/(Rt')) = \Theta(d)$. For~$M$, \lemref{genDomPairs-largeSigma}\itemref{genDomPairs-largeSigma-M} shows that $M(x,y) = \Oh(S^2 t) = \Oh( (d/(Rt'))^2 \cdot (1/t)) = \Oh( L^2 \cdot (M/L^2) ) = \Oh(M)$, where we used $d/(Rt') = \Oh(L)$ as shown above. Since $L(x,y) = |b|=St$, we obtain $\delta(x,y) = 0 \le \Oh(\delta)$ and $\Delta(x,y) = Rt' \le \Oh(\Delta)$. Finally, $|\Sigma(x,y)| = t = \Theta(L^2/M) \le \Oh(|\Sigma|)$ follows from the parameter relation $M\ge L^2/|\Sigma|$.
Observing that $x,y$, and $L(x,y)=St$ can be computed in time $\Oh(n)$ concludes the proof.
\end{proof}

\section{Hardness for Large Alphabet}
\label{sec:hardness}

In this section, we consider a parameter setting $\Valpha$ satisfying the relations of \tabref{paramChoiceRestr}, and we prove a lower bound for $\LCS_\le(\Valpha)$ assuming \OVH, thus proving \thmref{hardness}. We split our proof into the two cases $\alpha_\delta = \alpha_m$ (where $L$ may be small) and $\alpha_L = \alpha_m$ (where $L$ is large). For readability, but abusing notation, for the target value $\lceil n^{\alpha_p} \rceil$ of parameter $p$ we typically simply write~$p$. 

In this section we can assume that 
\begin{equation} \label{eq:lowerassumptions}
  \alpha_L,\alpha_m,\alpha_\delta,\alpha_\Delta > 0 \quad \text{and} \quad \alpha_d > 1,  \tag{LB}
\end{equation}
since otherwise the known $\tOh(n + \min\{d, \delta m, \delta \Delta\})$ algorithm runs in (near-)optimal time $\tOh(n)$ and there is nothing to show (here we used the parameter relations $d \le L m$ and $L,m,\delta,\Delta \le n$).

\subsection{Small LCS}
\label{sec:hardnesssmallLCS}

Assume $\alpha_\delta = \alpha_m$, i.e., $\delta = \Theta(m)$. In this case, the longest common subsequence might be arbitrarily small, i.e., any value $0 < \alpha_L \le \alpha_m$ is admissible.  

\subsubsection{Hard Core}
At the heart of our constructions lies the previous reduction from \OV\ to \LCS of \cite{BringmannK15}, which we restate here for our purposes.

\begin{lem}\label{lem:core}
Let two sets $\sA=\{a_1,\dots,a_A\}$ and $\sB=\{b_1,\dots,b_B\}$ of vectors in $\{0,1\}^D$ with $A\ge B$ be given. In time $\Oh(AD)$, we can construct strings $x_1,\dots,x_{2A}$ and $y_1,\dots,y_B$ over $\{0,1\}$ and $\gamma,\gamma' = \Theta(D)$ such that the strings $x$ and $y$ defined by
\begin{alignat*}{9}
x &:=\qquad \, && x_1 \; && 0^{\gamma} \; && x_2 \; && 0^{\gamma} \,\ldots\; && x_{2A-1} \; && 0^{\gamma} \; && x_{2A}&&,  \\
y &:= \, 0^{A\gamma'} \; && y_1 \; && 0^{\gamma} \; && y_2 \; && 0^{\gamma} \,\ldots\; && y_{B-1} \; && 0^{\gamma} \; && y_B \; && 0^{A\gamma'},
\end{alignat*}
satisfy the following properties:
\begin{enumerate}[label=(\roman{*})]
\item\label{itm:core-L} We can compute some $\rho$ in time $\Oh(AD)$ such that $L(x,y) \ge \rho$ if and only if there is a pair $i,j$ with $\langle a_i, b_j\rangle = 0$.
\item\label{itm:core-length} $|x|,|y| = \Oh(AD)$.
\item\label{itm:core-numones} $\occ_1(y) = \Oh(BD)$.
\item\label{itm:core-initialzeroes} For all $\beta \ge 0$, we have $L(x,0^\beta y) = L(x,y)$.
\end{enumerate}
\end{lem}
\begin{proof}
Claims \itemref{core-L} and \itemref{core-length} are a restatement of the main result in~\cite{BringmannK15}, instantiated for LCS. Claim \itemref{core-numones} follows directly from the construction. It is easy to check that in~\cite{BringmannK15} the value of $\gamma'$ is chosen large enough to satisfy $A \gamma' \ge \occ_0(x)$. This yields claim \itemref{core-initialzeroes}, since any common subsequence $z$ of $x,0^\beta y$ starts with at most $\occ_0(x)$ symbols $0$, which can all be matched to the $0^{A \gamma'}$-block of $y$, so $z$ is also a common subsequence of $x,y$, and we obtain $L(x,y) = L(x,0^\beta y)$.
\end{proof}

\subsubsection{Constant Alphabet}
First assume $\alpha_\Sigma = 0$ and thus $|\Sigma| = \Oh(1)$.
Consider any $n \ge 1$ and target values $p = n^{\alpha_p}$ for $p \in \params$. 
Let $\sA= \{a_1,\dots,a_A\}$, $\sB=\{b_1,\dots,b_B\} \subseteq \{0,1\}^D$ be a given \OV instance with $D=n^{o(1)}$ and where we set $A := \lfloor L/D \rfloor$ and $B := \lfloor d/(LD) \rfloor$.
Note that $A = n^{\alpha_{L} - o(1)} = n^{\Omega(1)}$ and $B=n^{\alpha_d-\alpha_L-o(1)} = n^{\Omega(1)}$ by (\ref{eq:lowerassumptions}). Also note that \UOVH implies that solving such \OV instances takes time $(AB)^{1-o(1)} = n^{\alpha_d - o(1)} = d^{1-o(1)}$. 

Construct strings $x,y$ as in \lemref{core}. Then from the LCS length $L(x,y)$ we can infer whether $\sA,\sB$ has an orthogonal pair of vectors by \lemref{core}(i). Moreover, this reduction runs in time $\Oh(A D) = \Oh(L) \le \Oh(d^{1-\eps})$ for sufficiently small $\eps > 0$ (since $\alpha_L \le \alpha_m \le 1 < \alpha_d$ by \tabref{paramChoiceRestr} and (\ref{eq:lowerassumptions})). We claim that $(n,x,y)$ is an instance of $\LCS_\le(\Valpha)$. This shows that any algorithm solving $\LCS_\le(\Valpha)$ in time $\Oh(d^{1-\eps})$ implies an algorithm for our \OV instances with running time $\Oh(d^{1-\eps})$, contradicting \UOVH. Hence, in the current case $\alpha_\delta = \alpha_m$ and $\alpha_\Sigma = 0$, any algorithm for $\LCS_\le(\Valpha)$ takes time $d^{1-o(1)}$, proving part of \thmref{hardness}.

It remains to show the claim that $(n,x,y)$ is an instance of $\LCS_\le(\Valpha)$. Using the parameter relations $L \le m \le n$, \lemref{core}(ii), and the definition of $A$, we have $L(x,y) \le m(x,y) \le n(x,y) = |x| \le \Oh(AD) = \Oh(L) \le \Oh(m) \le \Oh(n)$, so indeed $p(x,y) \le \Oh(p) = \Oh(n^{\alpha_p})$ for $p \in \{L,m,n\}$. Similarly, we obtain $\delta(x,y) \le \Delta(x,y) \le n(x,y) \le \Oh(m) = \Oh(\delta) \le \Oh(\Delta)$, where the equality holds in the current case $\alpha_\delta = \alpha_m$. Since $x,y$ use the binary alphabet $\{0,1\}$, we have $|\Sigma(x,y)| = 2 \le \Oh(n^{\alpha_\Sigma})$. For the number of matching pairs we have $M(x,y) \le n(x,y)^2 = \Oh((AD)^2) = \Oh(L^2)$. Since we are in the case $\alpha_\Sigma = 0$, from the parameter relation $M \ge L^2 / |\Sigma|$ (\lemref{Mbounds}(i)) we obtain $L^2 \le \Oh(M)$ and thus also $M(x,y)$ is sufficiently small. Finally, we use \lemrefs{few1s}{core}(iii) to bound $d(x,y) \le \Oh(L(x,y) \cdot \occ_1(y)) \le \Oh(AD \cdot BD)$, which by definition of $A,B$ is $\Oh(d)$. This proves that $(n,x,y)$ belongs to $\LCS_\le(\Valpha)$. 

We remark that our proof also yields the following lemma, which we will use for small alphabets.

\begin{lem} \label{lem:smallLCShardnessConstantAlphabet}
  Let $\Valpha$ be a parameter setting satisfying \tabref{paramChoiceRestr} with $\alpha_\Sigma = 0$ and $\alpha_\delta = \alpha_m$. There is a constant $\gamma \ge 1$ such that any algorithm for $\LCS_\le^\gamma(\Valpha,\{0,1\})$ takes time $d^{1-o(1)}$, unless \OVH fails. This holds even restricted to instances $(n,x,y)$ of $\LCS_\le^\gamma(\Valpha,\{0,1\})$ with $|x|, |y| \le \gamma \cdot n^{\alpha_L}$ and $\occ_1(y) \le \gamma \cdot n^{\alpha_d - \alpha_L}$ satisfying $L(x,0^\beta y) = L(x,y)$ for all $\beta \ge 0$.
\end{lem}

\subsubsection{Superconstant Alphabet}
To tackle the general case $\alpha_\Sigma \ge 0$ (while still assuming $\alpha_\delta = \alpha_m$), we use the following fact which is similar to the Disjoint Alphabets Lemma.

\begin{lem}[Crossing Alphabets] \label{lem:crossingalphabet}
  Let $\Sigma_1,\ldots,\Sigma_k$ be disjoint alphabets and let $x_i,y_i$ be strings over alphabet $\Sigma_i$. Consider $x := x_1 \ldots x_k$ and $y := y_k \ldots y_1$, i.e., the order in $y$ is reversed. For any parameter $p \in \{ n,m,|\Sigma|,M,d \}$ we have $p(x,y) = \sum_{i=1}^k p(x_i,y_i)$. Moreover, $L(x,y) = \max_{i} L(x_i,y_i)$.
\end{lem}
\begin{proof}
  The statement is trivial for the string lengths $n,m$, alphabet size $|\Sigma|$, and number of matching pairs $M$. For the LCS length $L$ we observe that any common subsequence $z$ that matches a symbol in $\Sigma_i$ cannot match any symbols in other alphabets, which yields $L(x,y) \le \max_{i} L(x_i,y_i)$. Since any common subsequence of $x_i,y_i$ is also a common subsequence of $x,y$, we obtain equality. 
  
  Since every dominant pair is also a matching pair, every dominant pair of $x,y$ stems from prefixes $x_1 \ldots x_{j-1} x'$ and $y_k \ldots y_{j+1} y'$, with $x'$ a prefix of $x_{j}$ and $y'$ a prefix of $y_{j}$ for some $j$. Since $L(x_1 \ldots x_{j-1} x', y_k \ldots y_{j+1} y') = L(x',y')$, we obtain that the dominant pairs of $x,y$ of the form $x_1 \ldots x_{j-1} x'$, $y_k \ldots y_{j+1} y'$ are in one-to-one correspondence with the dominant pairs of $x_{j},y_{j}$. Since these dominant pairs of $x,y$ are incomparable, this yields the claim for parameter $d$.
\end{proof}

We make use of the above lemma by invoking the following construction.

\begin{defn} \label{def:crossinggadget}
  Let $\Sigma_1,\dots,\Sigma_t$ be a collection of disjoint two-element alphabets. For any string $z$ over $\{0,1\}$ and $\Sigma_i$, let $z\lift\Sigma_i$ denote the string $z$ lifted to $\Sigma_i$, i.e., we replace the symbols $\{0,1\}$ in $z$ bijectively by $\Sigma_i$. Then for given $x_1,\dots,x_t,y_1,\dots,y_t\in \{0,1\}^*$ we construct
\begin{align*}
\cross^\x(x_1,\dots,x_t) &:= & x_1 & \lift \Sigma_1 \quad & x_2 & \lift \Sigma_2 & \quad & \dots \quad &  x_t & \lift \Sigma_t, \\
\cross^\y(y_1,\dots,y_t) &:= & y_t & \lift \Sigma_t \quad & y_{t-1} & \lift \Sigma_{t-2} & \quad & \dots \quad & y_1 & \lift \Sigma_1. 
\end{align*}
\end{defn}

We adapt the construction from \lemref{core} using the following trick that realizes an "OR" of $t \le \Oh(\Sigma)$ instances, without significantly increasing the parameters $d$ and $M$.

Consider any $n \ge 1$ and target values $p = n^{\alpha_p}$ for $p \in \params$. 
Let $\sA= \{a_1,\dots,a_A\}$, $\sB=\{b_1,\dots,b_B\} \subseteq \{0,1\}^D$ be a given \OV instance with $D=n^{o(1)}$ and where we set $A = \lfloor \frac{d}{\min\{L,\sqrt{d}\}\cdot D}\rfloor$ and $B=\lfloor \frac{\min\{L,\sqrt{d}\}}{D} \rfloor$. 
Note that \UOVH implies that solving such \OV instances takes time $(AB)^{1-o(1)} = n^{\alpha_d - o(1)} = d^{1-o(1)}$. 
Since clearly $A \ge B$, we can partition $\sA$ into $t := \lceil A/B \rceil$ groups $\sA_1,\ldots,\sA_t$ of size $B$ (filling up the last group with all-ones vectors). Using the relation $d \le L^2 |\Sigma|$ (\lemref{dUBs}(ii)), we obtain $t \le \Oh(d/L^2 + 1) \le \Oh(\Sigma)$. 

For each $i=1,\ldots,t$ we construct strings $x_i$ and $y_i$ for the sets $\sA_i$ and $\sB$ using \lemref{core}. Finally, we set $x := \cross^\x(x_1,\dots,x_t)$ and $y := \cross^\y(y_1,\dots,y_t)$. By the Crossing Alphabets Lemma and \lemref{core}.(i) from $L(x,y)$ we can infer whether $\sA,\sB$ has an orthogonal pair of vectors. We claim that $(n,x,y)$ is an instance of $\LCS_\le(\Valpha)$. This shows that any algorithm solving $\LCS_\le(\Valpha)$ in time $\Oh(d^{1-\eps})$ implies an algorithm for our \OV instances with running time $\Oh(d^{1-\eps})$, contradicting \UOVH. Hence, in the current case $\alpha_\delta = \alpha_m$, any algorithm for $\LCS_\le(\Valpha)$ takes time $d^{1-o(1)}$, proving part of \thmref{hardness}.

It remains to show the claim that $(n,x,y)$ is an instance of $\LCS_\le(\Valpha)$. This is similar to the proof for the case $\alpha_\Sigma = 0$, additionally using the Crossing Aphabets Lemma. Specifically, we obtain $m(x,y) \le n(x,y) = |x| = \sum_{i=1}^t |x_i| \le \Oh(t \cdot BD) = \Oh(AD) = \Oh(\max\{d/L, \sqrt{d}\})$, which is at most $\Oh(m) \le \Oh(n)$ using the parameter relations $d \le L m \le m^2$ (\lemref{dUBs}.(i)).
Similarly, we obtain $\delta(x,y) \le \Delta(x,y) \le n(x,y) \le \Oh(m) = \Oh(\delta) \le \Oh(\Delta)$, where the equality holds in the current case $\alpha_m = \alpha_\delta$. 
For $L$ we obtain $L(x,y) = \max_i L(x_i,y_i) \le |y_i| = \Oh(BD) \le \Oh(L)$. 
Since $t \le \Oh(\Sigma)$ we have $|\Sigma(x,y)| \le \Oh(\Sigma)$. Using the parameter relation $d \le M$ we have $d(x,y) \le M(x,y) = \sum_{i=1}^t M(x_i,y_i) \le t\cdot  |x_i| \cdot |y_i| = t \cdot \Oh((BD)^2) = \Oh(AD \cdot BD) = \Oh(d) \le \Oh(M)$. This proves that $(n,x,y)$ belongs to $\LCS_\le(\Valpha)$.

\subsection{Large LCS}
\label{sec:hardnesslargeLCS}

Now assume $\alpha_L = \alpha_m$, i.e., $L = \Theta(m)$. Then the number of deletions in the shorter string might be arbitrary small, i.e., any value $0 < \alpha_\delta \le \alpha_m$ is admissible.  In this case, the construction of~\cite{BringmannK15} is no longer applicable. The new 1vs1/2vs1 gadgets that we design for constructing hard strings for small $\delta$ can be seen as one of our main contributions.

\subsubsection{Hard Core}

The following lemma (which effectively represents an intermediate step in the proof of~\cite{BringmannK15}) yields the basic method to embed sets of vectors into strings $x$ and $y$.

\begin{lem}\label{lem:coreSingleStrings}
Let two sets $\sA= \{a_1,\dots,a_A\}$ and $\sB= \{b_1,\dots,b_B\}$ of vectors in $\{0,1\}^D$ be given. In time $\Oh((A+B)D)$ we can construct strings $x_1,\dots,x_A$ of length $\ell_\x$ and $y_1,\dots,y_B$ of length $\ell_\y$ over alphabet $\{0,1\}$, as well as integers $\rho_1 < \rho_0$,  such that for all $i\in[A], j\in [B]$ we have
\begin{enumerate}
\item[(i)] $\ell_\y \le \ell_\x \le \Oh(D)$,
\item[(ii)] $L(x_i,y_j) = \rho_0$ if $\langle a_i,b_j \rangle = 0$,
\item[(iii)] $L(x_i,y_j) = \rho_1$ if $\langle a_i,b_j \rangle \ne 0$, and
\item[(iv)] $L(x_i,y_j) > \ell_\y/2$.
\end{enumerate}
\end{lem}
\begin{proof}
We can construct strings $x'_1,\dots, x'_A$ of length $\ell_\x'=\Oh(D)$ and $y'_1,\dots,y'_B$ of length $\ell_\y'=\Oh(D)$ and integers $\rho_1' < \rho_0'$ as in \cite[Claim III.6]{BringmannK15} (using the so called normalized vector gadget) that satisfy $L(x'_i,y'_j) = \rho'_{0}$ if $\langle a_i,b_j \rangle = 0$ and $L(x'_i,y_j')=\rho_1'$ otherwise. To additionally enforce conditions (i) and (iv), we define $x_i := 1^{\ell_\y'} 0^{\ell_\y' + 1} x_i' $ and $y_j := 0^{\ell_\y' + 1} y_j'$. Since $L(x_i,y_j) =L(x'_i,y'_j) + \ell_\y'+1$ by \lemrefs{zeroblocklcs}{greedy}, we thus obtain conditions (ii) and (iii) for $\rho_0 := \rho_0' + \ell_\y' + 1$  and $\rho_1:=\rho_1' + \ell_\y' + 1$. Since by definition $\ell_\y = 2\ell_\y' + 1$ holds, the first condition follows directly and the trivial bound $L(x_i,y_j) \ge \ell_\y' + 1 > \ell_\y/2 $ shows that the last condition is fulfilled. 
\end{proof}

\paragraph{1vs1/2vs1 gadget.} The aim of the following construction is to embed given strings $y_1,\dots, y_Q$ into a string $y$ and strings $x_1, \dots,x_P$ into  $x$, where $P=\Theta(Q)$,  such that in an LCS each $y_j$ is either aligned with a single string $x_i$ or with several strings $x_i, x_{i+1},\dots,x_{i'}$. In the first case, $|y_j| - L(x_i,y_j)$ characters of $y_j$ are not contained in an LCS of $x$ and $y$, while in the second case $y_j$ can be completely aligned. By choosing $P=2Q-N$ for an arbitrary $1\le N \le Q$, it will turn out that the LCS aligns $N$ strings $y_j$ with a single partner $x_i$, and the remaining $Q-N$ strings $y_j$ with two strings $x_i, x_{i+1}$ each. Thus, only $N$ strings $y_j$ are not completely aligned.

To formalize this intuition, let $P\ge Q$. We call a set $\Lambda = \{ (i_1,j_1), \dots, (i_k, j_k) \}$ with $0 \le k\le Q$ and $1\le i_1 < i_2 < \dots < i_k\le P$ and $1\le j_1 \le j_2 \le \dots \le j_k \le Q$ a \emph{(partial) multi-alignment}. Let $\Lambda(j) = \{ i \mid (i,j) \in \Lambda\}$.  We say that every $j\in [Q]$ with $|\Lambda(j)|=k$ is \emph{$k$-aligned}. We will also refer to a 1-aligned $j\in [Q]$ as being \emph{uniquely aligned to $i$}, where $\Lambda(j) = \{i\}$. Every $j\in[Q]$ with $\Lambda(j)=\emptyset$ is called \emph{unaligned}. Note that each $i\in[P]$ occurs in at most one $(i,j)\in \Lambda$. 
 We denote the set of multi-alignments as $\algnMult_{P,Q}$.

We will also need the following specialization of multi-alignments. We call a multi-alignment $\Lambda\in \algnMult_{P,Q}$ a \emph{(1,2)-alignment}, if each $j$ is either 1-aligned or 2-aligned. Let $\algnTwo_{P,Q}$ denote the set of all (1,2)-alignments.

Given strings $x_1,\dots,x_P$ of length $\ell_\x$ and $y_1,\dots,y_Q$ of length $\ell_\y$, we define the \emph{value} $v(\Lambda)$ of a multi-alignment $\Lambda\in\algnMult_{P,Q}$ as $v(\Lambda) = \sum_{j=1}^Q v_j$ where 
\[v_j := \begin{cases} 
0 & \text{ if } j \text{ is unaligned}, \\
L(x_i,y_j) & \text{ if } j \text{ is uniquely aligned to } i, \\
\ell_\y & \text{ if } j \text{ is } k\text{-aligned for } k\ge 2. 
\end{cases} \]

\begin{lem}\label{lem:2vs1base}
Given strings $x_1,\dots,x_P$ of length $\ell_\x$ and $y_1,\dots,y_Q$ of length $\ell_\y$, construct 
\begin{eqnarray*}
x & := & \guard(x_1) \; \guard(x_2) \; \dots \; \guard(x_P), \\
y & := & \guard(y_1) \; \guard(y_2) \; \dots \; \guard(y_Q), 
\end{eqnarray*} 
where $\guard(w) := 0^{\gamma_1} \; 1^{\gamma_2} \; (01)^{\gamma_3} \; w \; 1^{\gamma_3}$ with $\gamma_3 := \ell_\x + \ell_\y$, $\gamma_2 := 8 \gamma_3$ and $\gamma_1 := 6\gamma_2$. 
Then we have 
\begin{equation}\label{eq:2vs1}
 \max_{\Lambda\in \algnTwo_{P,Q}} v(\Lambda) \le L(x,y) - Q (\gamma_1 + \gamma_2 + 3\gamma_3) \le \max_{\Lambda\in \algnMult_{P,Q}} v(\Lambda). 
\end{equation}
\end{lem}
\begin{proof}
For the first inequality of \eqref{eq:2vs1}, let $\Lambda\in \algnTwo_{P,Q}$. For every $y_j$, we define $z_j = \bigconcat_{i\in \Lambda(j)} \guard(x_i)$. Consider a 1-aligned $j$ and let $i\in [P]$ be the index $j$ is uniquely aligned to. We have that $z_j = \guard(x_i) = 0^{\gamma_1} 1^{\gamma_2} (01)^{\gamma_3} x_i 1^{\gamma_3}$ and hence by \lemref{greedy}, we obtain $L(z_j, \guard(y_j)) = \gamma_1 + \gamma_2 + 3\gamma_3 + L(x_i,y_j) = \gamma_1+\gamma_2+3\gamma_3 + v_j$.
Likewise, consider a 2-aligned $j$ and let $i,i'\in [P]$ be such that $\Lambda(j) = \{i,i'\}$. Then $z_j = \guard(x_i) \guard(x_{i'})$. We compute
\begin{eqnarray*}
L(z_j, \guard(y_j)) & = & \gamma_1 + \gamma_2 + 3\gamma_3 + L(x_i 1^{\gamma_3} 0^{\gamma_1} 1^{\gamma_2} (01)^{\gamma_3} x_{i'}, y_j) \\
 & \ge  & \gamma_1 + \gamma_2 + 3\gamma_3 + L((01)^{\gamma_3}, y_j) \\
 & = & \gamma_1 + \gamma_2 + 3\gamma_3 + \ell_\y = \gamma_1 + \gamma_2 + 3\gamma_3 + v_j,
\end{eqnarray*}
where the first line follows from \lemref{greedy}, the second line from monotonicity and the third line from $\gamma_3 \ge \ell_\y = |y_j|$.
Observe that $z_1 z_2 \dots z_Q$ is a subsequence of $x$. We conclude that
\[ L(x,y) \ge \sum_{j=1}^Q L(z_j, \guard(y_j)) = Q(\gamma_1 + \gamma_2 + 3\gamma_3) + \sum_{j=1}^Q v_j. \]

It remains to prove the second inequality of~\eqref{eq:2vs1}. Write $x=z_1 z_2 \dots z_Q$ such that $L(x,y) = \sum_{j=1}^Q L(z_j,\guard(y_j))$. We define a multi-alignment $\Lambda$ by letting $(i,j) \in \Lambda$ if and only if $z_j$ contains strictly more than half of the $0^{\gamma_1}$-block of $\guard(x_i)$. Note that the thus defined set satisfies the definition of a multi-alignment, since no two $z_j$'s can contain more than half of $\guard(x_i)$'s $0^{\gamma_1}$-block and if $(i,j), (i',j')\in \Lambda$, then $j < j'$ implies $i< i'$. It remains to show that $L(z_j, \guard(y_j)) \le \gamma_1 + \gamma_2 + 3\gamma_3 + v_j$ for all $j$ to prove the claim.

In what follows, we use the shorthand $H(w) := 1^{\gamma_2} (01)^{\gamma_3} w 1^{\gamma_3}$. Note that $G(w) = 0^{\gamma_1} H(w)$.
Consider an unaligned $j\in[Q]$. By definition, $z_j$ is a subsequence of $0^{\gamma_1/2} H(x_i) 0^{\gamma_1/2}$ for some $i\in [P]$. We can thus bound (using \lemref{greedy})
$$
L(z_j,\guard(y_j)) \le  L(0^{\gamma_1/2} H(x_i) 0^{\gamma_1/2}, 0^{\gamma_1} H(y_j)) = \frac{\gamma_1}{2} + L(H(x_i) 0^{\gamma_1/2}, 0^{\gamma_1/2} H(y_j)).
$$
By \lemref{zeroblocklcs} with $\ell := \gamma_1/2 \ge 2\gamma_2 + 6 \gamma_3 + \ell_\x + \ell_\y = |H(x_i)| + |H(y_j)| \ge \occ_0(H(x_i)) + |H(y_j)|$, 
$$
L(H(x_i) 0^{\gamma_1/2}, 0^{\gamma_1/2} H(y_j)) = \gamma_1/2 + L(0^{\occ_0(H(x_i))}, H(y_j)) \le \gamma_1/2 + \occ_0(H(y_j)) \le \gamma_1/2 + \gamma_3 + \ell_\y.
$$
Hence, in total we have $L(z_j, \guard(y_j)) \le \gamma_1 + \gamma_3 + \ell_\y  \le \gamma_1 + \gamma_2 + 3\gamma_3 = \gamma_1 + \gamma_2 + 3\gamma_3 + v_j$, as desired.

Consider a $j\in [Q]$ that is uniquely aligned (under $\Lambda$) to some $i$. Then $z_j$ is a subsequence of $0^{\gamma_1/2} H(x_{i-1}) 0^{\gamma_1} H(x_i) 0^{\gamma_1/2}$. Analogously to above we compute 
\begin{eqnarray*}
L(z_j, \guard(y_j)) & \le&  \frac{\gamma_1}{2} + L(H(x_{i-1}) 0^{\gamma_1} H(x_i) 0^{\gamma_1/2}, 0^{\gamma_1/2} H(y_j)) \\
& = & \gamma_1 + L(0^{\occ_0(H(x_{i-1})) + \gamma_1} H(x_i) 0^{\gamma_1/2}, H(y_j))  \\
& = & \gamma_1 + L(0^{\occ_0(H(x_{i-1})) + \gamma_1} 1^{\gamma_2} (01)^{\gamma_3} x_i 1^{\gamma_3} 0^{\gamma_1/2}, 1^{\gamma_2} (01)^{\gamma_3} y_j 1^{\gamma_3}).
\end{eqnarray*}
Using \lemref{zeroblocklcs} with symbol $0$ replaced by $1$ yields, since $\ell := \gamma_2 \ge 3 \gamma_3 + \ell_\y = |(01)^{\gamma_3} y_j 1^{\gamma_3}|$ and $\occ_1(0^{\occ_0(H(x_{i-1})) + \gamma_1}) = 0$, 
$$
L(z_j, \guard(y_j)) \le \gamma_1 + \gamma_2 + L( (01)^{\gamma_3} x_i 1^{\gamma_3} 0^{\gamma_1/2}, (01)^{\gamma_3} y_j 1^{\gamma_3})
= \gamma_1 + \gamma_2 + 2\gamma_3 + L(x_i 1^{\gamma_3} 0^{\gamma_1/2}, y_j 1^{\gamma_3}).
$$
Similarly, using \lemref{zeroblocklcs} with symbol $0$ replaced by $1$ on the reversed strings yields, since $\ell := \gamma_3 \ge \ell_\y = |y_j|$ and $\occ_1(0^{\gamma_1/2}) = 0$, 
$$ L(x_i 1^{\gamma_3} 0^{\gamma_1/2}, y_j 1^{\gamma_3}) = \gamma_3 + L(x_i,y_j). $$
Hence, we obtain the desired $L(z_j, \guard(y_j)) \le \gamma_1 + \gamma_2 + 3\gamma_3 + L(x_i,y_j) = \gamma_1 + \gamma_2 + 3\gamma_3 + v_j$.

It remains to consider $j\in [Q]$ that is $k$-aligned for $k\ge 2$. In this case, the claim follows from the trivial bound $L(z_j, \guard(y_j)) \le |\guard(y_j)| = \gamma_1 + \gamma_2 + 3\gamma_3 + v_j$.

Thus $z_1,\dots,z_Q$ defines a multi-alignment $\Lambda\in \algnMult_{P,Q}$ with 
\[L(x,y) = \sum_{j=1}^Q L(z_j, \guard(y_j)) \le Q(\gamma_1 + \gamma_2 + 3\gamma_3) + v(\Lambda), \]
proving the second inequality of \eqref{eq:2vs1}.
\end{proof}

We can now show how to embed an \OV instance $\sA= \{a_1,\dots,a_A\}, \sB=\{b_1,\dots,b_B\} \subseteq \{0,1\}^D$ with $A\le B$ into strings $x$ and $y$ of length $\Oh(B\cdot D)$ whose LCS can be obtained by deleting at most $\Oh(A\cdot D)$ symbols from $y$. For this we will without loss of generality assume that $A$ divides $B$ by possibly duplicating some arbitrary element of $\sB$ up to $A-1$ times without affecting the solution of the instance.

The key idea is that for any $P$ and $Q=2P-N$ with $N\in\{0,\dots,P\}$, $\algnTwo_{P,Q}$ is non-empty and each $\Lambda\in \algnTwo_{P,Q}$ has exactly $N$ uniquely aligned $j\in [Q]$ and exactly $P-N$ 2-aligned $j\in[Q]$. At the same time each $\Lambda\in \algnMult_{P,Q}$ leaves at least $N$ indices $j\in [Q]$ either unaligned or uniquely aligned. 

\begin{lem}\label{lem:2vs1gadget}
Let $a_1,\dots,a_A,b_1,\dots b_B\subseteq\{0,1\}^D$ be given with $A\mid B$. Construct the corresponding strings $x_1,\dots,x_A$ of length $\ell_\x$, $y_1,\dots,y_B$ of length $\ell_\y \le \ell_\x \le \Oh(D)$, and integers $\rho_0,\rho_1$ as in \lemref{coreSingleStrings} and define
\begin{eqnarray*}
\tilde{x} & := & (\tilde{x}_1,\dots,\tilde{x}_P) = (\overbrace{x_1,\dots,x_A,x_1,\dots,x_A,\dots,x_1,\dots,x_A}^{2\cdot(B/A)+3 \text{ {\rm groups of size }} A}), \\
\tilde{y} & := & (\tilde{y}_1,\dots,\tilde{y}_Q) = (\underbrace{y_1,\dots,y_1}_{A \text{ {\rm copies of }} y_1}, y_1, \dots, y_B, \underbrace{y_1, \dots, y_1}_{A \text{ {\rm copies of }}y_1}),
\end{eqnarray*}
where $P:= 2B+3A$ and $Q:= B+2A$. 
Then the instance $x := \bigconcat_i \guard(\tilde{x}_i)$, $y:=\bigconcat_j \guard(\tilde{y}_j)$ of~\lemref{2vs1base} (with the corresponding choice of $\gamma_1,\gamma_2$ and $\gamma_3$) satisfies the following properties:
\begin{itemize}
\item[(i)] For every $i\in [A], j \in [B]$, there is a (1,2)-alignment $\Lambda\in \algnTwo_{P,Q}$ such that some $\ell\in[Q]$ is uniquely aligned to some $k\in[P]$ with $\tilde{x}_k = x_i$ and $\tilde{y}_{\ell} = y_j$.
\item[(ii)] We have $L(x,y) \ge Q(\gamma_1+\gamma_2+3\gamma_3) + (A-1)\rho_1 +\rho_0 + (Q-A)\ell_\y$ if and only if there are $i\in [A], j\in [B]$ with $\langle a_i,b_j \rangle = 0$.
\item[(iii)] We have $|y| \le |x| \le \Oh(B\cdot D)$ and $\delta(x,y) = \Oh(A\cdot D)$.
\end{itemize}
\end{lem}
\begin{proof}
For (i), we let $j\in [B]$ and note that $y_j = \tilde{y}_\ell$ for $\ell := A + j$. We will show that for every $\lambda \in \{0,\dots, A-1\}$, there is a (1,2)-alignment $\Lambda$ with $(k,\ell) \in \algnTwo_{P,Q}$ where $k:=2(A+j)-1-\lambda$. By the cyclic structure of $\tilde{x}$, $(\tilde{x}_k)_{0 \le \lambda < A}$ cycles through all values $x_1,\dots,x_A$. Hence, for some choice of~$\lambda$ the desired $\tilde{x}_{k} = x_i$ follows, yielding the claim.

To see that for any $\lambda \in \{0,\dots,A-1\}$, some $\Lambda\in \algnTwo_{P,Q}$ with $(k,\ell)\in \Lambda$ exists, observe that there are $\ell-1 = A+j-1$ predecessors of $\tilde{y}_\ell$ and $k-1 = 2(A+j-1)-\lambda = 2(\ell - 1) - \lambda$ predecessors of $\tilde{x}_k$. Hence there is a (1,2)-alignment $\Lambda_1\in\algnTwo_{k-1,\ell-1}$ (leaving $\lambda$ indices $j\in [Q]$ uniquely aligned). Similarly, observe that there are $Q-\ell = B+A-j$ successors of $\tilde{y}_\ell$ and $P-k = 2B+A-2j + \lambda+1 = 2(Q-\ell)-(A-\lambda-1)$ successors of $\tilde{x}_k$, hence there is a (1,2)-alignment $\Lambda_2\in \algnTwo_{P-k, Q-\ell}$ (which leaves $A-(\lambda+1)$ indices $j$ uniquely aligned). By canonically composing $\Lambda_1$, $(k,\ell)$ and $\Lambda_2$ we can thus obtain $\Lambda\in \algnTwo_{P,Q}$ with $(k,\ell)\in \Lambda$.

For (ii), assume that there are $i\in [A],j\in [B]$ satisfying $\langle a_i,b_j \rangle = 0$. By (i), there is some $\Lambda\in \algnTwo_{P,Q}$ where some $\ell\in[Q]$ is uniquely aligned to some $k\in[P]$ such that $\tilde{x}_k = x_i$ and $\tilde{y}_\ell = y_j$. To apply \lemref{2vs1base}, observe that $\Lambda$ has $Q-A$ 2-aligned $j\in [Q]$, which contribute value $\ell_\y$ to $v(\Lambda)$, and $A$ uniquely aligned $j\in [Q]$, in particular, $\ell$ is uniquely aligned to $k$. Since any $\tilde{x}_i$ corresponds to some $x_{i'}$, every $\tilde{y}_j$ corresponds to some $y_{j'}$ and $L(x_{i'},y_{j'}) \in \{\rho_0,\rho_1\}$, we conclude that $\ell$ contributes $\rho_0$ to $v(\Lambda)$ and the other $A-1$ uniquely aligned $j$ contribute at least $\rho_1$. Hence by the lower bound in \lemref{2vs1base}, we obtain $L(x,y) \ge Q(\gamma_1+\gamma_2+3\gamma_3) + v(\Lambda)$, where $v(\Lambda) \ge (A-1)\rho_1 + \rho_0 + (Q-A)\ell_\y$.

Assume now that no $i\in [A],j\in[B]$ satisfy $\langle a_i,b_j \rangle = 0$, and let $\Lambda\in \algnMult_{P,Q}$. Then any $j\in[Q]$ uniquely aligned to some $i\in[P]$ contributes $L(\tilde{x}_i,\tilde{y}_j) = \rho_1$ to $v(\Lambda)$. Let $\lambda$ be the number of $j\in[Q]$ that are $k$-aligned for any $k\ge 2$, each contributing $\ell_\y$ to $v(\Lambda)$. Then there are at most $\min\{P-2\lambda, Q-\lambda\}$ uniquely aligned $j\in[Q]$ (since every $k$-aligned $j$ blocks at least two $i\in[P]$ for other alignments), and the remaining $j\in[Q]$ are unaligned, with no contribution to $v(\Lambda)$. Hence $v(\Lambda) \le \lambda \ell_\y + \min\{P-2\lambda,Q-\lambda\} \cdot \rho_1 = \min\{P\rho_1 + (\ell_\y - 2\rho_1)\lambda,Q\rho_1 + (\ell_\y - \rho_1)\lambda\}$. Note that $\ell_\y / 2< \rho_1 \le \ell_\y$ (by \lemref{coreSingleStrings}(iv)), hence this minimum of linear functions with leading coefficients $\ell_\y - 2\rho_1< 0$ and $\ell_\y - \rho_1 \ge 0$ is maximized when both have the same value, i.e., when $\lambda = P-Q = Q-A$. Thus, $v(\Lambda) \le (Q-A)\ell_\y + A\rho_1 < (Q-A)\ell_\y + (A-1)\rho_1 + \rho_0$. Thus by the upper bound of \lemref{2vs1base} we conclude that $L(x,y) < Q(\gamma_1 + \gamma_2 + 3\gamma_3) + (Q-A)\ell_\y + (A-1)\rho_1 + \rho_0$.

For (iii), since $P \ge Q$ and $\ell_\x \ge \ell_\y$ we have $|x| \ge |y|$, and by $P \le \Oh(A)$ and $|G(\tilde x_i)| \le \Oh(\ell_\x) \le \Oh(D)$ we obtain $|x| \le \Oh(AD)$.
Note that for any (1,2)-alignment $\Lambda\in \algnTwo_{P,Q}$, we have 
$$v(\Lambda) = Q\cdot \ell_\y - \sum_{j \text{ uniquely aligned to } i} (\ell_\y - L(x_i,y_j)) = Q\cdot \ell_\y - \Oh(A \cdot D),$$ 
since by $P = 2Q-A$ the number of uniquely aligned indices $j$ in $\Lambda$ equals $A$, and $\ell_\y = \Oh(D)$. Hence by \lemref{2vs1base}, $L(x,y) \ge Q(\gamma_1+\gamma_2 +3\gamma_3) + Q\ell_\y - \Oh(A \cdot D) = |y| - \Oh(A\cdot D)$, implying $\delta(x,y) = |y| - L(x,y) \le \Oh(A \cdot D)$. 
\end{proof}

\subsubsection{Constant Alphabet}

First assume $\alpha_\Sigma = 0$ and thus $|\Sigma| = \Oh(1)$.
Consider any $n \ge 1$ and target values $p = n^{\alpha_p}$ for $p \in \params$. 
We write $\lfloor x \rfloor_2$ for the largest power of 2 less than or equal to $x$. 
Let $\sA= \{a_1,\dots,a_A\}$, $\sB=\{b_1,\dots,b_B\} \subseteq \{0,1\}^D$ be a given \OV instance with $D=n^{o(1)}$ and where we set 
$$A := \Big\lfloor \frac 1D \min\Big\{ \delta, \frac{d}{\min\{m,\Delta\}} \Big\} \Big\rfloor_2 \quad \text{and} \quad B := \Big\lfloor \frac 1D \min\{m,\Delta\} \Big\rfloor_2. $$
By $\alpha_m,\alpha_\Delta \le 1$ and (\ref{eq:lowerassumptions}) we obtain $A \ge n^{\min\{\alpha_{L}, \alpha_d - 1\} - o(1)} = n^{\Omega(1)}$ and $B=n^{\min\{\alpha_m,\alpha_\Delta\}-o(1)} = n^{\Omega(1)}$. Also note that \UOVH implies that solving such \OV instances takes time $(AB)^{1-o(1)} = \min\{d, \delta m, \delta \Delta\}^{1-o(1)}$, which is the desired bound. 
We claim that $A \le B$, implying $A \mid B$. Indeed, if $\delta \le d/\min\{m,\Delta\}$ this follows from the simple parameter relations $\delta \le m$ and $\delta \le \Delta$. Otherwise, if $\delta > d/\min\{m,\Delta\}$, then in particular $\delta \Delta > d$, implying $d < \Delta^2$. Together with the parameter relations $d \le Lm \le m^2$ we indeed obtain $d / \min\{m,\Delta\} \le \min\{m,\Delta\}$.

Thus, we may construct strings $x, y$ as in \lemref{2vs1gadget}. We finish the construction by invoking the Dominant Pair Reduction (\lemref{dreduction}) to obtain strings $x' := 0^k 1^k y 1^\ell 0^k 1^k x$ and $y' := 1^\ell 0^k 1^k y$ with $k := 2|y| + |x| + 1$ and $\ell := \Theta(A \cdot D)$ with sufficiently large hidden constant, so that $\ell > \delta(x, y)$. Then from the LCS length $L(x',y')$ we can infer whether $\sA,\sB$ has an orthogonal pair of vectors by $L(x',y') = L(x, y) + \ell + 2k$ and \lemref{2vs1gadget}(ii). Moreover, this reduction runs in time $\Oh(|x'|+|y'|) = \Oh(|x| + |y|) = \Oh(B D) \le \Oh(\min\{d,\delta m, \delta \Delta\}^{1-\eps})$ for sufficiently small $\eps > 0$ (since $\alpha_\delta > 0$ and $\alpha_d > 1 \ge \alpha_m, \alpha_\Delta$ by (\ref{eq:lowerassumptions}) and \tabref{paramChoiceRestr}). We claim that $(n,x',y')$ is an instance of $\LCS_\le(\Valpha)$. This shows that any algorithm solving $\LCS_\le(\Valpha)$ in time $\Oh(\min\{d,\delta m, \delta \Delta\}^{1-\eps})$ implies an algorithm for our \OV instances with running time $\Oh(\min\{d,\delta m, \delta \Delta\}^{1-\eps})$, contradicting \UOVH. Hence, in the current case $\alpha_L = \alpha_m$ and $\alpha_\Sigma = 0$, any algorithm for $\LCS_\le(\Valpha)$ takes time $\min\{d,\delta m, \delta \Delta\}^{1-o(1)}$, proving part of \thmref{hardness}.

It remains to show the claim that $(n,x',y')$ is an instance of $\LCS_\le(\Valpha)$. 
From \lemrefs{dreduction}{2vs1gadget}(iii) we obtain $L(x',y') = \ell + 2k + L(x, y) = \ell + 2k + |y| - \delta(x, y) \ge |y'| - \Oh(AD)$, and thus $\delta(x',y') \le \Oh(AD) \le \Oh(\delta)$.
Using the parameter relations $L \le m \le n$, \lemref{2vs1gadget}(iii), and the definition of $B$, we have $L(x',y') \le m(x',y') \le n(x',y') = |x'| \le \Oh(BD) = \Oh(\min\{m,\Delta\})$, which together with the relation $m \le n$ and the assumption $\alpha_L = \alpha_m$ shows that $p(x',y') \le \Oh(p) = \Oh(n^{\alpha_p})$ for $p \in \{L,m,n\}$. Similarly, we obtain $\Delta(x',y') \le n(x',y') \le \Oh(\min\{m,\Delta\}) \le \Oh(\Delta)$. Since $x',y'$ use the binary alphabet $\{0,1\}$, we have $|\Sigma(x',y')| = 2 \le \Oh(n^{\alpha_\Sigma})$. For the number of matching pairs we have $M(x',y') \le n(x',y')^2 = \Oh((BD)^2) = \Oh(L^2)$. Since we are in the case $\alpha_\Sigma = 0$, from the parameter relation $M \ge L^2 / |\Sigma|$ (\lemref{Mbounds}(i)) we obtain $L^2 \le \Oh(M)$ and thus also $M(x',y')$ is sufficiently small. Finally, we use \lemref{dreduction} to bound $d(x',y') \le \Oh(\ell \cdot |y|) \le \Oh(AD \cdot BD)$, which by definition of $A,B$ is $\Oh(d)$. This proves that $(n,x',y')$ belongs to $\LCS_\le(\Valpha)$. 

We remark that our proof also yields the following, which we will use for small alphabets.

\begin{lem} \label{lem:largeLCShardnessConstantAlphabet}
  Let $\Valpha$ be a parameter setting satisfying \tabref{paramChoiceRestr} with $\alpha_\Sigma = 0$ and $\alpha_L = \alpha_m$. There is a constant $\gamma \ge 1$ such that any algorithm for $\LCS_\le^\gamma(\Valpha,\{0,1\})$ takes time $\min\{d,\delta m, \delta \Delta\}^{1-o(1)}$, unless \OVH fails. This holds even restricted to instances $(n,x,y)$ of $\LCS_\le^\gamma(\Valpha,\{0,1\})$ with $|y| \le |x| \le \gamma \cdot \min\{ n^{\alpha_m}, n^{\alpha_\Delta} \}$.
\end{lem}

\subsubsection{Superconstant Alphabet}

The crucial step in extending our construction to larger alphabets is to adapt the 1vs1/2vs1 gadget such that the strings use each symbol in the alphabet $\Sigma$ roughly evenly, thus reducing the number of matching pairs by a factor $|\Sigma|$. 

Recall that given a 2-element alphabet $\Sigma'$ and a string $z$ over $\{0,1\}$, we let $z \lift \Sigma'$ denote the string $z$ lifted to alphabet $\Sigma'$ by bijectively replacing $\{0,1\}$ with $\Sigma'$.

\begin{lem}\label{lem:2vs1Largebase}
Let $P=2B+3A$ and $Q=B+2A$ for some $A \mid B$. Given strings $x_1,\dots,x_P$ of length $\ell_\x$ and $y_1,\dots,y_Q$ of length $\ell_\y$, we define, as in \lemref{2vs1base}, $\guard(w) := 0^{\gamma_1} \; 1^{\gamma_2} \; (01)^{\gamma_3} \; w \; 1^{\gamma_3}$ with $\gamma_3 := \ell_\x + \ell_\y$, $\gamma_2 := 8 \gamma_3$ and $\gamma_1 := 6\gamma_2$.
Let $\Sigma_1, \dots, \Sigma_t$ be disjoint alphabets of size 2 with $Q/t \ge A/2 + 1$. We define  
\begin{eqnarray*}
x & := & H(x_1) \; H(x_2) \; \dots \; H(x_P), \\
y & := & \guard(y_1)\lift \Sigma_{f(1)} \quad \guard(y_2)\lift \Sigma_{f(2)} \quad \dots \quad \guard(y_Q) \lift \Sigma_{f(Q)}, 
\end{eqnarray*} 
where $f(j) = \lceil \frac{j}{Q}\cdot t \rceil$ and
\begin{equation*}
H(x_i) := \begin{cases}
\guard(x_i) \lift \Sigma_{k+1} \; \guard(x_i) \lift \Sigma_k & \text{ if } \bigcup_{j= \lceil i/2 \rceil}^{\lfloor (i+A)/2 \rfloor} \{f(j)\} = \{k,k+1\} \\
\guard(x_i) \lift \Sigma_k & \text{ if } \bigcup_{j= \lceil i/2 \rceil}^{\lfloor (i+A)/2 \rfloor} \{f(j)\} = \{k\} 
\end{cases}
\end{equation*}

Then we have 
\begin{equation}\label{eq:2vs1large}
 \max_{\Lambda\in \algnTwo_{P,Q}} v(\Lambda) \le L(x,y) - Q (\gamma_1 + \gamma_2 + 3\gamma_3) \le \max_{\Lambda\in \algnMult_{P,Q}} v(\Lambda). 
\end{equation}
\end{lem}
\begin{proof}
Note that $H(\cdot)$ is well-defined, since $f(\cdot)$ maps $\{1,\dots,Q\}$ to constant-valued intervals of length at least $Q/t -1 \ge A/2$, as $f(j) = k$ if and only if $j \in \big( \frac{Qk}t - \frac Qt, \frac{Qk}t\big]$, containing at least $Q/t-1$ integers. Hence for every $i$, the $\le A/2$ values $f(\lceil i/2 \rceil),\dots, f(\lfloor (i+A)/2 \rfloor)$ can touch at most 2 different constant-valued intervals.

The proof of \eqref{eq:2vs1large} is based on the proof of \lemref{2vs1base} (the analogous lemma for alphabet $\Sigma = \{0,1\}$). For the first inequality of \eqref{eq:2vs1large}, let $\Lambda\in \algn_{P,Q}^{1,2}$ and define for every $j$ the substring $z'_j = \bigconcat_{i\in \Lambda(j)} H(x_i)$. Note that under $\Lambda$, each $\Lambda(j)$ consists of one or two elements from $\{2j - A, \dots, 2j\}$, since there are at most $2Q-P = A$ uniquely aligned $j$. In other words, for any $i \in \Lambda(j)$ we have $j \in \{ \lceil i/2 \rceil, \ldots, \lfloor (i+A)/2 \rfloor \}$. Thus, by definition each $H(x_i)$ for $i\in \Lambda(j)$ contains $\guard(x_i)\lift \Sigma_{f(j)}$ as a substring and hence $z'_j$ contains $\bigconcat_{i\in \Lambda(j)} \guard(x_i)\lift \Sigma_{f(j)}$ as a subsequence. This proves 
\[L\big(z'_j,\guard(y_j) \lift \Sigma_{f(j)} \big) \ge L\Big(\bigconcat_{i\in \Lambda(j)} \guard(x_i) \lift \Sigma_{f(j)}, \guard(y_j) \lift \Sigma_{f(j)}\Big) = L\Big(\bigconcat_{i\in \Lambda(j)} \guard(x_i), \guard(y_j)\Big),\]
which reduces the proof to the case of $\Sigma = \{0,1\}$ -- note that the last term is equal to $L(z_j,\guard(y_j))$ in the proof of the same inequality of \lemref{2vs1base} and thus the remainder follows verbatim.

It remains to show the second inequality of \eqref{eq:2vs1large}. Essentially as in the proof of \lemref{2vs1base}, we write $x=z'_1 z'_2 \dots z'_Q$ with $L(x,y) = \sum_{j=1}^Q L(z'_j, \guard(y_j)\lift \Sigma_{f(j)})$. 
For every $z'_j$, we obtain a string $z_j$ by deleting all symbols not contained in $\Sigma_{f(j)}$ and then lifting it to the alphabet $\{0,1\}$. We conclude that $L(z'_j,\guard(y_j)\lift \Sigma_{f(j)}) = L(z_j,\guard(y_j))$. We claim that $z := z_1 z_2 \dots z_Q$ is a subsequence of $x_{\{0,1\}} := \guard(x_1) \dots \guard(x_P)$ (which is equal to the string $x$ that we constructed in the case of $\Sigma=\{0,1\}$). Indeed, if $H(x_i)$ is of the form $w_{k+1} w_k$ for some $k$ with $w_\ell = \guard(x_i)\lift \Sigma_\ell$, then symbols of at most one of $w_k$ and $w_{k+1}$ are contained in $z$. To see this, note that if $w_k$ is not deleted then at least one of its symbols is contained in some $z'_j$ with $f(j) = k$, but then no symbol in $w_{k+1}$ can be contained in $z'_{j'}$ with $f(j') = k+1$, since this would mean $j' > j$, so $w_{k+1}$ is deleted. Thus,
\[ L(x,y) = \sum_{j=1}^Q L\big(z'_j, \guard(y_j) \lift \Sigma_{f(j)}\big) = \sum_{j=1}^Q L\big(z_j, \guard(y_j)\big) \le L(x_{\{0,1\}},y_{\{0,1\}}),\]
where $y_{\{0,1\}} := \guard(y_1)\dots \guard(y_Q)$ is the string $y$ that we constructed in the case of $\Sigma = \{0,1\}$. Hence, the second inequality of \eqref{eq:2vs1large} follows from the proof of \lemref{2vs1base}.
\end{proof}

By the same choice of vectors as in \lemref{2vs1gadget}, we can embed orthogonal vectors instances.

\begin{lem}\label{lem:2vs1gadgetlarge}
Let $a_1,\dots,a_A,b_1,\dots b_B\subseteq\{0,1\}^D$ be given with $A\mid B$. Construct the corresponding strings $x_1,\dots,x_A$ of length $\ell_\x$, $y_1,\dots,y_B$ of length $\ell_\y \le \ell_\x \le \Oh(D)$ and integers $\rho_0,\rho_1$ as in \lemref{coreSingleStrings} and define
\begin{eqnarray*}
\tilde{x} & := & (\tilde{x}_1,\dots,\tilde{x}_P) = (\overbrace{x_1,\dots,x_A,x_1,\dots,x_A,\dots,x_1,\dots,x_A}^{2\cdot(B/A)+3 \text{ {\rm groups of size }} A}), \\
\tilde{y} & := & (\tilde{y}_1,\dots,\tilde{y}_Q) = (\underbrace{y_1,\dots,y_1}_{A \text{ {\rm copies of }} y_1}, y_1, \dots, y_B, \underbrace{y_1, \dots, y_1}_{A \text{ {\rm copies of }}y_1}),
\end{eqnarray*}
where $P:= 2B+3A$ and $Q:= B+2A$. 
For disjoint alphabets $\Sigma_1,\dots,\Sigma_t$ of size 2 with $Q/t \ge A/2+1$, we construct the instance $x := \bigconcat_i H(\tilde{x}_i)$, $y:=\bigconcat_j \guard(\tilde{y}_j)$ of \lemref{2vs1Largebase} (with the corresponding choice of $\gamma_1,\gamma_2$ and $\gamma_3$). This satisfies the following properties:
\begin{itemize}
\item[(i)] We have that $L(x,y) \ge Q(\gamma_1+\gamma_2+3\gamma_3) + (A-1)\rho_1 +\rho_0 + (Q-A)\ell_\y$ if and only if there are $i\in [A], j\in [B]$ with $\langle a_i,b_j \rangle = 0$.
\item[(ii)] We have $|y| \le |x| \le \Oh(B\cdot D)$ and $\delta(x,y) = \Oh(A\cdot D)$. 
\end{itemize}
\end{lem}
\begin{proof}
The lemma and its proof are a slight adaptation of \lemref{2vs1gadget}: For (i), since \lemref{2vs1Largebase} proves \eqref{eq:2vs1large} which is identical to \eqref{eq:2vs1}, we can follow the proof of \lemref{2vs1gadget}(i) and (ii) verbatim (since we have chosen $\tilde{x}$ and $\tilde{y}$ as in this lemma). 
For (ii), the bounds $|y| \le |x| \le \Oh(B\cdot D)$ and $\delta(x,y) = \Oh(A\cdot D)$ follow exactly as in \lemref{2vs1base} (note that only $|x|$ has increased by at most a factor of 2, so that $|x| \le \Oh(B \cdot D)$ still holds by the trivial bound). 
\end{proof}

We can now finish the proof of \thmref{hardness} for the case of $\alpha_L = \alpha_m$ and $\alpha_\Sigma > 0$.
Consider any $n \ge 1$ and target values $p = n^{\alpha_p}$ for $p \in \params$. 
Let $\sA= \{a_1,\dots,a_A\}$, $\sB=\{b_1,\dots,b_B\} \subseteq \{0,1\}^D$ be a given \OV instance with $D=n^{o(1)}$ and where we set, as in the case $\alpha_\Sigma = 0$,
$$A := \Big\lfloor \frac 1D \min\Big\{ \delta, \frac{d}{\min\{m,\Delta\}} \Big\} \Big\rfloor_2 \quad \text{and} \quad B := \Big\lfloor \frac 1D \min\{m,\Delta\} \Big\rfloor_2. $$
As before, we have $A \mid B$, so we may construct strings $x, y$ as in \lemref{2vs1gadgetlarge}, where we set $t := \min\{\lfloor Q/(A/2+1) \rfloor, |\Sigma|\} = \Theta(\min\{B/A,|\Sigma|\})$. We finish the construction by invoking the Dominant Pair Reduction (\lemref{dreduction}) to obtain strings $x' := y 2^\ell x$ and $y' := 2^\ell y$, where $2$ is a symbol not appearing in $x,y$ and we set $\ell := \Theta(A \cdot D)$ with sufficiently large hidden constant, so that $\ell > \delta(x, y)$. 

For the remainder of the proof we can follow the case $\alpha_\Sigma=0$ almost verbatim. The only exception is the bound on the number of matching pairs. Note that symbol $2$ appears $\Oh(AD)$ times in $x'$ and $y'$. As in $x$ and $y$ every symbol appears roughly equally often and the total alphabet size is $\Theta(t)$, for any symbol $\sigma \ne 2$ we have $\occ_\sigma(x) \le \Oh(|x| / t)$ and $\occ_\sigma(y) \le \Oh(|y|/t)$, implying $\occ_\sigma(x'), \occ_\sigma(y') \le \Oh(BD/t)$. Hence, $M(x',y') \le \Oh((AD)^2 + t \cdot (BD/t)^2)$. Using $t = \Theta(\min\{B/A,|\Sigma|\})$ and $A \le B$, we obtain $M(x',y') \le \Oh(\max\{AD \cdot BD, (BD)^2 / |\Sigma|\}) \le \Oh(\max\{d, m^2 / |\Sigma|\}$. The assumption $\alpha_L = \alpha_m$ and the parameter relations $M \ge L^2 / |\Sigma|$ and $M \ge d$ now imply $M(x',y') \le \Oh(M)$. This concludes the proof of \thmref{hardness}.

\section{Hardness for Small Constant Alphabet}
\label{sec:hardnessSmallSigma}

In this section, we show hardness of the parameter settings $\LCS(\alpha, \Sigma)$ for alphabets of constant size $|\Sigma| \ge 2$, i.e., we prove \thmref{main2}. The general approach, as outlined in \secref{roughproof}, is to take the hard instances $x,y$ of $\LCS_\le(\Valpha, \{0,1\})$ constructed in \secref{hardness} and pad them to instances $x',y'$ of $\LCS(\Valpha, \Sigma)$. 
Notably, unlike the black-box method of \lemref{reductiontoclosure} that effectively considered each parameter separately, we now cannot make extensive use of the Disjoint Alphabets Lemma, as this would introduce more symbols than admissible. 
Instead, for small alphabet size such as $|\Sigma|=2$ we need to pad all parameters simultaneously in a combined construction, taking care of the interplay of the parameters manually. Additionally, for $|\Sigma|\in \{2,3\}$, more complex parameter relations hold. 

Unfortunately, this general approach fails for $\Sigma=\{0,1\}$, i.e., we cannot always pad hard strings $x,y$ of $\LCS_\le(\Valpha, \{0,1\})$ to $\LCS(\Valpha, \{0,1\})$. Surprisingly, the reason is that by an $\Oh(n+\delta M/n)$-time algorithm (given in \secref{algo}), some parameter settings $\LCS(\Valpha,\{0,1\})$ are indeed simpler to solve than $\LCS_\le(\Valpha, \{0,1\})$ (conditional on SETH). In these cases, we take hard instances $(n,x,y)$ from $\LCS_\le(\Valpha', \{0,1\})$ for a suitably defined ``simpler'' parameter setting $\Valpha'$ and pad $x,y$ to instances of $\LCS(\Valpha, \{0,1\})$. 

As in \secref{hardness}, we distinguish between the two cases $\alpha_\delta = \alpha_m$ (i.e., $\delta = \Theta(m)$ and any $0< \alpha_L \le \alpha_m$ is admissible) and $\alpha_L = \alpha_m$ (i.e., $L = \Theta(m)$ and any $0 < \alpha_\delta < \alpha_m$ is admissible).

\subsection{Small LCS}

In this section, we assume $\alpha_\delta = \alpha_m$. It can be checked that this assumption implies $\alpha_\Delta = 1$, i.e., $\Delta = \Theta(n)$. Moreover, if $|\Sigma| = 2$ then the assumption and the parameter relation $M \ge nd/(80L) \ge \Omega(n d /m)$ imply $\delta M / n = \Omega(d)$. Thus, the desired running time bound simplifies to $d^{1-o(1)}$. \thmref{main2} in this regime follows from the following statement (and \lemref{paramsnecessary}).

\begin{lem}\label{lem:smallLCSsmallSigma}
Let $(\Valpha,\Sigma)$ be a parameter setting satisfying \tabref{paramChoiceRestr} with $\alpha_\delta=\alpha_m$. There is a constant $\gamma \ge 1$ such that any algorithm for $\LCS^\gamma(\Valpha,\Sigma)$ takes time $d^{1-o(1)}$ unless \OVH\ fails.
\end{lem}

We prove the above lemma in the remainder of this section.
Note that any parameter setting $(\Valpha,\Sigma)$ satisfying \tabref{paramChoiceRestr} gives rise to a parameter setting $\Valpha$ satisfying \tabref{paramChoiceRestr} with $\alpha_\Sigma = 0$ (where the converse does not hold in general). 
Recall that for any such $\Valpha$, in \lemref{smallLCShardnessConstantAlphabet} we constructed hard instances $(n,x,y)$ of $\LCS^\gamma_\le(\Valpha, \{0,1\})$ with an additional threshold $\rho$ such that deciding $L(x,y) \ge \rho$ decides the corresponding \OV\ instance, yielding hardness of $\LCS^\gamma_\le(\Valpha, \{0,1\})$. Furthermore, the constructed instances have the additional guarantees that $|x|,|y| \le \gamma \cdot n^{\alpha_L}$ and $\occ_1(y) \le \gamma \cdot n^{\alpha_d-\alpha_L}$ and for any $\beta \ge 0$ we have $L(x, 0^\beta y) = L(x,y)$.  

Hence, to prove the above lemma it suffices to show how to compute, given any such instance $(n,x,y)$ and threshold $\rho$, an instance $x',y'$ of $\LCS^{\gamma'}(\Valpha, \Sigma)$ (for some $\gamma'\ge 1)$ and an integer $\rho'$ in time $\Oh(n)$ such that $L(x',y') \ge \rho'$ if and only if $L(x,y) \ge \rho$. More precisely, we show how to compute an integer $\tau$ in time $\Oh(n)$ such that $L(x',y') \ge \tau + \rho$ if and only if $L(x,y) \ge \rho$.

We will do this successively for alphabet sizes $|\Sigma|=2,|\Sigma|=3, |\Sigma|=4$, and $|\Sigma| \ge 5$. To this end, the following basic building block will be instantiated with different parameters. Recall that in \lemref{genDomPairs}, we defined strings $a=(01)^{R+S}$ and $b=0^R(01)^S$ with the properties $L(a,b)=|b|=R+2S$ and $d(a,b)=\Theta(R\cdot S)$.

\begin{lem}[Basic building block]\label{lem:bbbI}
Let $x,y$ be given strings. Given $\alpha,\beta, R, S \ge 0$, we set $\ell :=|x|+|y|$,  and define 
\begin{alignat*}{2}
x' \;& := \quad a \; 1^\alpha \; 0^\ell \; x \quad &&= \quad (01)^{R+S} \; 1^\alpha \; 0^\ell \; x \\
y' \;& := \quad b \; 0^\beta \; 0^\ell \; y \quad &&= \quad 0^R(01)^S \; 0^\beta \; 0^\ell \; y.
\end{alignat*}
Then we have $L(x',y') = L(a,b) +\ell + L(x,0^\beta y)=R+2S+\ell+L(x,0^\beta y)$. 
If $L(x, 0^\beta y) = L(x,y)$ then we even have $L(x',y') = R+2S+\ell+L(x,y)$.
\end{lem}
\begin{proof}
Clearly, $L(x',y') \ge L(a,b) + L(0^\ell, 0^\ell) + L(x,0^\beta y) = (R+2S)+\ell+L(x,0^\beta y)$ since $L(a,b)=|b|=R+2S$ by \lemref{genDomPairs}.
To prove a corresponding upper bound, note that we can partition $y'=w z$ such that $L(x',y') = L(a 1^\alpha, w) + L(0^\ell x, z)$. Consider first the case that $z$ is a subsequence of $y$. Then 
\begin{align*}
 L(x',y')  = L(a 1^n ,w) + L(0^\ell x, z)  
  \le L(a 1^n, y')  + L(0^\ell x, y),
\end{align*}
since $w,z$ are subsequences of $y',y$, respectively. Using $L(u,v) \le \sum_{\sigma\in\Sigma} \min\{ \occ_\sigma(u), \occ_\sigma(v) \}$ for any strings $u,v$, we obtain
\begin{align*}
 L(x',y') & \le (\occ_0(a 1^n) + \occ_1(y')) + (\occ_0(y) + \occ_1(0^\ell x)) \\
& = (R+S) + (S + \occ_1(y)) + \occ_0(y) + \occ_1(x) & & \\
& \le (R+2S) + \ell + L(x,0^\beta y),
\end{align*}
since $\ell \ge |x| + |y| \ge \occ_0(x) + \occ_0(y) + \occ_1(y)$.
It remains to consider the case that $z$ is not a subsequence of $y$ and hence $w$ is a subsequence of $b 0^{\beta+ \ell}$. By \lemref{genDomPairs}\itemref{genDomPairsLmore}, we can without loss of generality assume that $w$ is a subsequence of $b$, since $L(a1^\alpha ,b0^{\beta+\ell}) = L(a,b)$. We write $z = z' z''$ such that $z''$ is a subsequence of $0^{\ell+\beta} y$ and maximal with this property. Hence, $w z'$ is a subsequence of $b$. 
Using the facts $L(u,v) \le |v|$ and $L(u,v'v'') \le |v'| + L(u,v'')$, we bound
\begin{align*}
 L(x',y')   =  L(a 1^n, w)  + L(0^\ell x, z' z'') 
 \le  |w| + (|z'| + L(0^\ell x, z'')), 
\end{align*}
Since $w z'$ is a subsequence of $b$ and $z''$ is a subsequence of $0^{\ell+\beta} y$, this yields
\begin{align*}
 L(x',y')  \le  |b| + L(0^\ell x, 0^{\ell + \beta} y)  
 =  (R+2S) + \ell + L(x,0^\beta y), 
\end{align*}
where we used greedy prefix matching. This finishes the proof.
\end{proof}

We now instantiate the basic building block to prove \lemref{smallLCSsmallSigma} for $\Sigma= \{0,1\}$. Note that in the remainder we again simply write $p$ for the target value $\lceil n^{\alpha_p} \rceil$ of parameter $p \in \params$, while the parameter value attained by any strings $x,y$ is denoted by $p(x,y)$, as in \secref{hardness}. Note that the additional guarantees for $(n,x,y)$ are satisfied by \lemref{smallLCShardnessConstantAlphabet}.

\begin{lem} \label{lem:smallConstAlphabetTwo}
Consider a parameter setting $(\Valpha,\{0,1\})$ satisfying \tabref{paramChoiceRestr} with $\alpha_\delta = \alpha_m$. Let $(n,x,y)$ be an instance of $\LCS_\le^\gamma(\Valpha,\{0,1\})$ with $|x|,|y| \le \gamma \cdot L$ and $\occ_1(y) \le \gamma \cdot d/L$ satisfying $L(x, 0^{\beta'} y) = L(x,y)$ for any $\beta' \ge 0$. We obtain strings $x',y'$ from \lemref{bbbI} (recall that in this lemma we set $\ell := |x| + |y|$), where we choose 
\begin{align*}
R &:= L, & S & := \lfloor d/L \rfloor, &  
 \beta & := \tilde{m} := \max\{m,2|x|\}, & \alpha &:= \tilde{n} := \max\{n,\tilde{m}+|y|\}. 
\end{align*}
Then, setting $\kappa:= \lfloor M/n \rfloor$, the strings defined by
\begin{alignat*}{3}
x'' := &\quad 1^{\kappa} \; x' \quad &&= \quad 1^{\kappa} \; a \;1^{\tilde{n}} \; 0^\ell \; x \quad  &&= \quad 1^{\kappa} \; (01)^{R+S} \;1^{\tilde{n}} \; 0^\ell \; x, \\
y'' := & \quad 1^{\kappa} \; y' \quad &&= \quad 1^{\kappa} \; b \; 0^{\ell+\tilde{m}} \; y \quad  &&= \quad 1^{\kappa} \; 0^R (01)^S \; 0^{\ell+\tilde{m}} \; y.
\end{alignat*}
are an instance of $\LCS^{\gamma'}(\Valpha, \{0,1\})$ (for some constant $\gamma'\ge \gamma$) and can be computed in time $\Oh(n)$, together with an integer $\tau$ such that $L(x'',y'') \ge \tau + \rho$ if and only if $L(x,y)\ge \rho$.
\end{lem}
\begin{proof}
Note that 
\begin{equation}\label{eq:largelcssigma2lcs}
L(x'',y'') = \kappa  + L(x',y') = \kappa + (R+2S) + \ell + L(x,y),
\end{equation}
where the first equality follows from greedy prefix matching and the second follows from \lemref{bbbI}. Thus by setting $\tau = \kappa+(R+2S) + \ell$, we have that $L(x'',y'') \ge \tau + \rho$ if and only if $L(x, y)\ge \rho$. Clearly, $x'',y''$, and $\tau$ can be computed in time $\Oh(n)$, and $\Sigma(x'',y'') = \{0,1\}$.

We first verify that $|x|,|y|,\ell,R,S,|a|,|b|,\kappa = \Oh(L)$. \dopara{L}  By assumption, $|x|,|y|=\Oh(L)$ and thus $\ell = |x|+|y| = \Oh(L)$. By the parameter relation $d \le |\Sigma| \cdot L^2=2L^2$, we note that $d/L = \Oh(L)$ and hence by choice of $R,S$, we have $|a|,|b| = \Theta(R+S) = \Theta(L+d/L) = \Theta(L)$. Furthermore, the parameter relation $M\le 2Ln$ implies $\kappa \le M/n\le2L$. Since $L(x,y) \le |x| = \Oh(L)$,  the bound $L(x'',y'') = \kappa+R+2S + \ell + L(x,y) = R + \Oh(L) = \Theta(L)$ follows directly from~\eqref{eq:largelcssigma2lcs}.

Observe that $\tilde{n}$ is chosen such that $|x''| \ge |y''|$. Also, $\tilde{m} = \Theta(m)$ and $\tilde{n} = \Theta(n)$. \dopara{n,m} Since $L\le m\le n$, we thus have $|x''| = \kappa + |a| + \tilde{n} + \ell + |x| = \tilde{n} + \Oh(L) = \Theta(n)$ and $|y| = \kappa + |b| + \tilde{m} + \ell+|y|=\tilde{m} + \Oh(L) = \Theta(m)$. 

Note that by \eqref{eq:largelcssigma2lcs}, $\delta(x'',y'') = (\tilde{m} + |y|) - L(x,y) \ge \tilde{m} - |x| \ge m/2$. \dopara{\delta,\Delta} Hence, $\delta(x'',y'') = \Theta(m) = \Theta(\delta)$ (by the assumption $\alpha_\delta=\alpha_m$). 
Moreover, since $\delta = \Theta(m)$, for some constant $\eps > 0$ we have $\Delta = \delta + (n-m) \ge \eps m + n - m = n - (1-\eps)m \ge n - (1-\eps)n = \Omega(n)$ (where we used the parameter relation $m \le n$). Since also $\Delta \le n$ we have $\Delta = \Theta(n)$. By the same argument, using $\delta(x'',y'') = \Theta(m) = \Theta(m(x'',y''))$ and $n(x'',y'') = \Theta(n)$ as shown above, we obtain $\Delta(x'',y'') = \Theta(n(x'',y'')) = \Theta(n)$, and thus $\Delta(x'',y'') = \Theta(\Delta)$.

For $M$\dopara{M}, observe that $\occ_1(x'') = \kappa + \occ_1(a) + \tilde{n} + \occ_1(x)=\tilde{n} + \Oh(L) = \Theta(n)$. Moreover, $\occ_1(y)= \Oh(d/L)$ (by assumption) and $\occ_1(b) = S = \Oh(d/L)$ yield $\occ_1(y'') = \kappa+ \occ_1(b)+\occ_1(y) = \Theta(M/n)  + \Oh(d/L)$ (here $\kappa = \Theta(M/n)$ follows from the parameter relation $M\ge n$). This yields $\occ_1(x'')\cdot \occ_1(y'') = \Theta(M) + \Oh(dn/L) = \Theta(M)$ (using the parameter relation $M\ge nd/(5L)$). Also note that $\occ_0(x'') = \occ_0(a) + \ell + \occ_0(x) = \Oh(L)$ and $\occ_0(y'') = \occ_0(b) + \ell + \tilde{m} + \occ_0(y) = \tilde{m} + \Oh(L) = \Oh(m)$. This yields $\occ_0(x'')\cdot \occ_0(y'') = \Oh(Lm) = \Oh(M)$ (using the parameter relation $M\ge Lm/4$). Combining these bounds, we obtain $M(x'',y'') = \occ_0(x'')\cdot \occ_0(y'') + \occ_1(x'')\cdot \occ_1(y'') = \Theta(M)$. Note that the last two parameter relations used here exploited that we have $\Sigma= \{0,1\}$. 

It remains to determine the number of dominant pairs\dopara{d}.  Since $L(x',y') = \Theta(L)$ (as argued above) and $\occ_1(y') = \Oh(d/L)$, \lemref{few1s} yields $d(x',y') \le 5 L(x',y')\cdot \occ_1(y') = \Oh(L \cdot d/L)= \Oh(d)$. For a corresponding lower bound, from \obsref{prefix} and \lemref{genDomPairs} we obtain $d(x',y') \ge d(a,b) \ge R\cdot S = \Omega(d)$. By~\lemref{greedy}, the claim now follows from $d(x'',y'') = \kappa + d(x',y') = \Oh(L) + \Theta(d)=\Theta(d)$, where we use $\kappa = \Oh(L)$ and the parameter relation $d\ge L$.
 \end{proof}

The case $\Sigma = \{0,1,2\}$ is similar to $\{0,1\}$, except that we use the new symbol $2$ to pad the parameter $n$, we use symbol $1$ to pad $m$, and we have to swap the constructions for $x''$ and $y''$.

\begin{lem} \label{lem:smallConstAlphabetThree}
Consider a parameter setting $(\Valpha,\{0,1,2\})$ satisfying \tabref{paramChoiceRestr} with $\alpha_\delta = \alpha_m$. Let $(n,x,y)$ be an instance of $\LCS_\le^\gamma(\Valpha,\{0,1\})$ with $|x|,|y| \le \gamma \cdot L$ and $\occ_1(y) \le \gamma \cdot d/L$ satisfying $L(x, 0^{\beta'} y) = L(x,y)$ for any $\beta' \ge 0$. We obtain strings $x',y'$ from \lemref{bbbI} (recall that in this lemma we set $\ell := |x| + |y|$), where we choose
\begin{align*}
R &:= L, & S & := \lfloor d/L \rfloor, &  
 \beta & := 0, & \alpha &:= m. 
\end{align*}
Then, setting $\kappa:= \lfloor M/n \rfloor$ and $\tilde{n}:=\max\{n,\kappa+|a|+m+|x|\}$, the strings defined by
\begin{alignat*}{7}
x'' &:= \quad 2^{\tilde{n}} \; y' \quad & & = \quad 2^{\tilde{n}} & & \; b \; & & 0^{\ell} \; y \quad & & = \quad 2^{\tilde{n}} & & \; 0^R (01)^S \; & & 0^{\ell} \; y, \\
y'' &:=  \quad 2^{\kappa} \; x' \quad & & = \quad 2^{\kappa}  & & \; a  \;1^{m} \; & &  0^\ell \; x \quad & & = \quad 2^{\kappa}  & & \; (01)^{R+S}  \;1^{m} \; & &  0^\ell \; x.
\end{alignat*}
are an instance of $\LCS^{\gamma'}(\Valpha, \{0,1,2\})$ (for some constant $\gamma'\ge \gamma$) and can be computed in time $\Oh(n)$, together with an integer $\tau$ such that $L(x'',y'') \ge \tau + \rho$ if and only if $L(x,y)\ge \rho$.
\end{lem}
\begin{proof}
Note that unlike the case $\{0,1\}$ the string $x$ now appears in $y''$ and $y$ appears in $x''$, so the constructions are swapped. This is necessary to realize $m$ and $M$, using the parameter relation $M \ge md/(80L)$ that holds for $\Sigma = \{0,1,2\}$. Observe that as usual $|x''| \ge |y''|$.

We first compute
\begin{equation}\label{eq:largelcssigma3lcs}
L(x'',y'') =  L(2^{\tilde{n}}, 2^\kappa)  + L(y',x') = \kappa + (R+2S) + \ell + L(x,y),
\end{equation}
where the first equality follows from the Disjoint Alphabets Lemma and the second equality from greedy prefix matching and \lemref{bbbI}. Thus, by setting $\tau = \kappa+(R+2S) + \ell$, we have $L(x'',y'') \ge \tau + \rho$ if and only if $L(x,y)\ge \rho$. Clearly, $x'',y''$, and $\tau$ can be computed in time $\Oh(n)$, and $\Sigma(x'',y'') = \{0,1,2\}$.

As in the case $\{0,1\}$, we have $|x|,|y|,\ell,R,S,|a|,|b|,\kappa = \Oh(L)$.\dopara{L, m, n} Thus, by \eqref{eq:largelcssigma3lcs}, we have $L(x'',y'') = R + \Oh(L) = \Theta(L)$. Furthermore, note that $\tilde{n} = \Theta(n)$. Thus, $|y''| = \kappa+|a| + m + \ell + |x| = m + \Oh(L) = \Theta(m)$ and $|x''| = \tilde{n} + |b| + \ell+|y|=\tilde{n} + \Oh(L) = \Theta(n)$. Since $L(x, y) \le |x| = \Oh(L)$,  the bound $L(x'',y'') = R+\Oh(L) = \Theta(L)$ follows directly from~\eqref{eq:largelcssigma3lcs}. 

By~\eqref{eq:largelcssigma3lcs},\dopara{\delta,\Delta} we see that $\delta(x'',y'') = |y''| - L(x'',y'') = R+m +(|x|-L(x,y)) \ge m$. Hence $\delta(x'',y'') = \Theta(m)$. Thus, $\Delta(x'',y'') = \delta(x'',y'') + (|x''| - |y''|) = \Theta(n)$ follows as in the case $\{0,1\}$.

For $M$\dopara{M,\Sigma}, observe that $\occ_1(y'') = \occ_1(a) + m + \occ_1(x)=m + \Oh(L) = \Theta(m)$. Moreover, $\occ_1(y)=\Oh(d/L)$ (by assumption) yields $\occ_1(x'') = \occ_1(b)+\occ_1(y) = S + \Oh(d/L) = \Oh(d/L)$. Also note that $\occ_0(y'') = \occ_0(a) + \ell + \occ_0(x) = \Oh(L)$ and $\occ_0(x'') = \occ_0(b) + \ell + \occ_0(y) = \Oh(L)$. Since furthermore $\occ_2(y'') = \Theta(M/n)$ (by the parameter relation $M\ge n$) and $\occ_2(x'') = \Theta(n)$, we conclude that $M(x'',y'') = \sum_{\sigma \in \{0,1,2\}} \occ_\sigma(x'')\cdot \occ_\sigma(y'') = \Oh(dm/L + L^2) + \Theta(M)$. By the parameter relations $M \ge md/(80L)$ (using that $\Sigma =\{0,1,2\}$) and $M\ge L^2/|\Sigma| = \Omega(L^2)$, this yields $M(x'',y'') = \Theta(M)$. 

For the remaining parameter $d$, by the disjoint alphabets lemma and \lemref{greedy} we have $d(x'',y'') = d(2^{\tilde{n}},2^{\kappa}) + d(y',x') = \kappa + d(x',y')$ (using symmetry $d(x,y) = d(y,x)$). The remaining arguments are the same as in the case $\{0,1\}$. 
 \end{proof}

In the case $\Sigma = \{0,1,2,3\}$ we can use the new symbol $3$ to pad $m$ (instead of using symbol $1$, as in the previous case). Note that now $x$ appears in $x''$ and $y$ in $y''$, as in the case $\{0,1\}$.

\begin{lem} \label{lem:smallConstAlphabetFour}
Consider a parameter setting $(\Valpha,\{0,1,2,3\})$ satisfying \tabref{paramChoiceRestr} with $\alpha_\delta = \alpha_m$. Let $(n,x,y)$ be an instance of $\LCS_\le^\gamma(\Valpha,\{0,1\})$ with $|x|,|y| \le \gamma \cdot L$ and $\occ_1(y) \le \gamma \cdot d/L$ satisfying $L(x, 0^{\beta'} y) = L(x,y)$ for any $\beta' \ge 0$. We obtain strings $x',y'$ from \lemref{bbbI} (recall that in this lemma we set $\ell := |x|+|y|$), where we choose
\begin{align*}
R &:= L, & S & := \lfloor d/L \rfloor, &  
 \beta & := 0, & \alpha &:= 0.
\end{align*}
Then, setting $\kappa:= \lfloor M/n \rfloor$ and $\tilde{n}:=\max\{n,m+\kappa+|y|\}$, the strings defined by
\begin{alignat*}{7}
 x'' &:= \quad  3 \; & &2^{\tilde{n}} \; x' \quad & &= \quad 3 \; & &  2^{\tilde{n}} \; a \; 0^\ell \; x \quad & &= \quad 3 \; & &  2^{\tilde{n}} \; (01)^{R+S} \; & & 0^\ell \; x, \\
 y'' &:=  \quad 3^{m} \; & & 2^{\kappa} \; y' \quad & &= \quad 3^m \; & & 2^{\kappa} \; b \; 0^{\ell} \; y \quad & &= \quad 3^m \; & & 2^{\kappa} \; 0^R (01)^S \; & & 0^{\ell} \; y,
\end{alignat*}
are an instance of $\LCS^{\gamma'}(\Valpha, \{0,1,2,3\})$ (for some constant $\gamma'\ge \gamma$) and can be computed in time $\Oh(n)$, together with an integer $\tau$ such that $L(x'',y'') \ge \tau + \rho$ if and only if $L(x,y)\ge \rho$.
\end{lem}
\begin{proof}
We compute
\begin{equation}\label{eq:smalllcssigma4lcs}
L(x'',y'') = L(3,3^m) + L(2^{\tilde{n}}, 2^{\kappa})  + L(x',y') = 1+\kappa + (R+2S) + \ell + L(x,y),
\end{equation}
where the first equality follows from the Disjoint Alphabets Lemma and the second follows from greedy prefix matching and \lemref{bbbI}. Hence, by setting $\tau = 1+\kappa+R+2S+\ell$, we have $L(x,y) \ge \rho$ if and only if $L(x'',y'') \ge \tau + \rho$. Clearly, $x'',y''$, and $\tau$ can be computed in time $\Oh(n)$, and $\Sigma(x'',y'') = \{0,1,2,3\}$.

As for the cases $\{0,1\}$ and $\{0,1,2\}$, we have $|x|,|y|,\ell, R, S, |a|,|b|,\ell, \kappa =\Theta(L)$.\dopara{n,m,L,\delta,\Delta} Note that by choice of $\tilde{n}$, we have again $|x''|\ge|y''|$ and $\tilde{n} = \Theta(n)$. Hence, $|x''| = 1+\tilde{n} + |a| + \ell + |x| = \tilde{n} + \Oh(L) = \Theta(n)$ and $|y''| = m + \kappa + |b| + \ell+|y|= m + \Oh(L) = \Theta(m)$. Since $L(x, y) \le |x| = \Oh(L)$,  the bound $L(x'',y'') = R+ \Oh(L) = \Theta(L)$ follows directly from~\eqref{eq:smalllcssigma4lcs}. Note that \eqref{eq:smalllcssigma4lcs} also implies that $\delta(x'',y'') = |y''| - L(x'',y'') = m-1 + |x|-L(x,y) \ge m - 1$ and hence $\delta(x'',y'') = \Theta(m)$. Thus, $\Delta(x'',y'') = \delta(x'',y'') + (|x''| - |y''|) = \Theta(n)$ follows as for the case $\{0,1\}$.

For $M$\dopara{M,\Sigma}, observe that $|a0^\ell x|,|b0^\ell y|= \Oh(L)$ implies that $M(a0^\ell x,b 0^\ell y) = \Oh(L^2)$. By the Disjoint Alphabets Lemma, we obtain
\[M(x'',y'') = M(2^{\tilde{n}},2^\kappa) + M(3,3^m)+ M(a0^\ell x, b0^\ell y) = \kappa \tilde{n} + \Oh(m + L^2) = \Theta(M),\]
where we used $\kappa \tilde{n} = \Theta(M/n \cdot n) = \Theta(n)$ (note that $M\ge n$ implies $\kappa = \Theta(M/n)$) and the parameter relations $M\ge n \ge m$ and $M\ge L^2/|\Sigma| = \Omega(L^2)$. 

For the remaining parameter $d$, as in the case $\{0,1\}$ we show that $d(x',y') = \Theta(d)$. Now the Disjoint Alphabets Lemma and \lemref{greedy} prove that $d(x'',y'') = d(3,3^m) + d(2^{\tilde{n}},2^{\kappa}) + d(x',y') = 1+\kappa+d(x',y') = d(x',y') + \Oh(L) = \Theta(d)$ using $\kappa = \Oh(L)$ and the parameter relation $d\ge L$. 
\end{proof}

Finally, observe that for any parameter setting $(\Valpha,\Sigma)$ with $|\Sigma| \ge 5$ satisfying \tabref{paramChoiceRestr}, also the parameter setting $(\Valpha, \{0,1,2,3\})$ satisfies \tabref{paramChoiceRestr}. Hence, the following lemma transfers the hardness of $\LCS(\Valpha,\{0,1,2,3\})$ to $\LCS(\Valpha, \Sigma)$.

\begin{lem} \label{lem:smallConstAlphabetFive}
Let $\Valpha$ be a parameter setting satisfying \tabref{paramChoiceRestr} with $\alpha_\delta = \alpha_m$. 
Let $\Sigma$ be an alphabet of size $|\Sigma| \ge 5$. If there is an $\Oh(n^\beta)$-time algorithm for $\LCS(\Valpha,\Sigma)$, then also $\LCS(\Valpha,\{0,1,2,3\})$ admits an $\Oh(n^\beta)$-time algorithm.
\end{lem}
\begin{proof}
Given an instance $(x,y)$ of $\LCS(\Valpha,\{0,1,2,3\})$ with $n := |x|$, we show how to compute in time $\Oh(n)$ an instance $(x',y')$ of $\LCS(\Valpha,\Sigma)$ such that $L(x',y') = 1 + L(x,y)$. The claim then follows from applying the $\Oh(n^\beta)$-time algorithm on $x',y'$ (and subtracting 1 from the result).

Without loss of generality, let $\Sigma = \{0, \dots, \sigma\}$ with $\sigma \ge 4$. Define $x' := w x$ and $y' = w^R y$, where $w= 4 \dots \sigma$ and $w^R = \sigma \dots 4$. Then by the Disjoint Alphabets and Crossing Alphabets Lemmas (\lemrefs{disjointalphabet}{crossingalphabet}), we obtain $L(x',y') = L(w,w^R) + L(x,y) = 1 + L(x,y)$. It remains to show that $(x',y')$ is an instance of $\LCS(\Valpha, \Sigma)$. By the Crossing Alphabets Lemma, for all parameters $p\in \{d,M,n,m\}$ we have $p(w,w^R) =  \sum_{\sigma'=4}^\sigma p(\sigma',\sigma') = |\Sigma|-4$, and hence the Disjoint Alphabets Lemma yields $p(x',y') = p(w,w^R) + p(x,y) = |\Sigma|-4+p(x,y) = \Theta(p(x,y))$, by the parameter relations $n \ge m \ge |\Sigma|$ and $M \ge d \ge |\Sigma|$. For $L$ we obtain $L(x',y') = L(w,w^R) + L(x,y) = 1 + L(x,y) = \Theta(L(x,y))$. For $p\in \{\delta,\Delta\}$ this yields $p(x',y') = (|w|-1)  + p(x,y) = \Theta(p(x,y))$, since $\alpha_\Delta \ge \alpha_\delta = \alpha_m \ge \alpha_\Sigma$ (by the assumption $\alpha_\delta = \alpha_m$ and the parameter relations $\Delta \ge \delta$ and $m \ge |\Sigma|$) and thus $\Delta(x,y) \ge \delta(x,y) \ge \Omega(|w|-1)$. Hence, $(x',y')$ has the same parameters as $(x,y)$ up to constant factors, so all parameter relations satisfied by $(x,y)$ are also satisfied by $(x',y')$.
Since clearly $x',y'$ use alphabet $\Sigma$, indeed $(x',y')$ is an instance of $\LCS(\Valpha, \Sigma)$.
\end{proof}

\lemrefsss{smallConstAlphabetTwo}{smallConstAlphabetThree}{smallConstAlphabetFour}{smallConstAlphabetFive} of this section, together with the construction of hard strings in $\LCS_\le(\Valpha,\{0,1\})$ in \lemref{smallLCShardnessConstantAlphabet}, prove hardness of $\LCS(\Valpha,\Sigma)$ for any constant alphabet size in the case $\alpha_\delta = \alpha_m$, i.e., \lemref{smallLCSsmallSigma}.

\subsection{Large LCS, Alphabet Size at least 3}

In this section, we study the case that $\alpha_{L}=\alpha_m$ (and $\alpha_{\delta},\alpha_{\Delta}$ may be small). Additionally, we assume that $|\Sigma| \ge 3$. In this regime, \thmref{main2} follows from the following statement (and \lemref{paramsnecessary}).

\begin{lem}\label{lem:largelcsternaryalphabets}
Let $(\Valpha,\Sigma)$ be a parameter setting satisfying \tabref{paramChoiceRestr} with $\alpha_L=\alpha_m$ and $|\Sigma| \ge 3$. There is a constant $\gamma \ge 1$ such that any algorithm for $\LCS^\gamma(\Valpha,\Sigma)$ takes time $\min\{d, \delta m, \delta \Delta\}^{1-o(1)}$ unless \OVH\ fails.
\end{lem}

By the following lemma, it suffices to prove the result for $\Sigma = \{0,1,2\}$ (note that for any $(\Valpha, \Sigma)$ satisfying \tabref{paramChoiceRestr} with $\alpha_L=\alpha_m$ and $|\Sigma|\ge4$, also $(\Valpha, \{0,1,2\})$ satisfies \tabref{paramChoiceRestr}, since the only additional constraint $\alpha_M \ge \alpha_m + \alpha_d - \alpha_L$ for ternary alphabets simplifies, by $\alpha_L=\alpha_m$, to the constraint $\alpha_M \ge \alpha_d$, which is satisfied by $\Valpha$).

\begin{lem}
Let $\Valpha$ be a parameter setting satisfying \tabref{paramChoiceRestr} with $\alpha_L = \alpha_m$. 
Let $\Sigma$ be an alphabet of size $|\Sigma| \ge 4$. If there is an $\Oh(n^\beta)$-time algorithm for $\LCS(\Valpha,\Sigma)$, then also $\LCS(\Valpha,\{0,1,2\})$ admits an $\Oh(n^\beta)$-time algorithm.
\end{lem}
\begin{proof} 
Given an instance $(x,y)$ of $\LCS(\Valpha,\{0,1,2\})$ with $n := |x|$, we show how to compute in time $\Oh(n)$ an instance $(x',y')$ of $\LCS(\Valpha,\Sigma)$ such that $L(x',y') = |\Sigma|-3 + L(x,y)$. The claim then follows from applying the $\Oh(n^\beta)$-time algorithm on $x',y'$ (and subtracting $|\Sigma|-3$ from the result).

Without loss of generality, let $\Sigma = \{0, \dots, \sigma\}$ with $\sigma \ge 3$. Define $x' := w x$ and $y' = w y$, where $w= 3 \dots \sigma$. Then by the Disjoint Alphabets Lemma (\lemref{disjointalphabet}), we obtain $L(x',y') = L(w,w) + L(x,y) = |\Sigma|-3 + L(x,y)$. It remains to show that $(x',y')$ is an instance of $\LCS(\Valpha, \Sigma)$. By the Disjoint Alphabets Lemma, for all parameters $p \in \params^* = \{n,m, L, \delta, \Delta, |\Sigma|, M, d\}$ we have $p(x',y') = p(w,w) + p(x,y) = \sum_{\sigma'=3}^\sigma p(\sigma',\sigma') + p(x,y)$. For $p \in \{n,m,L,M,d\}$ we have $p(\sigma',\sigma') = 1$ and thus $p(x',y') = |\Sigma|-3 + p(x,y) = \Theta(p(x,y))$ by the assumption $\alpha_L = \alpha_m$ and the parameter relations $n \ge m \ge |\Sigma|$ and $M \ge d \ge |\Sigma|$. For $p\in \{\delta,\Delta\}$ this yields $p(x',y') = 0 + p(x,y) = p(x,y)$. Hence, $(x',y')$ has the same parameters as $(x,y)$ up to constant factors, so all parameter relations satisfied by $(x,y)$ are also satisfied by $(x',y')$.
Since clearly $x',y'$ use alphabet $\Sigma$, indeed $(x',y')$ is an instance of $\LCS(\Valpha, \Sigma)$.
\end{proof}

To prepare the construction, we adapt the construction of \lemref{bbbI} to obtain a desired value for $d$ (and, in later sections, $\delta$). Recall that \lemref{genDomPairs} defines $a=(01)^{R+S}$, $b=0^R (01)^S$ with $L(a,b)=|b|=R+2S$ and $d(a,b) = \Theta(R\cdot S)$.

\begin{lem}[Basic building block]\label{lem:bbbII}
Given strings $x,y$ and $R,S, \ell, \beta \ge 0$ with $\ell \ge R+|x|+|y|$ 
 we define 
\begin{alignat*}{10}
x' &\; := \quad & & \; a \; &0^\ell\; x \quad &= \quad & & \; (01)^{R+S} \; &0^\ell\; x, \\
y' &\; := \quad & 0^\beta &\; b \; &0^\ell\; y \quad &= \quad & 0^\beta &\; 0^R (01)^S \; &0^\ell\; y.
\end{alignat*}
Assume that $S \ge |x|$ \emph{or} $L(x,0^\beta y) = L(x,y)$. Then we have $L(x',y') = R+2S + \ell + L(x,y)$. 
\end{lem}
\begin{proof}
Clearly, $L(x',y') \ge L(a,b) + L(0^\ell, 0^\ell) + L(x, y) = R+2S+\ell+L(x,y)$, since $L(a,b)=|b|=R+2S$. To prove the corresponding upper bound, we partition $y'=y_1 y_2$ such that $L(x',y') = L(a, y_1) + L(0^\ell x, y_2)$. Consider first the case that $y_2$ is a subsequence of $y$. Then 
\begin{align*}
 L(x',y')  \le L(a, y_1)  + L(0^\ell x, y) \le |a|+|y| = 2(R+S)+|y| \le (R+2S) + \ell + L(x,y),
\end{align*}
since $\ell \ge R+|y|$.

It remains to consider the case that $y_2$ is not a subsequence of $y$ and hence $y_1$ is a subsequence of $0^\beta b 0^{\ell}$. By \lemref{genDomPairs}\itemref{genDomPairsLmore}, we can without loss of generality assume that $y_1$ is a subsequence of~$0^\beta b$, since $L(a,0^\beta b 0^{\ell}) = |b| = L(a,0^\beta b)$. Hence, we can partition $0^\beta b = y_1 z$ with $L(x',y') \le L(a,y_1) + L(0^\ell x, z 0^\ell y)$.
We bound
\[ L(a, y_1) \le \min\{\occ_0(a), \occ_0(y_1)\} + \min\{\occ_1(a),\occ_1(y_1)\} \le (R+S) + \occ_1(y_1).\]

Observe that $L(0^\ell x, z 0^\ell y) \le \occ_1(z) + L(0^\ell x, 0^{\occ_0(z) + \ell} y)$, since each ``1'' in $z$ can increase the LCS by at most 1. By greedy prefix matching, we obtain $L(0^\ell x, z 0^\ell y) \le \occ_1(z) +  \ell + L(x, 0^{\occ_0(z)} y)$. With the assumption $L(x, 0^\beta y) = L(x,y)$ for any $\beta \ge 0$, this yields 
$$L(x',y') \le L(a,y_1) + L(0^\ell x, z 0^\ell y) \le (R + S + \occ_1(y_1)) + (\occ_1(z) +  \ell + L(x,y)) = R + 2S + \ell + L(x,y), $$
since $0^\beta b = y_1 z$ and hence $\occ_1(y_1) + \occ_1(z) = \occ_1(0^\beta b) = S$.

With the alternative assumption $S \ge |x|$, we bound $L(0^\ell x, z 0^\ell y) \le |0^\ell x| \le \ell + S$. If $\occ_1(y_1) = 0$ this yields
$$L(x',y') \le L(a,y_1) + L(0^\ell x, z 0^\ell y) \le (R + S) + (\ell + S) \le R + 2S + \ell + L(x,y). $$
Otherwise, if $\occ_1(y_1) \ge 1$, by inspecting the structure of $0^\beta b = 0^{\beta + R} (01)^S$ we see that $z$ is a subsequence of $(01)^{S-\occ_1(y_1)}$. Consider an LCS of $0^\ell x, z 0^\ell y$. If no ``1'' in $z$ is matched by the LCS, then we obtain 
$$ L(0^\ell x, z 0^\ell y) \le L(0^\ell x, 0^{\occ_0(z) + \ell} y) \le \occ_0(z) + L(0^\ell x, 0^{\ell} y) = \occ_0(z) + \ell + L(x,y), $$
where we used the fact $L(u,v w) \le |v| + L(uw)$ and greedy prefix matching.
Otherwise, if a ``1''  in $z$ is matched by the LCS to a ``1'' in $0^\ell x$, i.e., to a ``1'' in $x$, then the $0^\ell$-block of $0^\ell x$ is matched to a subsequence of $z$ and hence 
$$ L(0^\ell x, z 0^\ell y) \le L(0^\ell, z) + L(x, z 0^\ell y) \le \occ_0(z) + |x| \le \occ_0(z) + \ell + L(x,y), $$
where we used $\ell \ge |x|$.
Since $z$ is a subsequence of $(01)^{S-\occ_1(y_1)}$ and thus $\occ_0(z) \le S-\occ_1(y_1)$, in both cases we obtain
$$L(x',y') \le L(a,y_1) + L(0^\ell x, z 0^\ell y) \le (R + S+ \occ_1(y_1)) + (\occ_0(z) + \ell + L(x,y)) \le R + 2S + \ell + L(x,y), $$
which finally proves the claim.
\end{proof}

\begin{lem}\label{lem:bbbII-d}
Consider $x',y'$ as in \lemref{bbbII} with $\beta = 0$, and assume $S \ge |x|$. Then 
\[R\cdot S \le d(x',y') \le (R+1) (4R+6S + \ell) + d(x,y).\]
\end{lem}
\begin{proof}
Note that the simple fact $L(u'u'',v) \le |u''| + L(u',v)$ implies that for any strings~$w,z$ and any $i$ we have $L(w,z) \le |w[(i+1)..|w|]| + L(w[1..i],z) = |w|-i + L(w[1..i],z)$, and hence $L(w[1..i], z) \ge i - (|w|-L(w,z))$. 
Recall that $L(a,b) = |b|=R+2S$ by \lemref{genDomPairs}\itemref{genDomPairsLmore}. This yields $L(a[1..i], b)\ge i - (|a|-(R+2S)) = i-R$ and $L(a, b[1..j]) \ge j - (|b| - |b|) = j$.

The lower bound follows from $d(x',y') \ge d(a, b) \ge R\cdot S$ by \obsref{prefix} and \lemref{genDomPairs}\itemref{genDomPairsdmore}. For the upper bound, we consider all possible prefixes $\px:=x'[1..i], \py:=y'[1..j]$ of $x',y'$ and count how many of them correspond to dominant pairs. Clearly, for $i\le |a|, j \le |b|$, there are $d(a,b)$ dominant pairs. 

By the above observation, for any $i \le |a|$, we have $L(a[1..i], b) \ge i - R$. Hence, any dominant pair of the form $(i,j)$ satisfies $i-R \le L(a[1..i], b) \le |a[1..i]| = i$. By \obsref{atMostOned}, there are at most $R+1$ such dominant pairs for fixed $i$. Thus, there are at most $|a| \cdot (R+1)$ dominant pairs with $i \le |a|$ and $j > |b|$. 
Similarly, for $j \le |b|$, we have $L(a, b[1..j]) \ge j$. Hence, there are no dominant pairs with $i > |a|$ and $j \le |b|$, since already the prefix $a$ of $x'[1..i]$ includes $b[1..j]$ as a subsequence. In total, there are at most $d(a,b) + |a|\cdot (R+1)$ dominant pairs with $i\le |a|$ or $j\le |b|$. 

Let $i=|a| + k$, $j=|b|+k$ with $k \in [\ell]$. Then $L(\px,\py) = L(a0^k,b0^k) = L(a,b)+k=|b|+k$ by greedy suffix matching and \lemref{genDomPairs}\itemref{genDomPairsLmore}. As any such choice could correspond to a dominant pair, we count at most $\ell$ dominant pairs. Analogously to above, for $i = |a|+k$, there can be at most $i-L(a0^k, b0^k) \le R$ dominant pairs with $j > |b|+k$. Symmetrically, for $j = |b|+k$, there are at most $j- L(a0^k, b 0^k) = 0$ dominant pairs with $i > |a|+k$. This yields at most $(R+1)\cdot \ell$ dominant pairs with $|a|<i \le |a|+ \ell$ or $|b|< j \le |b|+\ell$.

It remains to count dominant pairs with $i= |a|+\ell+\tilde{i}$ and $j = |b| + \ell+\tilde{j}$, with $\tilde{i} \in [|x|], \tilde{j} \in [|y|]$. Here, \lemref{bbbII} bounds $L(\px, \py) = L(a0^\ell x[1..\tilde{i}], b0^\ell y[1..\tilde{j}]) = L(a,b) + \ell + L(x[1..\tilde{i}],y[1..\tilde{j}])$. Hence, the dominant pairs of this form are in one-to-one correspondence to the dominant pairs of $x,y$. 

Summing up all dominant pairs, we obtain
\begin{align*}
 d(x',y') & \le d(a,b) + (R+1) (|a|+\ell) + d(x,y) \\
& \le (R+1) (4R+6S + \ell) + d(x,y),
\end{align*}
since $|a| = 2R + 2S$ and \lemref{genDomPairs}\itemref{genDomPairsdmore} yields $d(a,b) \le 2(R + 1) (R+2S)$.
\end{proof}

Finally, to pad $\delta$ (and later, in \secref{DeltaSmallerm}, $\Delta$), we need the following technical lemma.
\begin{lem}\label{lem:deltapadding}
Let $x,y$ be arbitrary and $\mu,\nu$ such that $\nu \ge \mu + |y|$. We define
\begin{alignat*}{4}
x' & := 0^\mu& & \; 1^\nu \; 0^\mu\;  x,\\
y' & := & & \; 1^\nu \; 0^\mu \; y.
\end{alignat*}
Then $L(x', y') = \mu + \nu + L(x,y)$ and $d(x', y') = 2\mu + \nu + \occ_1(y) + d(x,y)$.
\end{lem}
\begin{proof}
Note that for any prefix $w, z$ of $0^\mu x, 0^\mu y$, we have 
\begin{equation}
L(0^\mu 1^\nu w, 1^\nu z) = \nu + L(w,z), \label{eq:eliminateOneBlock}
\end{equation}
 by \lemref{zeroblocklcs} (swapping the role of ``0''s and ``1''s). In particular, we obtain $L(x',y') = \nu + L(0^\mu x, 0^\mu y) = \mu +\nu + L(x,y)$ by greedy prefix matching.

For the second statement, we consider all possible prefixes $\px:=x'[1..i], \py:=y'[1..j]$ of $x',y'$ and count how many of them correspond to dominant pairs. Note that these prefixes have to end in the same symbol, since any dominant pair is a matching pair. Recall that $\px:=x'[1..i], \py:=y'[1..j]$ gives rise to a dominant pair if and only if $L(\px,\py) > L(x'[1..i-1],\py)$ and $L(\px,\py) > L(\px,y'[1..j-1])$.
\begin{itemize}
\item $\px = 0^\mu 1^\nu w, \py = 1^\ell$ (with $w$ non-empty prefix of $0^\mu x$, $\ell \in [\nu]$): These prefixes do not correspond to a dominant pair, since $L(0^\mu 1^\nu w, 1^\ell) = L(0^\mu 1^\ell, 1^\ell) = \ell$ is obtained already by a shorter prefix of $\px$.
\item $\px = 0^\mu 1^\nu w, \py = 1^\nu z$ (with $w$ non-empty prefix of $0^\mu x$, $z$ non-empty prefix of $0^\mu y$): These prefixes correspond to a dominant pair if and only if $w,z$ correspond to a dominant pair of $0^\mu x, 0^\mu y$, since by \eqref{eq:eliminateOneBlock} we have $L(\px, \py) = \nu + L(w, z)$. This yields $d(0^\mu x, 0^\mu y)$ dominant pairs, which by \lemref{greedy} evaluates to $\mu + d(x,y)$.
\item $\px = 0^k, \py = 1^\mu 0^\ell$ (with $k,\ell \in [\mu]$): Clearly, $L(\px,\py) = \min\{k, \ell\}$. It follows that $\px,\py$ corresponds to a dominant pair if and only if $k=\ell$. This yields exactly $\mu$ dominant pairs.
\item $\px = 0^\mu 1^k, \py = 1^\ell$ (with $k,\ell \in [\nu]$): Analogously to above, $L(\px,\py) = \min\{k, \ell\}$, hence this corresponds to a dominant pair if and only if $k = \ell$. This yields exactly $\nu$ dominant pairs.
\item $\px = 0^k, \py = 1^\nu 0^\mu z$ (with $k \in [\mu]$, $z$ non-empty prefix of $y$): We have $L(\px,\py) = L(0^k, 1^\nu 0^k) = k$, hence these prefixes do not correspond to a dominant pair, since the LCS is already obtained for a shorter prefix of $\py$.
\item $\px = 0^\mu 1^k, \py = 1^\nu 0^\mu z$ (with $k\in [\nu]$, $z$ non-empty prefix of $y$): Since we can either match some ``1'' of the $1^\nu$-block in $\py$ to a ``1'' in $\px$ (necessarily discarding the initial $0^\mu$-block of $\px$) or delete the complete $1^\nu$-block, we obtain
\begin{align*}
L(0^\mu 1^k, 1^\nu 0^\mu z) & =  \max\{L(1^k, 1^\nu 0^\mu z), L(0^\mu 1^k, 0^\mu z)\} \\
& =  \max\{k, \mu + L(1^k, z) \} = \max\{k, \mu + \min\{k, \occ_1(z)\}\}.
\end{align*}
Consider the case $L(\px, \py) = L(x'[1..i],y'[1..j]) = k$. Then also $L(x'[1..i],y'[1..(j-1)]) = k$, and hence $(i,j)$ is no dominant pair. If, however, $L(\px, \py) = \mu + \min\{k, \occ_1(z)\}$, then this corresponds to a dominant pair if and only if $k=\occ_1(z)$ (and $z$ ends in ``1''): if $k > \occ_1(z)$, then also $L(x'[1..(i-1)],y'[1..j]) = \mu + \occ_1(z)$, if $k < \occ_1(z)$, then also $L(x'[1..i],y'[1..(j-1)])= \mu + k$. Thus, there are exactly $\min\{\nu, \occ_1(y)\} = \occ_1(y)$ such dominant pairs.
\end{itemize}

In total, we have counted $\nu + 2\mu + \occ_1(y) + d(x,y)$ dominant pairs.
\end{proof}

We start with the basic building block from \lemref{bbbII} and then further pad the strings to obtain the desired $n,m,\Delta,\delta,M$ as follows. Note that the guarantee $|y| \le |x| \le \Oh(\min\{\Delta, m\})$ is satisfied by \lemref{largeLCShardnessConstantAlphabet}.

\begin{lem} \label{lem:largelcsternaryalphabets-step}
Let $(\Valpha,\{0,1,2\})$ be a parameter setting satisfying \tabref{paramChoiceRestr} with $\alpha_L=\alpha_m$. Let $(n,x,y)$ be an instance of $\LCS_\le^\gamma(\Valpha, \{0,1\})$ with $|y| \le |x| \le \Oh(\min\{\Delta, m\})$. We set
\begin{align*}
S & = \max\{m,|x|\}, & R& = \lfloor d/m \rfloor, 
\end{align*}
to instantiate the basic building block $x'= a 0^\ell x = (01)^{R+S} 0^\ell x$ and $y'= b 0^\ell y = 0^R (01)^S 0^\ell y$ of \lemref{bbbII} with $\ell := R+S+|x| + |y|$. Moreover, we define  $\kappa := \lfloor M/n \rfloor$ and $\tilde{m}:=\max\{m, \delta + 2R+2S+\ell+|x|\}$ to further pad the instance to 
\begin{alignat*}{6}
x'' &\; =  \quad 2^{\kappa} & \; 2^\Delta & \; 1^{\tilde{m}} & \; 0^\delta & \;  a \; &0^\ell \; x, \\
y'' &\; =  \quad 2^{\kappa} & \; 0^\delta &  \; 1^{\tilde{m}} & \; 0^\delta & \; b \; & 0^\ell \; y.
\end{alignat*}
Then $x'',y''$ is an instance of $\LCS^{\gamma'}(\Valpha,\{0,1,2\})$ for some constant $\gamma'\ge 1$ and can be computed in time $\Oh(n)$, together with an integer $\tau$ such that $L(x'',y'') = \tau+L(x,y)$.
\end{lem}
\begin{proof}
We first use the Disjoint Alphabets Lemma and greedy prefix matching to obtain
\begin{equation}\label{eq:largeLCSternarylcs}
L(x'',y'') = \kappa + L(1^{\tilde{m}} 0^\delta x', 0^\delta 1^{\tilde{m}} 0^\delta y') = \kappa + \tilde{m} + \delta + L(x',y') = \kappa + \tilde{m} + \delta + R+2S + \ell + L(x,y),
\end{equation}
where we used \lemref{deltapadding} for the second equality (with the roles of $x,y$ swapped) and \lemref{bbbII} for the last equality.
Observe that $x'',y''$ and $\tau = \kappa + \tilde{m} + \delta +  R+2S + \ell$ can be computed in time $\Oh(n)$.

It remains to verify that $x,y$ is an instance of $\LCS^{\gamma'}(\Valpha,\{0,1,2\})$ for some $\gamma'\ge 1$. We first observe that $S = \Theta(m)$ by $|x| = \Oh(m)$ and $R = \Oh(S)$ by the parameter relation $d\le Lm \le m^2$. Since also $|y| = \Oh(m)$, we conclude $R,S,|x|,|y|,\ell = \Oh(m)$. Observe that $\kappa = \Oh(M/n) = \Oh(m)$ by the relation $M\le mn$ and $\tilde{m} = \Theta(m)$ (using that $R,S,\ell, |x|= \Oh(m)$ and the parameter relation $\delta \le m$). Thus, $|x''| = \kappa + \Delta + \tilde{m} + \delta + 2(R+S) + \ell + |x| = \Theta(m + \Delta) + \Oh(m) = \Theta(n)$. Similarly, $|y''| = \kappa + \tilde{m} + \delta + R + 2S + \ell + |y| = \Theta(m + \delta) + \Oh(m)= \Theta(m)$. \dopara{n,m,L} By \eqref{eq:largeLCSternarylcs}, we also have $L(x'',y'') = \tilde{m} + \Oh(m) = \Theta(m)$, as desired.

Note that by $\Delta \ge \delta$, $|a| \ge |b|$ and $|x| \ge |y|$, we indeed have $|x''| \ge |y''|$.  Furthermore, using the relation $d\le 2L(\Delta+1) = \Oh(m\Delta)$, we obtain $R = \Oh(d/m) = \Oh(\Delta)$.  By~\eqref{eq:largeLCSternarylcs}, we obtain $\Delta(x'',y'') = \Delta + R + (|x| - L(x,y)) = \Delta + \Oh(\Delta) = \Theta(\Delta)$ since $|x|= \Oh(\Delta)$. Similarly, \eqref{eq:largeLCSternarylcs} yields $\delta(x'',y'') = \delta + (|y|-L(x,y)) = \delta + \Oh(\delta) = \Theta(\delta)$, since $|y|-L(x,y) = \delta(x,y) = \Oh(\delta)$.  \dopara{\delta,\Delta}

For the dominant pairs, we apply the disjoint alphabets lemma, \lemref{greedy} and \lemref{deltapadding} to compute\dopara{d}
\begin{equation}\label{eq:ternaryd}
d(x'',y'') = \kappa + d(1^{\tilde{m}} 0^\delta x', 0^\delta 1^{\tilde{m}} 0^\delta y') = \kappa + \tilde{m} + 2\delta + d(x',  y').
\end{equation}
\lemref{bbbII-d} yields the lower bound $d(x',y') \ge R\cdot S = \Omega(d)$ and the corresponding upper bound 
\[d(x',y') \le (R+1) (4R+6S + \ell) + d(x,y) = \Oh( R\cdot S + d(x,y) ) = \Oh(d).\]
Thus, \eqref{eq:ternaryd} yields $d(x'',y'') = \Theta(d) + \Oh(m) = \Theta(d)$, where we used that $d\ge L = \Omega(m)$ since $\alpha_L=\alpha_m$. 

For $M$, we count $\occ_2(x'') = \Delta + \kappa$ and $\occ_2(y'') = \kappa$, as well as $\occ_0(y''),\occ_1(y'') \le |y''| = \Oh(m)$, $\occ_0(x'') \le \delta+|x'| = \Oh(m)$ and $\occ_1(x'') \le \tilde{m} + |x'| = \Oh(m)$. Thus, $M(x'',y'') = (\Delta + \kappa) \kappa + \Oh(m^2)$ and by $M\ge L^2/|\Sigma| = \Omega(m^2)$ (since $\alpha_L=\alpha_M$), it suffices to prove that $(\Delta+\kappa)\kappa = \Theta(M)$ to verify $M(x'',y'') = \Theta(M)$. Indeed, we have $\kappa = \Theta(M/n)$, since $M\ge n$. If $\alpha_\Delta < 1$, we have $\alpha_m=\alpha_n=1$ and $M = \Omega(m^2)$ together with the relation $M\le mn = \Oh(m^2)$ implies $M=\Theta(m^2)$. Thus, $\kappa = \Theta(m)$ and hence $(\kappa+\Delta)\kappa = \Theta(m^2) = \Theta(M)$. If $\alpha_\Delta \ge \alpha_m$, then $\alpha_\Delta = 1$ and hence $\Delta+\kappa = \Delta + \Oh(m) = \Theta(n)$, which implies $(\kappa + \Delta) \kappa = \Theta( n \cdot M/n) = \Theta(M)$. Finally, note that indeed $\Sigma(x'',y'')= \{0,1,2\}$.
\end{proof}

Combining \lemref{largelcsternaryalphabets-step} with \lemref{largeLCShardnessConstantAlphabet} finally proves \lemref{largelcsternaryalphabets}.

\subsection{Large LCS, Alphabet Size 2}

In this section, we study the case that $\alpha_{L}=\alpha_m$ (and $\alpha_{\delta},\alpha_{\Delta}$ may be small) for the case of binary alphabets, i.e., $\Sigma = \{0,1\}$. In this regime, \thmref{main2} follows from the following statement (and \lemref{paramsnecessary}).

\begin{lem}\label{lem:largelcsbinaryalphabets}
Let $(\Valpha,\{0,1\})$ be a parameter setting satisfying \tabref{paramChoiceRestr} with $\alpha_L=\alpha_m$. There is a constant $\gamma \ge 1$ such that any algorithm for $\LCS^\gamma(\Valpha,\{0,1\}$ takes time $\min\{d, \delta \Delta, \delta M/n\}^{1-o(1)}$ unless \OVH\ fails.
\end{lem}

We present different constructions for three subcases, that we discuss shortly in the following paragraphs and in detail in the remainder of this section. We hope that this short discussion conveys enough intuition about the ``complexity'' of the task to make it believable that our lengthy and technical case distinction is indeed necessary.
  
\emph{Case 1:} $\alpha_\Delta \le \alpha_m = \alpha_L$. Then $n = \Theta(m)$ and it follows that \emph{any} binary strings $x,y$ satisfy $M(x,y) = \Theta(m^2)$, so $M$ poses no constraints, in particular there are no constraints on the numbers of ``0''s and ``1''s in the constructed strings. On the other hand, the potentially small value of $\Delta$ renders some of our gadgets useless (e.g., \lemref{advancedDeltapadding}).
Since $\delta$ may be small, we use the hardness construction from \secref{hardnesslargeLCS} (for large LCS). 

Otherwise, we have $\alpha_\Delta > \alpha_m$ and thus $\Delta = \Theta(n) \gg m$. Note that any string $x$ of length $n$ contains at least $n/2$ ``0''s or ``1''s, say $x$ contains many ``1''s. Then to obtain $\Theta(M)$ matching pairs, $y$ must contain at most $\Oh(M/n)$ ``1''s. Thus, we need to pay close attention to the number of ``1''s in the constructed string $y$. 
We split the case $\alpha_\Delta > \alpha_m$ into two subcases. \emph{Case 2: $\alpha_\Delta > \alpha_m = \alpha_L$ and $\alpha_\delta \ge \alpha_M - 1$}. Here, the constraint on $\occ_1(y)$ is stronger than the constraint on $\delta$, and we use the hardness construction from \secref{hardnesssmallLCS} (for small LCS), since it introduces few ``1''s. \emph{Case 3: $\alpha_\Delta > \alpha_m = \alpha_L$ and $\alpha_\delta < \alpha_M - 1$}. Here, the constraint on $\delta$ is stronger than the constraint on $\occ_1(y)$, and we use the hardness construction from \secref{hardnesslargeLCS} (for large LCS), since it keeps $\delta$ small.

\subsubsection{Case $\alpha_\Delta \le \alpha_m = \alpha_L$}
\label{sec:DeltaSmallerm}

Since $n = \Delta + L$, the assumptions $\alpha_L = \alpha_m$ and $\alpha_\Delta \le \alpha_m$ imply $n = \Theta(m)$. Together with the parameter relations $L^2 / |\Sigma| \le M \le 2Ln$ and $\Sigma = \{0,1\}$ we obtain $M = \Theta(m^2)$. In particular, in this regime the $\tOh(\delta \Delta)$ time bound beats $\tOh(\delta M/n)$, and \lemref{largelcsbinaryalphabets} simplifies to the following result.

\begin{lem}\label{lem:DeltaSmallerm}
Let $(\Valpha,\Sigma)$ be a parameter setting satisfying \tabref{paramChoiceRestr} with $\alpha_L=\alpha_m$ and $\alpha_\Delta \le \alpha_m$. There is a constant $\gamma \ge 1$ such that any algorithm for $\LCS^\gamma(\Valpha,\Sigma)$ takes time $\min\{d, \delta \Delta\}^{1-o(1)}$ unless \OVH\ fails.
\end{lem}

We can now instantiate the parameters of \lemref{bbbII} to create a string with the desired number of dominant pairs. The remaining parameters will be padded in an additional construction. Note that the preconditions, specifically the additional guarantee, are satisfied by \lemref{largeLCShardnessConstantAlphabet}.

\begin{lem}\label{lem:DeltaSmallerm-step1}
Let $(\Valpha, \Sigma)$ be a parameter setting satisfying \tabref{paramChoiceRestr} with $\alpha_L=\alpha_m$ and $\alpha_\Delta \le \alpha_m$. Given any instance $(n,x,y)$ of $\LCS_\le^\gamma(\Valpha, \{0,1\})$ with the additional guarantee $|y|\le |x| \le \gamma \cdot m$, we can construct an instance $(n,x',y')$ of $\LCS^{\gamma'}_\le(\Valpha, \{0,1\})$ (for some constant $\gamma'\ge \gamma$) and $\tau$ in time $\Oh(n)$ such that 
\begin{enumerate}[label=(\roman{*})]
\item $L(x',y') = \tau+L(x,y)$, \label{itm:DeltaSmallerm-L}
\item $d(x',y') = \Theta(d)$. \label{itm:DeltaSmallerm-d}
\item $|x'| \ge |y'|$. \label{itm:DeltaSmallerm-size}
\end{enumerate} 
\end{lem}
\begin{proof}
We construct $x',y'$ as in \lemref{bbbII} with $S= \max\{|x|, m\}$, $R = \lceil d/S\rceil$, $\beta = 0$ and  $\ell = (R+S+|x|+|y|)$. Note that indeed $|x'| \ge |y'|$ by $|x| \ge |y|$, $|a| \ge |b|$, and $\beta=0$. The assumption $|x|,|y| = \Oh(m)$ yields $S= \Theta(m)$. By the parameter relation $d\le 2(\Delta +1)\cdot L = \Oh(\Delta \cdot m)$, we also have $R= \Oh(d/S + 1) = \Oh(d/m + 1) = \Oh(\Delta)$. We conclude that $|x'|, |y'| = \Oh(R+S+|x|+|y|) = \Oh(m + \Delta) = \Oh(m)$, since by assumption $\alpha_\Delta \le \alpha_m$. This yields $L(x',y') = \Oh(m) = \Oh(L)$ (by the assumption $\alpha_L = \alpha_m$) and $M(x',y') = \Oh(m^2) = \Oh(M)$ (note that $M \ge L^2/|\Sigma| = \Omega(m^2)$ by the assumption $\alpha_L = \alpha_m$).

By \lemref{bbbII}, we have $L(x',y') = R+2S +\ell + L(x,y)$, satisfying~\itemref{DeltaSmallerm-L}. 
This yields $\delta(x',y') = \delta(x,y) = \Oh(\delta)$ and $\Delta(x',y') = R + \Delta(x,y) = \Oh(d/m + \Delta) = \Oh(\Delta)$, by the parameter relation $d \le 2L(\Delta+1) = \Oh(m \Delta)$. 
For $d$, we first observe that $\lceil d/S \rceil = \Theta(d/m)$ (by the parameter relation $d\ge L = \Theta(m)$) and $\ell = \Oh(R+S+|x|+|y|) = \Oh(m)$. \lemref{bbbII-d} yields the lower bound $d(x',y') \ge R\cdot S = \Omega(d/m \cdot m) = \Omega(d)$ as well as the corresponding upper bound $d(x',y') = \Oh( R \cdot \ell + d(x,y) ) = \Oh(d/m \cdot m + d) = \Oh(d)$.

These bounds prove that $(n,x',y')$ is an instance of $\LCS^{\gamma'}_\le(\Valpha, \{0,1\})$ for some $\gamma'\ge 1$.
\end{proof}

We use \lemref{deltapadding} to finally pad $\delta, \Delta$, and $m$.

\begin{lem}\label{lem:DeltaSmallerm-step2}
Let $x,y,x',y',\tau$ be as given in \lemref{DeltaSmallerm-step1}. Then, in time $\Oh(n)$ we can construct an instance $x''',y'''$ of $\LCS^{\gamma''}(\Valpha,\{0,1\})$ (for some constant $\gamma'' \ge 1$) and an integer $\tau'$ such that $L(x''',y''') = \tau' + L(x',y') = (\tau+\tau') + L(x,y)$.
\end{lem}
\begin{proof}
As an intermediate step, let $\tilde{m} := \max\{m, |x'|, |y'|, \Delta\}$ and construct
\begin{alignat*}{4}
x'' & := 0^\Delta & & \; 1^{\tilde{m}+\Delta} \; 0^\Delta \; x', \\
y'' & := & & \; 1^{\tilde{m}+\Delta} \; 0^\Delta \; y'.
\end{alignat*}
We obtain the final instance as
\begin{alignat*}{4}
x''' & := & & \; 1^{5\tilde{m}+\delta} \; 0^\delta \; x'', \\
y''' & := 0^\delta & & \; 1^{5\tilde{m}+\delta} \; 0^\delta \; y''.
\end{alignat*}
Note that by definition of $\tilde{m}$, we have $\tilde{m}+\Delta \ge \Delta + |y'|$  and $\delta + 5\tilde{m} \ge \delta + (3\Delta + \tilde{m} + |x'|) = \delta + |x''|$, satisfying the conditions of \lemref{deltapadding}. Hence, this lemma yields
\begin{equation}\label{eq:Deltasmallerm-lcs}
 L(x''',y''') = 5\tilde{m}+2\delta + L(x'',y'') = 6\tilde{m}+2\delta + 2\Delta + L(x',y').
\end{equation}
Clearly, $x',y',\tau$ and hence also $x'',y''$ and $\tau':= 6\tilde{m}+2\delta+2\Delta$ can be computed in time $\Oh(n)$.

We now verify all parameters.\dopara{n,m,L} Clearly, $\Sigma(x''',y''') = \{0,1\}$. Since by assumption $\alpha_\Delta \le \alpha_m = 1$, we have $|x'| = \Oh(n) = \Oh(m)$, $|y'| = \Oh(m)$, and $\Delta = \Oh(m)$. This implies $\tilde{m} = \Theta(m)$ and consequently $|x'''| = 6\tilde{m} + 2 \delta + 3 \Delta + |x'| = 6\tilde{m} + \Oh(m) = \Theta(m)=\Theta(n)$ (using $\delta \le \Delta$ and $\alpha_\Delta \le \alpha_m =1$). Similarly, $|y'''| = 6\tilde{m} + 3\delta + 2\Delta  + |y'| = 6\tilde{m} + \Oh(m) = \Theta(m)$. By \eqref{eq:Deltasmallerm-lcs}, we have $L(x''',y''')=6\tilde{m} + 2\delta + 2\Delta +  L(x',y') = 6\tilde{m} + \Oh(L) = \Theta(L)$ (using the assumption $\alpha_L=\alpha_m$).

Note that $|x'''| \ge |y'''|$ follows from $\Delta \ge \delta$ and $|x'|\ge |y'|$ (by \lemref{DeltaSmallerm-step1}\itemref{DeltaSmallerm-size}). Hence, \eqref{eq:Deltasmallerm-lcs} provides $\Delta(x''',y''') = \Delta + \Delta(x',y') = \Theta(\Delta)$ and $\delta(x''',y''') = \delta + \delta(x',y') = \Theta(\delta)$.\dopara{\delta,\Delta}

The number of matching pairs satisfies $M(x''',y''') = \occ_0(x''')\occ_0(y''') + \occ_1(x''')\occ_1(y''') = \Theta(m^2) = \Theta(M)$, where the last bound follows from $M\ge L^2/|\Sigma| = \Omega(m^2)$ and $M \le 2Ln = \Oh(m^2)$ by $\alpha_L = \alpha_m = 1$.\dopara{M,d} For the number of dominant pairs, we apply \lemref{deltapadding} to bound 
\begin{align*}
 d(x''',y''') &= 3\delta + 5\tilde{m} + \occ_1(x'') + d(x'',y'') \\
& = 3\delta + 5\tilde{m} + (\tilde{m} + \Delta + \occ_1(x')) + (3\Delta + \tilde{m} + \occ_1(y') + d(x',y')) \\
& = 7\tilde{m} + 3(\delta + \Delta) + \occ_1(x') + \occ_1(y') + d(x',y') = \Theta(d) + \Oh(m) = \Theta(d),
\end{align*}
where the last two bounds follow from $\tilde{m}, |x'|,|y'|, \delta, \Delta = \Oh(m)$, $d(x',y') = \Theta(d)$ by \lemref{DeltaSmallerm-step1} and the parameter relation $d\ge L = \Omega(m)$.
\end{proof}

Combining \lemrefs{DeltaSmallerm-step1}{DeltaSmallerm-step2} with \lemref{largeLCShardnessConstantAlphabet} finally proves \lemref{DeltaSmallerm}.

\subsubsection{Case $\alpha_\Delta > \alpha_m = \alpha_L$ and $\alpha_\delta \ge \alpha_M - 1$}

In this section, we consider the case where $\alpha_L=\alpha_m < \alpha_\Delta$ and $\alpha_\delta \ge \alpha_M - 1$. In this case, we have $\Delta = \Theta(n) \gg m$ and $M/n \le  \Oh(\delta)$. Since $M/n \le \Oh(m) \le \Oh(\Delta)$, the fastest known algorithm runs in time $\tOh(\min\{d, \delta M/n\})$. Consequently, in this regime \lemref{largelcsbinaryalphabets} simplifies to the following statement.

\begin{lem}\label{lem:MnSmallerdelta}
Let $(\Valpha,\{0,1\})$ be a parameter setting satisfying \tabref{paramChoiceRestr} with $\alpha_L=\alpha_m < \alpha_\Delta$ and $\alpha_M - 1\le \alpha_\delta$. There is a constant $\gamma \ge 1$ such that any algorithm for $\LCS^\gamma(\Valpha,\Sigma)$ takes time $\min\{d, \delta M/n\}^{1-o(1)}$ unless \OVH\ fails.
\end{lem}

To obtain this result, we cannot simply pad our hard instances $(n,x,y)$ of $\LCS_\le(\Valpha, \{0,1\})$ to $\LCS(\Valpha, \{0,1\})$, since the desired running time bound $\min\{d, \delta M/n\}$ is not monotone. In other words,
for $\LCS_\le(\Valpha, \{0,1\})$ we have a lower bound of $\min\{\delta \Delta, \delta m, d\}^{1-o(1)}$ (see \lemref{largeLCShardnessConstantAlphabet}) which can be higher than the running time $\Oh(n + \delta M/n)$ of our new algorithm (\thmref{algo}) and thus we would violate SETH. In fact, we even cannot start from an instance as constructed in \lemref{largeLCShardnessConstantAlphabet}, since this would generate too many ``1''s. Instead, we use instances of a different parameter setting $\LCS_\le(\Valpha',\{0,1\})$ with $\alpha_\delta' = \alpha_m'$, i.e., we invoke \lemref{smallLCShardnessConstantAlphabet}.

\begin{obs} \label{obs:MnConstructionalphaprime}
Let $(\Valpha,\{0,1\})$ be a parameter setting satisfying \tabref{paramChoiceRestr} with $\alpha_L = \alpha_m < \alpha_\Delta$ and $\alpha_M - 1\le \alpha_\delta$. Then $\Valpha':=\Valpha'(\Valpha)$ defined by
\begin{align*}
\alpha'_d &= \min\{\alpha_d, \alpha_\delta + \alpha_M - 1\}, & \alpha'_M & = 2\alpha'_L, \\
\alpha'_d - \alpha'_L & = \min\{\alpha_M - 1, \alpha_d/2\}, & \alpha'_\Delta& = 1, \\
\alpha'_m &= \alpha'_\delta = \alpha'_L,  & \alpha'_\Sigma & = 0,
\end{align*}
yields a parameter setting $(\Valpha', \{0,1\})$ satisfying \tabref{paramChoiceRestr}. The definition of $\Valpha'$ implies 
\begin{equation}\label{eq:alphaL}
\alpha'_L = \min\{\alpha_\delta, \max\{\alpha_d - \alpha_M + 1, \alpha_d/2\}\}.
\end{equation}
Moreover, there is some constant $\gamma \ge 1$ such that no algorithm solves $\LCS^\gamma_\le(\Valpha', \{0,1\})$ in time $\min\{d, \delta M/n\}^{1-o(1)}$ unless \OVH\ fails.
This holds even restricted to instances $(n,x,y)$ with $|x|,|y| \le \gamma n^{\alpha'_L} = \Oh(\min\{\delta, \max\{d n / M, \sqrt{d}\}\})$ and $\occ_1(y) \le \gamma \cdot n^{\alpha'_d - \alpha'_L} = \Oh(\min\{M/n, \sqrt{d}\})$ satisfying $L(x,0^\beta y) = L(x,y)$ for any $\beta \ge 0$.
\end{obs}
\begin{proof}
We first prove \eqref{eq:alphaL}. Consider the case that $\alpha_d/2 \le \alpha_M - 1$. Then $\alpha_d \le 2(\alpha_M - 1) \le \alpha_\delta + (\alpha_M - 1)$, where we used the assumption $\alpha_M - 1 \le \alpha_\delta$. Thus, $\alpha'_d = \alpha_d$ and  by definition 
\[\alpha_L' = \alpha'_d - \min\{\alpha_M - 1, \alpha_d/2\} = \alpha_d - \alpha_d/2 = \alpha_d/2.\]
From $\alpha_d/2 \le \alpha_M - 1$, it follows that $\alpha_d - \alpha_M + 1 \le \alpha_d/2$ and $\alpha_d/2 \le \alpha_M - 1 \le \alpha_\delta$, hence $\alpha'_L = \alpha_d/2 = \min\{\alpha_\delta, \max\{\alpha_d - \alpha_M + 1, \alpha_d/2\}\}$, as desired. 

Consider the remaining case that $\alpha_d /2 > \alpha_M - 1$. Then by definition
\[ \alpha_L' = \alpha'_d - (\alpha_M-1) = \min\{\alpha_d - \alpha_M + 1, \alpha_\delta\}. \]
Since $\alpha_d /2 > \alpha_M -1$ implies $\alpha_d - \alpha_M + 1 \ge \alpha_d/2$, this indeed yields
\[ \alpha_L' = \min\{ \max\{\alpha_d - \alpha_M+1, \alpha_d/2\}, \alpha_\delta\},\]
as desired, concluding the proof of \eqref{eq:alphaL}.

Checking all constraints from \tabref{paramChoiceRestr} is straight-forward, except for the inequalities $0 \le \alpha'_L \le 1$ and $\alpha'_d \le 2\alpha'_L$. 
From \eqref{eq:alphaL}, $0 \le \alpha'_L \le 1$ follows immediately by the parameter relations $\alpha_\delta, \alpha_d \ge 0$ and $\alpha_\delta \le \alpha_m \le 1$. For the other inequality, note that $\alpha'_d \le 2\alpha'_L$ is equivalent to $\min\{\alpha_M - 1, \alpha_d/2\} = \alpha'_d - \alpha'_L \le \alpha'_L = \min\{\alpha_\delta, \max\{\alpha_d-\alpha_M + 1, \alpha_d/2\}\}$, which directly follows from the assumption $\alpha_M -1 \le \alpha_\delta$ and the trivial fact that $\alpha_d/2\le \max\{\alpha_d-\alpha_M+1, \alpha_d/2\}$.

The last statement directly follows from \lemref{smallLCShardnessConstantAlphabet}.
\end{proof}

It remains to pad strings $x,y$ of $\LCS_\le(\Valpha', \{0,1\})$ to $\LCS(\Valpha, \{0,1\})$. The first step is the following construction which pads $\delta$ and $d$. 

\begin{lem}\label{lem:MnSmallerdelta-step1}
Let $(\Valpha,\{0,1\})$ be a parameter setting satisfying \tabref{paramChoiceRestr} with $\alpha_L = \alpha_m < \alpha_\Delta$ and $\alpha_M - 1\le \alpha_\delta$, and construct $\Valpha'$ as in \obsref{MnConstructionalphaprime}. Let $(n,x,y)$ be an instance of $\LCS_\le(\Valpha', \{0,1\})$ with $|x|,|y| = \Oh(\min\{\delta, \max\{d n /M, \sqrt{d}\}\})$ and $\occ_1(y) = \Oh(\min\{M/n, \sqrt{d}\})$ satisfying $L(x,0^\beta y) = L(x,y)$ for any $\beta \ge 0$.
We set
\begin{align*}
S &= \lfloor \min\{M/n, \sqrt{d}\}\rfloor, & R &=  \lfloor d/S \rfloor,  & \ell& = |x|+|y|+R+S, & \beta & = \delta,
\end{align*}
and define, as in \lemref{bbbII},
\begin{alignat*}{10}
x' & \; := \quad & & a && \; 0^\ell\; x \quad &=& \quad & (01)^{R+S} \; &0^\ell\; x, \\
y' & \; := \quad & 0^\beta \;& b && \; 0^\ell\; y \quad &=& \quad 0^\beta \; & 0^R (01)^S \; &0^\ell\; y.
\end{alignat*}
Then $x',y'$ is an instance of $\LCS^{\gamma'}_\le(\Valpha,\{0,1\})$ (for some $\gamma'\ge 1$) with the additional guarantees that
\begin{enumerate}[label=(\roman{*})]
\item \label{itm:MnSmallerdelta-step1size} $|x'|,|y'|= \Oh(m)$,
\item \label{itm:MnSmallerdelta-step1occ1} $\occ_1(y') = \Oh(M/n)$ and $\occ_1(y') \cdot |b0^\ell y|  =\Oh(d)$.
\item \label{itm:MnSmallerdelta-step1d} $d(x',y') = \Theta(d)$,
\item \label{itm:MnSmallerdelta-step1delta} $|y'|-L(x',y') = \Theta(\delta)$,
\item \label{itm:MnSmallerdelta-step1tau} In time $\Oh(n)$ we can compute a number $\tau$ such that $L(x',y') = \tau + L(x,y)$.
\end{enumerate} 
\end{lem}
\begin{proof}
Note that $S \le \sqrt{d}$ implies that $S \le R$. Note also that by assumption, $|x|,|y| = \Oh(\min\{\delta,R\})$ and hence $R,S,\ell,|x|,|y| = \Oh(R)$. Furthermore, $R = \Theta(d/S) = \Theta(d/M n + \sqrt{d}) = \Oh(m)$, where the first bound follows from $\alpha_d \ge 0$ and $\alpha_M \ge 1$and the second follows from the parameter relations $d/M n \le 5L = \Oh(m)$ and $d \le Lm \le m^2$. Hence $|x'| = \Oh(R) = \Oh(m)$. Additionally, $\tilde{\delta} = \Theta(\delta)$ and thus $|y'|=\Oh(\delta + m) = \Oh(m)$ by $\delta \le m$; thus we have proven~\itemref{MnSmallerdelta-step1size}. This also implies $L(x',y') \le \Oh(m) = \Oh(L)$ since $\alpha_L=\alpha_m$. \dopara{n,m, L}

Note that $\occ_1(y) = S+\occ_1(y) = \Oh(S)$ by definition of $S$ and assumption on $\occ_1(y)$.\dopara{d,M} We compute $d(x',y')\le 5L(x',y') \cdot \occ_1(y')\le 5|x'| \cdot \occ_1(y') = \Oh(R\cdot S) = \Oh(d)$. Furthermore, we obtain by \obsref{prefix} and \lemref{genDomPairs} that $d(x',y') \ge d(a,0^\beta b) \ge R\cdot S = \Omega(d)$. Note that in particular $\occ_1(y') = \Oh(S) = \Oh(M/n)$ and $\occ_1(y') \cdot |b0^\ell y| = \Oh(S \cdot R) = \Oh(d)$, proving \itemref{MnSmallerdelta-step1occ1}.
The bound $M(x',y') \le \Oh(m^2) = \Oh(M)$ trivially follows from $|x'|,|y'| = \Oh(m)$ and $M\ge L^2/|\Sigma| = \Omega(m^2)$ since $\alpha_L=\alpha_m$. 

By \lemref{bbbII} we have $L(x',y') = R + 2S + \ell + L(x,y)$, which immediately yields \itemref{MnSmallerdelta-step1tau}. Moreover, we obtain $\delta(x',y') = \delta + \delta(x,y) = \Theta(\delta)$, since $\delta(x,y) \le |x| \le \Oh(\delta)$, proving~\itemref{MnSmallerdelta-step1delta}. Finally, $\Delta(x',y') \le |x'| \le \Oh(m) \le \Oh(\Delta)$ by the assumption $\alpha_\Delta > \alpha_m$.
\end{proof}

To finally pad all remaining parameters, we first prepare the following technical tool.

\begin{lem}\label{lem:deleteDeltaPadding}
Let $x = 1^\kappa 0^{\mu} w$ and $y = 0^{\mu} 0^\nu z$  with $\mu > |z|$. Then it holds that $L(x,y) = \mu + L(w,0^\nu z)$, as well as $d(x,y) \ge d(w,0^\nu z)$ and $d(x,y) \le \min\{\kappa,|z|\} + \mu  + d(1^\kappa 0^\mu, z) + d(w,0^\nu z)$.
\end{lem}
\begin{proof}
By \lemref{zeroblocklcs}, we have that $L(x,y) = \mu + L(w,0^\nu z)$. 

This immediately shows that $d(x,y) \ge d(w,0^\nu z)$, since the above statement implies, for any prefixes $\tilde{w}, \tilde{z}$ of $w,0^\nu z$,  that $L(1^\kappa 0^\mu \tilde{w}, 0^\mu \tilde{z}) = \mu + L(\tilde{w},\tilde{z})$ and hence any $k$-dominant pair $(i,j)$ of $w$ and $0^\nu z$ gives rise to a $(\mu+k)$-dominant pair $(\kappa+\mu+i,\mu+j)$ of $x$ and $y$.

For the upper bound, we count the number of prefixes $\px, \py$ of $x,y$ corresponding to dominant pairs. Note that $\px,\py$ have to end in the same symbol to be a dominant pair. Consider first the case that $\px = 1^k$. Hence we must have $\py = 0^\mu 0^\nu \pz$ for some prefix $\pz=z[1..\ell]$ of $\tilde{z}$. Clearly, $L(\px,\py) = \min\{k, \occ_1(\pz)\}$. Hence, $\px,\py$ corresponds to a dominant pair if and only if $\occ_1(\pz) =\occ_1(z[1..\ell])= k$ and $\occ_1(z[1..\ell-1]) < k$, i.e., $\pz$ is determined by the $k$-th occurrence of ``1'' in $z$. Thus, there can be at most $\min\{\kappa,\occ_1(z)\}\le \min\{\kappa,|z|\}$ such dominant pairs. 

Consider the case that $\px = 1^\kappa 0^k$ with $k\in [\mu]$. We separately regard the following types of prefixes of $y$. 
\begin{itemize}
\item $\py = 0^\ell$ with $\ell \in [\mu+\nu]$: By greedy suffix matching, $L(\px,\py)=L(1^\kappa 0^{k},0^\ell)=\min\{k,\ell\}$, hence as above there can be at most $\mu$ dominant pairs, since there are only $\mu$ choices for $k$.  
\item $\py = 0^\mu 0^\nu \pz$: We have $L(\px,\py) = \max\{k, L(1^\kappa 0^{k},\pz)\}$. To see that this holds, note that the longest common subsequence either includes none of the ones of $1^\kappa$ of $\py$, in which case $0^{k}$ is the LCS of $\py$ and $\px$, or otherwise it matches at least one 1 in $\py$, which means that the LCS deletes all ``0''s preceding the first ``1'' in $\py$, i.e., the whole $0^{\mu+\nu}$ block in $y'$. 

If $L(\px, \py)=k$, then $\px,\py$ cannot correspond to a dominant pair since already the prefix $0^k$ of $\py$ satisfies $L(\px,0^k) = k = L(\px,\py)$. Hence $\px,\py$ can only correspond to a dominant pair if $L(\px,\py) = L(1^\kappa 0^k, \pz)$ and hence $1^\kappa 0^k, \pz$ correspond to a dominant pair of $1^\kappa 0^\mu$, $z$. This yields at most $d(1^\kappa 0^\mu, z)$ dominant pairs.
\end{itemize}

Finally, consider the case that $\px = 1^\kappa 0^\mu \pw$ with $\pw$ a prefix of $w$. There are no dominant pairs for $\py = 0^\ell$ with $\ell \in [\mu]$: Already for the prefix $1^\kappa 0^\ell$ of $\px$, we have $L(1^\kappa 0^\ell, 0^\ell) = \ell = L(\px,\py)$, hence these prefixes cannot correspond to dominant pairs. It remains to consider $\py = 0^\mu \pz$ for a prefix $\pz$ of $0^\nu z$. Again by \lemref{zeroblocklcs}, we have $L(\px,\py) = \mu + L(\pw,\pz)$ and hence such dominant pairs are in one-to-one correspondence with the dominant pairs of $w$ and $0^\nu z$. This yields at most $d(w,0^\nu z)$ further dominant pairs. 

By summing up over the three cases, we conclude that there are at most $\min\{\kappa,|z|\} + \mu + d(1^\kappa 0^\mu, z) + d(w, 0^\nu z)$ dominant pairs.
\end{proof}

We can finally pad to $\LCS(\Valpha,\{0,1\})$.

\begin{lem} \label{lem:MnSmallerdelta-step2}
Let $x,y,x',y',\tau$ be as in \lemref{MnSmallerdelta-step1}. We set $\kappa:= \lfloor M/n \rfloor$, $\tilde{\Delta} := \max\{\Delta,|y'|\}$, and $\tilde{m} := \max\{m, |y'|\}$ and define
\begin{alignat*}{4}
x'' & \; =  \quad 1^{\kappa+\tilde{\Delta}} & & \; 0^{\tilde{m}} \; x', \\
y'' & \; =  \quad 1^{\kappa} & & \; 0^{\tilde{m}} \; y'. 
\end{alignat*}
Then $x'',y''$ is an instance of $\LCS(\Valpha, \{0,1\})$. Moreover, we can compute a number $\tau'$ in time $\Oh(n)$ such that $L(x'',y'') = \tau' + L(x,y)$.
\end{lem}
\begin{proof}
Note that $\kappa = \Oh(M/n) = \Oh(m)$ by the parameter relation $M\le mn$. \lemref{MnSmallerdelta-step1}\itemref{MnSmallerdelta-step1size} yields $|x'|,|y'| = \Oh(m)$ and hence $\tilde{m} = \Theta(m)$ and $\tilde{\Delta} = \Theta(\Delta)$ (since $\alpha_m\le \alpha_\Delta = 1$). Thus, $|x''| = \kappa + \tilde{\Delta} + \tilde{m} + |x'| = \Theta(\Delta) + \Oh(m) = \Theta(n)$ since $\alpha_\Delta=1$. Furthermore, $|y''| = \kappa + \tilde{m} + |y'| = \tilde{m} + \Oh(m) = \Theta(m)$. \dopara{n,m}

Observe that $\tilde{\Delta}$ has been defined such that $|x''| \ge |y''|$. By greedy prefix matching and \lemref{deleteDeltaPadding}, we obtain
\begin{equation}\label{eq:MnSmallerdelta-step2lcs}
L(x'',y'') = \kappa + L(1^{\tilde{\Delta}} 0^{\tilde{m}} x', 0^{\tilde{m}} y') = \kappa + \tilde{m} + L(x',y').
\end{equation}
Since $L(x',y') = \tau + L(x,y)$, we satisfy the last claim by setting $\tau' := \kappa + \tilde{m} + \tau$.
Moreover, $L(x'',y'') = \tilde{m} + \Oh(m) = \Theta(m) = \Theta(L)$ since $|x'|,|y'| \le \Oh(m)$ and $\alpha_L = \alpha_m$. Furthermore, \eqref{eq:MnSmallerdelta-step2lcs} yields $\Delta(x'',y'') = \tilde{\Delta} + (|x'| - L(x',y')) = \Theta(\Delta) + \Oh(m) = \Theta(\Delta)$ and $\delta(x'',y'') = |y'|-L(x',y') = \Theta(\delta)$ by \lemref{MnSmallerdelta-step1}\itemref{MnSmallerdelta-step1delta}.\dopara{L,\delta,\Delta}

For the dominant pairs\dopara{d}, we apply \lemref{greedy} to bound $d(x'',y'') = \kappa + d(1^{\tilde{\Delta}} 0^{\tilde{m}} x', 0^{\tilde{m}} y')$. To bound the latter term, note that \lemref{deleteDeltaPadding} yields $d(1^{\tilde{\Delta}} 0^{\tilde{m}} x', 0^{\tilde{m}} y') \ge d(x',y') = \Omega(d)$ by \lemref{MnSmallerdelta-step1}\itemref{MnSmallerdelta-step1d}. For the upper bound, we first recall that $y' = 0^{\tilde{\delta}} b0^\ell y$ and that $\occ_1(y')\cdot |b0^\ell y| = \Oh(d)$ by \lemref{MnSmallerdelta-step1}\itemref{MnSmallerdelta-step1occ1}. Hence we have $d(1^{\tilde{\Delta}} 0^{\tilde{m}}, b 0^{\ell} y) \le 5 \cdot L(1^{\tilde{\Delta}} 0^{\tilde{m}}, b 0^{\ell} y) \cdot \occ_1(b0^\ell y) = \Oh(|b0^\ell y| \cdot \occ_1(y')) = \Oh(d)$. We can finally compute, using \lemref{deleteDeltaPadding},
\begin{align*}
 d(1^{\tilde{\Delta}} 0^{\tilde{m}} x', 0^{\tilde{m}} y') & \le \min\{\tilde{\Delta},|y'|\} + \tilde{m} + d(1^{\tilde{\Delta}} 0^{\tilde{m}}, b0^\ell y) + d(x',y') \\
& \le |y'| + \tilde{m} + \Oh(d) + d(x',y') = \Oh(d),
\end{align*}
where the last bound follows from $|y'|,\tilde{m} = \Oh(m) = \Oh(d)$ by the relation $d \ge L = \Omega(m)$ (since $\alpha_L=\alpha_m$) and $d(x',y') = \Oh(d)$.

It remains to count the number of matching pairs. We have $\occ_0(x''),\occ_0(y'') \le |x'|+|y'| + \tilde{m} =\Oh(m)$, as well as $\occ_1(x'') = \kappa + \tilde{\Delta} + \occ_1(x') = \Theta(\Delta) + \Oh(m) = \Theta(n)$ (since $\alpha_\Delta = 1$) and $\occ_1(y'') = \kappa + \occ_1(y') = \kappa + \Oh(M/n) = \Theta(M/n)$ by \lemref{MnSmallerdelta-step1}\itemref{MnSmallerdelta-step1occ1}. Thus $M(x'',y'') = \occ_1(x'')\occ_1(y'') + \occ_0(x'')\occ_0(y'') = \Theta(M) + \Oh(m^2) = \Theta(M)$, where the last bound follows from $M\ge L^2/|\Sigma| = \Omega(m^2)$ since $\alpha_L=\alpha_m$.
\end{proof}

Combining \lemrefs{MnSmallerdelta-step1}{MnSmallerdelta-step2} with \obsref{MnConstructionalphaprime} finally proves \lemref{MnSmallerdelta}.

\subsubsection{Case $\alpha_\Delta > \alpha_m = \alpha_L$ and $\alpha_\delta \le \alpha_M - 1$}

In this section, we consider the case where $\alpha_L=\alpha_m < \alpha_\Delta$ and $\alpha_\delta \le \alpha_M - 1$. In this case, we have $\Delta = \Theta(n) \gg m$ and $\delta \le  \Oh(M/n)$. Since $M/n \le \Oh(m) \le \Oh(\Delta)$, the fastest known algorithm runs in time $\tOh(\min\{d, \delta M/n\})$. Consequently, in this regime \lemref{largelcsbinaryalphabets} simplifies to the following statement.

\begin{lem}\label{lem:deltaSmallerMn}
Let $(\Valpha,\{0,1\})$ be a parameter setting satisfying \tabref{paramChoiceRestr} with $\alpha_L=\alpha_m < \alpha_\Delta$ and $\alpha_\delta \le \alpha_M -1$. There is a constant $\gamma \ge 1$ such that any algorithm for $\LCS^\gamma(\Valpha,\Sigma)$ takes time $\min\{d, \delta M/n\}^{1-o(1)}$ unless \OVH\ fails.
\end{lem}

As in the previous section, to prove this result we cannot simply pad instances of $\LCS_\le(\Valpha, \{0,1\})$ to $\LCS(\Valpha, \{0,1\})$, since the desired running time bound $\min\{d, \delta M/n\}$ is not monotone. Instead, we start from instances of a suitably chosen different parameter setting $\LCS_\le(\Valpha',\{0,1\})$ with $\alpha'_L = \alpha'_m$, i.e., we invoke \lemref{largeLCShardnessConstantAlphabet}.

\begin{obs}\label{obs:deltaSmallerMn-alphaprime}
Let $(\Valpha,\{0,1\})$ be a non-trivial parameter setting satisfying \tabref{paramChoiceRestr} with $\alpha_L = \alpha_m < \alpha_\Delta$ and $\alpha_\delta\le \alpha_M-1$. Define $\Valpha':=\Valpha'(\Valpha)$ by
\begin{align*}
\alpha'_\delta & = \min\{\alpha_\delta, \alpha_d /2 \}, & \alpha'_M & = 1 + \alpha'_m, \\
\alpha'_L=\alpha'_m & = \min\{\alpha_M - 1, \alpha_d - \alpha'_\delta\}, & \alpha'_\Delta & = 1, \\
\alpha'_d &= \min\{\alpha_\delta +\alpha_M - 1, \alpha_d\}, & \alpha'_\Sigma & = 0,
\end{align*}
Then the parameter setting $(\Valpha',\{0,1\})$ satisfies \tabref{paramChoiceRestr}. Furthermore, there is some constant $\gamma \ge 1$ such that no algorithm solves $\LCS^\gamma_\le(\Valpha', \{0,1\})$ in time $n^{\alpha'_d(1-o(1))} = \min\{d, \delta M/n\}^{1-o(1)}$ unless \OVH\ fails. This holds even restricted to instances $(n,x,y)$ with $|x|,|y| \le \gamma \cdot n^{\alpha'_m} = \Oh(\min\{M/n, \max\{d/\delta,\sqrt{d}\}\})$.
\end{obs}
\begin{proof}
We only discuss the inequalities from \tabref{paramChoiceRestr} that are not straight-forward to verify. To see $\alpha'_\delta \le \alpha'_m$, note that $\alpha_\delta \le \alpha_M - 1$ by assumption and $\alpha_d/2 = \alpha_d - \alpha_d/2 \le \alpha_d - \alpha'_\delta$. The inequality $\alpha'_L \le \alpha'_d$ follows from $\alpha_M - 1 \le \alpha_M -1 + \alpha_\delta$ and $\alpha_d - \alpha'_\delta \le \alpha_d$. Furthermore, $\alpha'_d \le 2\alpha'_L$ follows from $\alpha_\delta + \alpha_M - 1 \le 2(\alpha_M-1)$ (by assumption) and $\alpha_d = 2(\alpha_d - \alpha_d/2) \le 2(\alpha_d - \alpha'_\delta)$. From $\alpha'_d \le 2\alpha'_L$ and $\alpha'_L = \alpha'_m$ we also obtain $\alpha'_M = 1 + \alpha'_m = 1 + \alpha'_L \ge 1 +\alpha'_d - \alpha'_L$, which corresponds to the parameter relation $M \ge nd / (5L)$ that only holds for $\Sigma = \{0,1\}$. Finally, $\alpha'_M \ge \alpha'_d$ follows from $\alpha'_M = 1 + \alpha'_m = \min\{\alpha_M, 1 + \alpha_d - \alpha'_\delta\} \ge \min\{\alpha_M, \alpha_d\}$ by $\alpha'_\delta \le \alpha_\delta \le 1$ and similarly $\alpha'_d = \min\{\alpha_\delta + \alpha_M - 1, \alpha_d\} \le \min\{\alpha_M, \alpha_d\}$.

\lemref{largeLCShardnessConstantAlphabet} shows that some $\gamma \ge 1$ exists such that $\LCS^\gamma_\le(\Valpha', \{0,1\})$ cannot be solved in time $\min\{n^{\alpha'_d}, n^{\alpha'_\delta + \alpha'_m}, n^{\alpha'_\delta + \alpha'_\Delta}\}^{(1-o(1))}$, even restricted to instances $(n,x,y)$ with $|x|,|y| \le \gamma \cdot n^{\alpha'_m} = \Oh(\min\{M/n, \max\{d/\delta,\sqrt{d}\}\})$.
We simplify the running time bound by noting that $\alpha'_\Delta = 1 \ge \alpha'_m$, so that $n^{\alpha'_\delta + \alpha'_m} \le n^{\alpha'_\delta + \alpha'_\Delta}$. Moreover, we have $\alpha'_\delta + \alpha'_m = \alpha'_d =  \min\{\alpha_\delta + \alpha_M - 1, \alpha_d\}$. Indeed, if $\alpha_\delta \le \alpha_d/2$, we have $\alpha'_\delta = \alpha_\delta$ and hence $\alpha'_\delta + \alpha'_m = \min\{\alpha_\delta + \alpha_M - 1, \alpha_d\}=\alpha'_d$. Otherwise, $\alpha_d/2 < \alpha_\delta \le \alpha_M - 1$, forcing $\alpha'_\delta = \alpha_d/2$ and $\alpha'_m = \alpha_d/2$, which yields $\alpha'_\delta + \alpha'_m = \alpha_d = \min\{\alpha_\delta + \alpha_M - 1,\alpha_d\} = \alpha'_d$. Thus, $\min\{\alpha'_d, \alpha'_\delta + \alpha'_m, \alpha'_\delta + \alpha'_\Delta\} = \alpha'_d$ and the lower bound simplifies to $n^{\alpha'_d(1-o(1))} = \min\{ \delta M / n, d \}^{1-o(1)}$.
\end{proof}

In this section, to pad the number of dominant pairs, we will construct instances $x' = a 0^\ell x = (01)^{R+S} 0^\ell x$, $y'= b 0^\ell y = 0^R (01)^S 0^\ell y$, where we choose $R,S$ proportional to $n^{\alpha_R}, n^{\alpha_S}$ with 
\begin{align*}
\alpha_S &:= \min\{\alpha_M -1 , \max\{\alpha_d -\alpha_\delta, \alpha_d/2\}\}, & \alpha_R & := \alpha_d- \alpha_S.
\end{align*}
Note that the use of $\alpha_S, \alpha_R$ is a slight abuse of notation, since $R,S$ are not actual input parameters, but $\alpha_S, \alpha_R$ depend only on $\Valpha$. We will later set $R=c\cdot n^{\alpha_R}$, $S=c'\cdot n^{\alpha_S}$ with suitably large constants $c,c'$. However, depending on whether $\alpha_S \le \alpha_R$ or $\alpha_S > \alpha_R$, we will extend and analyze the basic construction differently. We start with the simpler case of $\alpha_S \le \alpha_R$.

\begin{lem}[Construction for $\alpha_S\le \alpha_R$] \label{lem:deltaSmallerMn-RgreaterS}
Let $(\Valpha,\{0,1\})$ be a parameter setting satisfying \tabref{paramChoiceRestr} with $\alpha_L = \alpha_m < \alpha_\Delta$, $\alpha_\delta \le \alpha_M -1$, and $\alpha_S \le \alpha_R$. Given an instance $(n,x,y)$ of $\LCS^\gamma_\le(\Valpha', \{0,1\})$ with $|y|\le|x| = \Oh(\min\{M/n, \max\{d/\delta, \sqrt{d}\}\})$, we use the parameters
\begin{align*}
S & = \max\{|x|, \lfloor n^{\alpha_S} \rfloor \}, & R& = \lfloor d/S \rfloor, \\
\ell & = R+S+|x|+|y|, & \beta & = \delta,
\end{align*}
to instantiate $x' = a 0^\ell x = (01)^{R+S} 0^\ell x$ and $y'= 0^\beta b 0^\ell y = 0^\beta 0^R (01)^S 0^\ell y$ as in \lemref{bbbII}. We furthermore set $\tilde{m} := \max\{m, \delta + R+2S + \ell + |y|\}$ and  $\kappa := \lfloor M/n \rfloor$ and define 
\begin{alignat*}{10}
x'' & \; := \quad 1^{\Delta + \kappa} & & \; 0^{\tilde{m}} \; x' \quad &=& \quad  1^{\Delta + \kappa} & & \; 0^{\tilde{m}} \; & a \; 0^\ell \; x \quad &=\quad  1^{\Delta + \kappa} & & \; 0^{\tilde{m}} \; & (01)^{R+S} \; 0^\ell \; x, \\
y'' & \; := \quad 1^{\kappa} & & \; 0^{\tilde{m}} \; y' \quad &=& \quad  1^{\kappa} & & \; 0^{\tilde{m}} \; 0^\delta \; &  b \; 0^\ell \; y \quad &=\quad  1^{\kappa} & & \; 0^{\tilde{m}} \; 0^\delta \; &  0^R (01)^S \; 0^\ell \; y.
\end{alignat*}
Then $x'',y''$ is an instance of $\LCS^{\gamma'}(\Valpha, \{0,1\})$ for some constant $\gamma'\ge1$ and can be constructed in time $\Oh(n)$ together with some integer $\tau$ such that $L(x'',y'') = \tau+ L(x,y)$.
\end{lem}
\begin{proof}
By greedy prefix matching and \lemref{zeroblocklcs}, we obtain
\begin{equation}\label{eq:SsmallerRlcs}
L(x'',y'') = \kappa + L(1^{\Delta} 0^{\tilde{m}} x', 0^{\tilde{m}} y') = \kappa + \tilde{m} + L(x',y') = \kappa + \tilde{m} + R+2S + \ell + L(x,y),
\end{equation}
where the last equality follows from \lemref{bbbII}, which applies since $S \ge |x|$. Clearly, $x'',y''$ and $\tau = \kappa + \tilde{m} + R+2S + \ell$ can be computed in time $\Oh(n)$.

It remains to verify that $x,y$ is an instance of $\LCS^{\gamma'}(\Valpha, \{0,1\})$ for some $\gamma'\ge 1$. We observe that $|x|,|y| = \Oh(n^{\alpha_S})$ by assumption, and hence $S = \Theta(n^{\alpha_S})=\Theta(\min\{M/n, \max\{d/\delta, \sqrt{d}\}\})$. Thus, $R = \Theta(d/S) = \Theta(n^{\alpha_R}) = \Oh(d n /M + \min\{\delta, \sqrt{d}\}) = \Oh(m)$, where the last bound follows from the parameter relations $M\ge nd/(5L) = \Omega(nd/m)$ (since $\alpha_L = \alpha_m$) and $\delta \le m$. By assumption $\alpha_S \le \alpha_R$, we have $S = \Oh(R)=\Oh(m)$. Furthermore, we have $\kappa = \Oh(M/n) = \Oh(m)$ by the relation $M\le mn$ and $\tilde{m} = \Theta(m)$ by $R, S, |x|,|y|, \ell = \Oh(m)$ and $\delta \le m$. Thus, $|x''| = \kappa + \Delta + \tilde{m} + 2(R+S) + \ell + |x| = \Theta(\Delta + m) = \Theta(n)$ and $|y''| = \kappa + \tilde{m} + \delta + R+2S + \ell + |y| = \Theta(m)$. By~\eqref{eq:SsmallerRlcs}, we also conclude that $L(x'',y'') = \tilde{m} + \Oh(m) = \Theta(m) = \Theta(L)$ by the assumption $\alpha_L = \alpha_m$.\dopara{n,m,L}

Note that by $\Delta \ge \delta$, $|a| \ge |b|$ and $|x| \ge |y|$, we have $|x''| \ge |y''|$. Hence by~\eqref{eq:SsmallerRlcs},  $\Delta(x'',y'') =  \Delta + R + (|x|-L(x,y)) = \Theta(\Delta)$, since $R, |x| = \Oh(m) = \Oh(\Delta)$ by $\alpha_\Delta > \alpha_m$. Likewise, \eqref{eq:SsmallerRlcs} yields $\delta(x'',y'') = \delta + (|y|-L(x,y)) = \delta + \delta(x,y) = \Theta(\delta)$ since $\delta(x,y) =\Oh(\delta)$ by $\delta(x,y) = \Oh(n^{\alpha_\delta'})$ and $\alpha'_\delta \le \alpha_\delta$. \dopara{\delta,\Delta}

For the dominant pairs, we first compute
\[ d(1^\Delta 0^{\tilde{m}} x', 0^{\tilde{m}} y')  \ge d(x',y') \ge d(a,0^\delta b) \ge R\cdot S = \Omega(d),\] 
using \lemref{deleteDeltaPadding}, \obsref{prefix}, and \lemref{genDomPairs}. For a corresponding upper bound, we use \lemref{deleteDeltaPadding} to obtain
\[ d(1^\Delta 0^{\tilde{m}} x', 0^{\tilde{m}} y') = d(1^\Delta 0^{\tilde{m}} x', 0^{\tilde{m}} 0^\delta b 0^\ell y) \le |y'| + \tilde{m} + d(1^\Delta 0^{\tilde{m}}, b 0^\ell y) + d(x', y'). \]
By \lemref{few1s} we have $d(1^\Delta 0^{\tilde{m}}, b 0^\ell y) \le 5 \cdot  L(1^\Delta 0^{\tilde{m}}, b0^\ell y)\cdot  \occ_1(b0^\ell y) = \Oh(|b0^\ell y| \cdot (S+|y|)) = \Oh(R\cdot S) = \Oh(d)$. Since $|y'| + \tilde{m} = \Oh(m) = \Oh(d)$ (using $d \ge L = \Omega(m)$ since $\alpha_L = \alpha_m$) and $d(x',y') \le 5 \cdot L(x',y') \cdot \occ_1(y') = \Oh(|x'|\cdot \occ_1(y')) = \Oh(R\cdot S) = \Oh(d)$, we conclude that $d(1^\Delta 0^{\tilde{m}} x', 0^{\tilde{m}} y') = \Theta(d)$. Finally, \lemref{greedy} yields $d(x'',y'') = \kappa + d(1^\Delta 0^{\tilde{m}} x', 0^{\tilde{m}} y') = \Theta(d) + \Oh(m) = \Theta(d)$, as desired. \dopara{d}

It remains to count the matching pairs. Note that $\occ_0(x''),\occ_0(y'') = \Oh(\tilde{m} + |x'|+|y'|) = \Oh(m)$. Furthermore $\occ_1(x'') = \Delta + \kappa + \occ_1(x') = \Theta(\Delta) + \Oh(m) = \Theta(n)$ (since $\alpha_\Delta > \alpha_m$ implies $\alpha_\Delta = 1$) and $\occ_1(y'') = \kappa + S + \occ_1(y) = \Theta(\kappa)= \Theta(M/n)$, where we used $S, |y| = \Oh(M/n)$ and $\kappa = \Theta(M/n)$ (since $M\ge n$). Thus, $M(x'',y'') = \occ_1(x'')\occ_1(y'') + \occ_0(x'')\occ_1(y'') = \Theta(n\cdot M/n) + \Oh(m^2) = \Theta(M)$ using that $M\ge L^2/|\Sigma| = \Omega(m^2)$ by $\alpha_L=\alpha_m$. Note that indeed $\Sigma(x'',y'') = \{0,1\}$. \dopara{M,\Sigma}
\end{proof}

Before giving the construction for the case $\alpha_S > \alpha_R$, we present a technical lemma that is similar to the dominant pair reduction technique of \lemref{dreduction}.

\begin{lem}\label{lem:advancedDeltapadding}
Let $x' = a0^\ell x = (01)^{R+S} 0^\ell x$, $y'=b 0^\ell y = 0^R (01)^S 0^\ell y$ be an instance of \lemref{bbbI} and $R \ge |y| - L(x,y)$. We set $t:= R+|y'|+1$ and define
\begin{alignat*}{4}
\bar{x} \;& := \quad 1^t \; 0^t \; y' \; & & 0^R \; 1^{t + \Delta} & \; 0^t \; x',\\
\bar{y} \;& := \quad                          & & 0^R \; 1^t & \; 0^t \; y'.
\end{alignat*}
Then
\begin{enumerate}[label=(\roman{*})]
\item \label{itm:advancedDeltapadding-L} $L(\bar{x}, \bar{y}) = R+2t + L(x',y')$,
\item \label{itm:advancedDeltapadding-dUB} $d(\bar{x}, \bar{y}) \le (2t+|y'|) (R+1) + R^2$,
\item \label{itm:advancedDeltapadding-dLB} $d(\bar{x}, \bar{y}) \ge R\cdot S$.
\end{enumerate}
\end{lem}
\begin{proof}
We first prove the following property.
\begin{enumerate}[label=($\ast$)]
\item \label{itm:advancedDeltapadding-helper} For any prefixes $\px,\py$ of $x',y'$, we have 
\[ L(1^t 0^t y' \, 0^R \, 1^{t + \Delta} 0^t \px, 0^R \, 1^t 0^t \py) = \max\{2t+|\py|, R+2t + L(\px,\py)\}. \] 
\end{enumerate}
Note that the lower bound immediately follows from either matching $1^t 0^t \py$ with $1^t 0^t y'$, or by matching $0^R 1^t 0^t \py$ with $0^R 1^{t+\Delta} 0^t \px$. For the upper bound, fix an LCS and observe that it cannot match symbols in the $1^t$-prefix of $\bar x$ and symbols in the $0^R$-prefix of $\bar y$, so at least one of the two prefixes stays unmatched. Thus,
\begin{align*}
L(1^t 0^t y' \, 0^R \, 1^{t + \Delta} 0^t \px, 0^R \, 1^{t} 0^t \py) & \le \max\{ L(0^t y' \, 0^R \, 1^{t + \Delta} 0^t \px, 0^R \, 1^{t} 0^t \py), \, L(1^t 0^t y' \, 0^R \, 1^{t + \Delta} 0^t \px, 1^{t} 0^t \py)\} \\
& = \max\{ R + L(0^{t-R} y' \, 0^R \, 1^{t + \Delta} 0^t \px, \, 1^{t} 0^t \py), \, 2t +  |\py|\}, 
\end{align*}
where the second line follows from greedy prefix matching. Setting $\hat{x}:=0^{t-R} y' \, 0^R \, 1^{t + \Delta} 0^t \px$ and $\hat{y}:= 1^{t} 0^t \py$, it remains to provide the upper bound $R+L(\hat{x},\hat{y}) \le\max\{2t+|\py|, R+2t+L(\px,\py)\}$ to prove \itemref{advancedDeltapadding-helper}. Assume that an LCS $z$ of $\hat{x}, \hat{y}$ matches less than $t - R$ symbols of the $1^t$-prefix of $\hat{y}$. Then $|z| \le t - R + L(\hat{x}, 0^t \py) \le 2t - R + |\py|$, yielding $R+L(\hat{x},\hat{y})\le 2t +|\py|$. Hence, assume instead that at least $t -R$ symbols of the $1^t$-prefix of $\hat{y}$ are matched. Since the number of ``1''s in the prefix $0^{t-R} y' 0^R$ of $\hat{x}$ is only $\occ_1(y')\le|y'| < t - R$, all zeroes of this prefix have to be deleted, resulting in
\begin{align*}
L(\hat x, \hat y) = |z| &  \le L(1^{\occ_1(y') + t +\Delta} 0^t \px, 1^{t} 0^t \py) \\
& = t + L(1^{\occ_1(y') + R + \Delta} 0^t \px, 0^t \py) \\
& = 2t + L(\px,\py),
\end{align*}
where the second line follows from greedy prefix matching and the third follows from \lemref{zeroblocklcs}. Thus, we have verified $R+L(\hat{x},\hat{y}) \le \max\{2t+|\py|, R+2t+L(\px,\py)\}$, which implies \itemref{advancedDeltapadding-helper}.

As an immediate consequence, \itemref{advancedDeltapadding-helper} yields $L(\bar{x},\bar{y}) = \max\{2t+|y'|, R + 2t+L(x',y')\} = R+2t + L(x',y')$ since $R\ge |y|-L(x,y) = |y'|-L(x',y')$ (where the equality follows from \lemref{bbbI}). This proves~\itemref{advancedDeltapadding-L}.

For~\itemref{advancedDeltapadding-dUB}, an application of \lemref{domPairRedbase} yields $d(\bar{x},\bar{y}) \le (2t + |y'|)(R+1) + R^2$.

Finally, for \itemref{advancedDeltapadding-dLB} we consider, analogous to \lemref{genDomPairs}, for any $0\le s\le S$ and $s \le r \le R+s$, the prefixes $\px = (01)^r$ of $x'$ and $\py = 0^R (01)^s$ of $y'$. Then by~\itemref{advancedDeltapadding-helper},
\begin{align*}
  L(1^t 0^t y' \, 0^R \, 1^{t + \Delta} 0^t \px, 0^R \, 1^t 0^t \py) &= \max\{2t + |\py|, R + 2t + L(\px,\py)\} \\
  &= \max\{2t + R + 2s, R+ 2t + r + s \} = (R+2t) + r+s, 
\end{align*}
where we used \lemref{genDomPairs}\itemref{genDomPairsHelper} for the second equality and $r\ge s$ for the last equality. Analogously to the proof of \lemref{genDomPairs}\itemref{genDomPairsd}, this yields $d(\bar{x},\bar{y})\ge R\cdot S$, since any feasible choice of $r,s$ gives rise to a unique dominant pair.
\end{proof}

We can now give the construction for $\alpha_S > \alpha_R$. Recall that
\begin{align*}
\alpha_S &:= \min\{\alpha_M- 1, \max\{\alpha_d-\alpha_\delta, \alpha_d/2\}\}, & \alpha_R & := \alpha_d- \alpha_S.
\end{align*}

\begin{lem}[Construction for $\alpha_S > \alpha_R$] \label{lem:deltaSmallerMn-SgreaterR}
Let $(\Valpha,\{0,1\})$ be a parameter setting satisfying \tabref{paramChoiceRestr} with $\alpha_L = \alpha_m < \alpha_\Delta$, $\alpha_\delta \le \alpha_M -1$, and $\alpha_S > \alpha_R$. Given an instance $(n,x,y)$ of $\LCS^\gamma_\le(\Valpha', \{0,1\})$ with $|y|\le|x| = \Oh(\min\{M/n, \max\{d/\delta, \sqrt{d}\}\})$, we can construct an instance $x^{(4)}, y^{(4)}$ of $\LCS^{\gamma'}(\Valpha, \{0,1\})$ (for some constant $\gamma'\ge 1$) in time $\Oh(n)$ together with some integer $\tau$ such that $L(x^{(4)},y^{(4)}) = \tau+ L(x,y)$.
\end{lem}
\begin{proof}
We first set $\ell_1 := |x|$ to define
\begin{alignat*}{3}
x^{(1)} & := & \; 0^{\ell_1} & \; x, \\
y^{(1)} & := 1^\delta & \; 0^{\ell_1} &\; y,
\end{alignat*}
which pads the parameter $\delta$. For convenience, define $\delta_1:= |y^{(1)}|-L(x^{(1)},y^{(1)})$. We use the parameters
\begin{align*}
S & := \lfloor n^{\alpha_S} \rfloor, & R& := \max\{\lfloor n^{\alpha_R} \rfloor, \delta_1\}, & \ell_2 & := |x^{(1)}| + |y^{(1)}|,  
\end{align*}
to define, as in \lemref{bbbI},
\begin{alignat*}{4}
x^{(2)} & := a & \; 0^{\ell_2} & \; x^{(1)}, \\
y^{(2)} & := b & \; 0^{\ell_2} & \; y^{(1)}. 
\end{alignat*}
We then use the dominant pair reduction trick of \lemref{advancedDeltapadding}, that additionally pads $\Delta$, and define 
\begin{alignat*}{4}
x^{(3)} & := 1^{\ell_3} \; 0^{\ell_3} \; y^{(2)} \; & & 0^R \; 1^{\ell_3 + \Delta} & \; 0^{\ell_3} \; x^{(2)},\\
y^{(3)} & :=                          & & 0^R \; 1^{\ell_3 } & \; 0^{\ell_3} \; y^{(2)},
\end{alignat*}
where $\ell_3 := R+|y^{(2)}| + 1$. The final instance is then constructed as
\begin{align*}
x^{(4)} &= 1^\kappa \; 0^m \; x^{(3)},\\
y^{(4)} &= 1^\kappa \; 0^m \; y^{(3)},
\end{align*}
where $\kappa:= \lfloor M/n\rfloor$.

We first compute
\begin{align}
L(x^{(4)}, y^{(4)}) & = \kappa + m + L(x^{(3)},y^{(3)}) \nonumber\\
& = \kappa + m + R + 2\ell_3 + L(x^{(2)}, y^{(2)}) \nonumber\\
& = \kappa + m + 2R + 2S + \ell_2 + 2\ell_3 + L(x^{(1)},x^{(1)}) \nonumber\\
& = \kappa + m + 2R + 2S + \ell_1 + \ell_2 + 2\ell_3 + L(x,y),\label{eq:lcslastcase}
\end{align}
where we used greedy prefix matching in the first line, \lemref{advancedDeltapadding}\itemref{advancedDeltapadding-L} in the second, \lemref{bbbI} in the third, and \lemref{zeroblocklcs} in the last line. Note that $x^{(4)}, y^{(4)}$, and $\tau := \kappa+m+2R+2S+\ell_1+\ell_2+2\ell_3$ can be computed in time $\Oh(n)$.

It remains to verify that $x^{(4)}, y^{(4)}$ is an instance of $\LCS^{\gamma'}(\Valpha,\{0,1\})$ for some $\gamma'\ge 1$.
We first observe that $|x|,|y| = \Oh(n^{\alpha_S})$ and hence $|x^{(1)}|,|y^{(1)}| = \Oh(n^{\alpha_S}+\delta)$. Note that by definition $S= \Theta(n^{\alpha_S})$. 

Assume for contradiction that $\alpha_R < \alpha_\delta$. Note that by definition $\alpha_R = \alpha_d- \alpha_S = \max\{\alpha_d - \alpha_M +1, \min\{\alpha_\delta,\alpha_d/2\}\}$ and hence $\alpha_R < \alpha_\delta$ only holds if $\alpha_d-\alpha_M+1 \le \alpha_R = \alpha_d/2 < \alpha_\delta$. But then $\alpha_d-\alpha_\delta \le \alpha_d/2 < \alpha_\delta \le \alpha_M - 1$. This forces $\alpha_S = \alpha_d/2$ by definition, which contradicts the assumption $\alpha_S > \alpha_R$. We therefore obtain $\alpha_R \ge \alpha_\delta$. 

Note that $\delta_1 = |y^{(1)}| - L(x^{(1)},y^{(1)}) = \delta + \delta(x,y)$ by \lemref{zeroblocklcs}. Since $\delta(x,y) \le \Oh(n^{\alpha'_\delta})$ and $\alpha'_\delta \le \alpha_\delta$, this yields $\delta_1 = \Theta(\delta)$, and hence $R = \Theta(n^{\alpha_R})$. Thus, $R = \Oh(S)$, since $\alpha_R < \alpha_S$. It is immediate that $|x^{(2)}|,|y^{(2)}| = \Oh(R+S+|x^{(1)}|+|y^{(1)}|) = \Oh(S)$. Furthermore, it follows that $|x^{(3)}| = \Delta + \Oh(S)$ and $|y^{(3)}| = \Oh(S)$. Finally $|y^{(4)}| = m + \Oh(M/n + S) = \Theta(m)$, where the last bound follows from $S = \Oh(M/n) = \Oh(m)$ by the parameter relation $M\le mn$. Likewise, $|x^{(4)}| = m + \Delta + \Oh(M/n+S) = \Theta(m+\Delta) = \Theta(n)$. Finally, $L(x^{(4)},y^{(4)}) = m + \Oh(M/n + S) = \Theta(m) = \Theta(L)$ by \eqref{eq:lcslastcase} and $\alpha_L = \alpha_m$. \dopara{n,m,L}

Since $|x|\ge |y|$, $|a|\ge |b|$ and $\Delta \ge \delta$ it is easy to see that $|x^{(4)}| \ge |y^{(4)}|$. Hence~\eqref{eq:lcslastcase} yields  \dopara{\delta,\Delta}
\[\Delta(x^{(4)}, y^{(4)}) = 2\ell_3 + |y^{(2)}| + \Delta + R + \Delta(x,y) = \Delta + \Oh(m) = \Theta(\Delta), \]
since in particular $\Delta(x,y) \le |x| = \Oh(m)$. Similarly, $\delta(x^{(4)},y^{(4)}) = \delta + \delta(x,y) = \Theta(\delta)$ as above.

For the number of dominant pairs, we observe that \lemref{advancedDeltapadding}\itemref{advancedDeltapadding-dLB} yields $d(x^{(3)},y^{(3)}) \ge R\cdot S = \Omega(d)$. From \lemref{advancedDeltapadding}\itemref{advancedDeltapadding-dUB}, the corresponding upper bound $d(x^{(3)},y^{(3)}) \le (2\ell_3 + |y^{(2)}|) \cdot (R+1) + R^2 = \Oh( S\cdot R + R^2) = \Oh(d)$ follows, since $R= \Oh(S)$ by $\alpha_R < \alpha_S$. Thus, by \lemref{greedy} we obtain $d(x^{(4)},y^{(4)}) = \kappa + m + d(x^{(3)},y^{(3)}) = \Theta(d) + \Oh(m) = \Theta(d)$ by $d\ge L = \Omega(m)$ since $\alpha_L= \alpha_m$.

It remains to count the number of matching pairs. We have $\occ_1(y^{(4)}) = \kappa + \occ_1(y^{(3)}) = \Theta(M/n)$, since $\kappa = \Theta(M/n)$ by the parameter relation $M \ge n$, and $\occ_1(y^{(3)}) \le |y^{(3)}| = \Oh(S) = \Oh(M/n)$. Since $|y^{(4)}| = \Oh(m)$, we have $\occ_0(y^{(4)}) = \Oh(m)$. Note that  $\occ_1(x^{(4)}) = \kappa + 2\ell_3 + \Delta + |x^{(2)}| + |y^{(2)}| = \Delta + \Oh(S+ \kappa)  = \Theta(n)$, since $\alpha_\Delta > \alpha_m$ implies $\alpha_\Delta = 1$. Finally, $\occ_0(x^{(4)}) = m + 2\ell_3 + R + \occ_0(y^{(2)}) + \occ_0(x^{(2)}) = \Oh(m)$. Thus, we obtain $M(x^{(4)}, y^{(4)}) = \occ_1(x^{(4)})\cdot \occ_1(y^{(4)}) + \occ_0(x^{(4)})\cdot \occ_0(y^{(4)})  = \Theta(n \cdot M/n) + \Oh(m^2) = \Theta(M)$ by the relation $M\ge L^2/|\Sigma| = \Omega(m^2)$, since $\alpha_L=\alpha_m$.
\end{proof}

Note that combining \lemrefs{deltaSmallerMn-RgreaterS}{deltaSmallerMn-SgreaterR} with \obsref{deltaSmallerMn-alphaprime} yields \lemref{deltaSmallerMn}.

\section{New Algorithm for Binary Alphabet}
\label{sec:algo}

In this section we prove \thmref{algo}, i.e., we assume that $\Sigma = \{0,1\}$ and provide an algorithm running in time $\Oh(n + \delta M / n)$. More precisely, for any input $x,y$, by $\occ_0(x) + \occ_1(x) = n$ we have $\max\{\occ_0(x), \occ_1(x)\} \ge n/2$, so without loss of generality assume $\occ_1(x) \ge n/2$ (otherwise exchange 0 and 1). Since $M = \occ_0(x) \cdot \occ_0(y) + \occ_1(x) \cdot \occ_1(y)$, it follows that $\occ_1(y) \le 2M/n$. Hence, it suffices to design an algorithm running in time $\Oh(n + \delta \cdot \occ_1(y))$. 

\begin{thm}
  For $\Sigma = \{0,1\}$, LCS has an $\Oh(n + \delta \cdot \occ_1(y))$ algorithm.
\end{thm}

We preprocess $x$ in time $\Oh(n)$ to support the following queries. For $\sigma \in \{0,1\}$, $0 \le i \le n$, and $t \ge 1$, $\Next_\sigma^t(i)$ returns the position of the $t$-th occurrence of symbol $\sigma$ after position $i$ in~$x$, i.e., $\Next_\sigma^t(i) = i'$ if and only if $x[i'] = \sigma$ and $\occ_\sigma(x[i+1..i']) = t$ (if such a number $i'$ does not exist then $\Next_\sigma^t(i) := \infty$). For convenience, we let $\Next^0_\sigma(i) := i$.
For $t = 1$ we also write $\Next_\sigma^1(i) = \Next_\sigma(i)$. It is easy to implement $\Next_\sigma^t$ in time $\Oh(1)$ using rank/select data structures on the $0$'s and $1$'s in~$x$, which can be built in time~$\Oh(n)$~\cite{Jacobson89,Patrascu08}. The symbol succeeding $i$ is $\Next_{\Sigma}(i) := i+1$ if $i+1 \le n$, or $\infty$ otherwise, which can be computed in time $\Oh(1)$. 

Let $\lambda = \occ_1(y)$ and $1 \le j_1 < \ldots < j_{\lambda} \le m$ be the positions of all $1$'s in $y$, and for convenience set $j_0 := 0$.  We can assume that the last symbol in each of $x$ and $y$ is a $1$, in particular $j_\lambda = m$, because appending symbol $1$ to both $x$ and $y$ increases the LCS by exactly 1 (by \lemref{greedy}). We write $z_\ell$ for the number of $0$'s between $j_{\ell-1}$ and $j_\ell$ in $y$, i.e., $y = 0^{z_1} 1 0^{z_2} 1 \ldots 1 0^{z_\lambda} 1$ with $z_\ell \ge 0$.

Consider the dynamic programming table $T$ that contains for all $0 \le \ell \le \lambda$ and $k \ge 0$ (it remains to fix an upper bound on $k$) the value
\begin{align} \label{eq:defTable}
  T[\ell,k] = \min\{ 0 \le i \le n \mid L(x[1..i],y[1..j_\ell]) = j_\ell - k \}, 
\end{align}
where we set $\min \emptyset = \infty$.
Observe that from $T$ we can read off the LCS length as $L(x,y) = m - \min\{ k \mid T[\lambda,k] < \infty \}$. 
In particular, we may initialize $\tilde \delta := 1$, and compute the table $T$ for $0 \le \ell \le \lambda$ and $0 \le k \le \tilde \delta$. If there is no $0 \le k \le \tilde \delta$ with $T[\lambda,k] < \infty$ then we double $\tilde \delta$ and repeat. This exponential search ends once we find a value $\tilde \delta \in [\delta, 2 \delta)$.

Next we show how to recursively compute $T[\ell,k]$.
For $\ell = 0$, we clearly have $T[0,0] = 0$ and $T[0,k] = \infty$ for any $k > 0$. For $\ell > 0$, the following dynamic programming recurrence computes $T[\ell,k]$, as shown in Lemma~\ref{lem:correctnessSigmaTwoNumberOne} below. 
\begin{align}
  T[\ell,k] = \min\big\{ &\min\{ \Next_\Sigma( \Next_0^{z_\ell-k+k'}(T[\ell-1,k']) ) \mid \max\{0,k-z_\ell\} \le k' < k \},  \notag \\ 
  &\Next_1( \Next_0^{z_\ell}(T[\ell-1,k]) ),   \label{eq:algTable} \\
  & T[\ell-1,k-z_\ell-1] \big\}.  \notag
\end{align}
Note that the third line only applies if $k-z_\ell-1 \ge 0$, as $T[\ell',k'] = \infty$ for $k' < 0$. 

Let us discuss how to efficiently implement the above algorithm, assuming we already have computed the values $T[\ell-1,k]$, $0 \le k \le \tilde \delta$. Clearly, we can evaluate the second and third line in constant time, using the $\Next$ data structures that we built in the preprocessing.
For the first line, observe that $\Next_0^{t}(i)$ is the position of the $(t + \occ_0(x[1..i]))$-th $0$ in $x$. Hence, $\Next_0^{z_\ell-k+k'}(T[\ell-1,k'])$ is the position of the $(z_\ell-k+k' + \occ_0(x[1..T[\ell-1,k']]))$-th $0$ in $x$, so it is minimized if $k' + \occ_0(x[1..T[\ell-1,k']])$ is minimized\footnote{Here we interpret $\occ_0(x[1..\infty])$ as $\infty$.}. Thus, it suffices to compute a range minimum query over the interval $[\max\{0,k-z_\ell\}, k)$ on the array $A_\ell[0..\tilde \delta]$ with $A_\ell[k'] := k' + \occ_0(x[1..T[\ell-1,k']])$. From the answer $A_\ell[r]$ to this range minimum query we can infer $T[\ell,k]$ in time $\Oh(1)$. Specifically, the first line evaluates to the next symbol after the position of the $(z_\ell - k + A[r])$-th $0$ in $x$, i.e., $\Next_\Sigma( \Next_0^{z_\ell-k+A_\ell[r]}(0) )$.  

Note that range minimum queries can be performed in time $\Oh(1)$, after a preprocessing of $\Oh(|A_\ell|) = \Oh(\tilde \delta)$~\cite{Sadakane07,BrodalDR12}, where $|A_\ell|$ is the size of array $A_\ell$. Since we can reuse the array $A_\ell$ for all $0 \le k \le \tilde \delta$, we spend (amortized) preprocessing time $\Oh(1)$ per entry of $T[\ell,\cdot]$. In total, this yields time $\Oh(\tilde \delta \cdot \lambda) = \Oh(\tilde \delta \cdot \occ_1(y))$ to build the table $T$. The exponential search for $\delta$ yields time $\Oh(\delta \cdot \occ_1(y))$. Adding the preprocessing time $\Oh(n)$, we obtain an $\Oh(n + \delta \cdot \occ_1(y))$ algorithm. It remains to prove correctness of the recursive formula (\ref{eq:algTable}).

\begin{lem} \label{lem:correctnessSigmaTwoNumberOne}
  Table (\ref{eq:defTable}) follows the recursive formula (\ref{eq:algTable}).
\end{lem}
\begin{proof}
  For any $1 \le \ell \le \lambda$ and $0 \le k \le \tilde \delta$ we show that the value $T[\ell,k]$ of (\ref{eq:defTable}) follows the recursive formula (\ref{eq:algTable}).
  Let $i = T[\ell,k]$ and let $i'$ be minimal with 
  \begin{align} \label{eq:split} 
    L(x[1..i],y[1..j_\ell]) = L(x[1..i'],y[1..j_{\ell-1}]) + L(x[i'+1..i],y[j_{\ell-1}+1..j_\ell]).
  \end{align} 
  Let $k' = j_{\ell-1} - L(x[1..i'],y[1..j_{\ell-1}])$. Then we claim $i' = T[\ell-1,k']$. Indeed, since $i'$ satisfies the condition $L(x[1..i'],y[1..j_{\ell-1}]) = j_{\ell-1} - k'$ of (\ref{eq:defTable}) we have $i' \ge T[\ell-1,k']$. Moreover, if we had $i' > T[\ell-1,k']$ then we could replace $i'$ by $T[\ell-1,k']$, as both values satisfy the condition $L(x[1..i'],y[1..j_{\ell-1}]) = j_{\ell-1} - k'$, contradicting minimality of $i'$.
  
  Let $r = L(x[i'+1..i],y[j_{\ell-1}+1..j_\ell])$. By (\ref{eq:split}) we have $j_\ell - k = j_{\ell-1} - k' + r$, and we obtain $r = 1 + z_\ell - k + k'$ using $z_\ell = j_{\ell} - j_{\ell-1} - 1$. Note that $i \ge i'$ is the smallest value attaining $L(x[i'+1..i],y[j_{\ell-1}+1..j_\ell]) = r$. Indeed, if there was a smaller value $i' \le i^* < i$ with $L(x[i'+1..i^*],y[j_{\ell-1}+1..j_\ell]) = r$, then $L(x[1..i^*],y[1..j_\ell]) \ge L(x[1..i'],y[1..j_{\ell-1}]) + L(x[i'+1..i^*],y[j_{\ell-1}+1..j_\ell]) = L(x[1..i'],y[1..j_{\ell-1}]) + L(x[i'+1..i],y[j_{\ell-1}+1..j_\ell]) = L(x[1..i],y[1..j_\ell]) = j_\ell - k$.
  Then there also exists $0 \le i'' \le i^* < i$ with equality, i.e., $L(x[1..i''],y[1..j_\ell]) = j_\ell - k$. Indeed, if $L(x[1..i^*],y[1..j_\ell]) >  j_\ell - k$ then we can repeatedly reduce $i^*$ by 1, this reduces $L(x[1..i^*],y[1..j_\ell])$ by at most 1, and we eventually reach $j_\ell - k$ since $L(x[1..t],y[1..j_\ell]) = 0$ for $t=0$. However, existence of $i'' < i$ contradicts minimality of $i = T[\ell,k]$. 
  
  Now we show that $i$ is one of the terms on the right hand side of (\ref{eq:algTable}), considering three cases.
  
  Case 1: If $1 \le r < z_\ell + 1$, then the LCS of $x[i'+1..i]$ and $y[j_{\ell-1}+1..j_\ell] = 0^{z_\ell} 1$ consists of $r-1$ $0$'s and one additional symbol which is $0$ or $1$. Thus, the smallest $i$ attaining $r$ is $\Next_\Sigma(\Next_0^{r-1}(i'))$, accounting for $r-1$ $0$'s and one additional symbol. Since $r-1 = z_\ell - k + k'$ and $i' = T[\ell-1,k']$, we have shown that $i = T[\ell,k]$ is of the form $\Next_\Sigma( \Next_0^{z_\ell - k + k'}(T[\ell-1,k']) )$ for some $k'$. Observe that $1 \le r < z_\ell + 1$ implies $k - z_\ell \le k' < k$. We clearly also have $k' \ge 0$. This corresponds to the first line of (\ref{eq:algTable}).
  
  Case 2: If $r = z_\ell+1$ then $x[i'+1..i]$ contains $y[j_{\ell-1}+1..j_\ell] = 0^{z_\ell} 1$. Thus, $i = \Next_1( \Next_0^{z_\ell}( i' ) )$, accounting for $z_\ell$ $0$'s followed by a $1$. In this case, we have $k' = k + r - z_\ell - 1 = k$ so that $i = T[\ell,k]$ is of the form $\Next_1( \Next_0^{z_\ell}( T[\ell-1,k] ) )$. This corresponds to the second line of (\ref{eq:algTable}).
  
  Case 3: If $r = 0$ then $i = i'$, since the smallest value $i \ge i'$ attaining $L(x[i'+1..i],y[j_{\ell-1}+1..j_\ell]) = 0$ is $i'$. In this case, we have $k' = k-z_\ell-1$, and we obtain $T[\ell,k] = i = i' = T[\ell-1,k'] = T[\ell-1,k-z_\ell-1]$. This only applies if $k-z_\ell-1 \ge 0$. This corresponds to the third line of (\ref{eq:algTable}).
  
  This case distinction shows that $i$ is one of the terms on the right hand side of (\ref{eq:algTable}). 
  Also observe that we have $i \le \Next_1( \Next_0^{z_\ell}( T[\ell-1,k] ) )$, since the number $\Next_1( \Next_0^{z_\ell}( T[\ell-1,k] ) )$ is part of the set of which $i = T[\ell,k]$ is the minimum. Similarly, we have $i \le \Next_\Sigma( \Next_0^{z_\ell - k + k'}(T[\ell-1,k']) )$ for any $\max\{0, k-z_\ell\} \le k' < k$, and $i \le T[\ell-1,k-z_\ell-1]$ if $k-z_\ell-1 \ge 0$. This proves that $i$ is the minimum over all expressions on the right hand side of (\ref{eq:algTable}), proving that $i = T[\ell,k]$ follows the recursive formula (\ref{eq:algTable}).
\end{proof}

\section{Strengthening Hardness via BP-SETH}
\label{sec:bpseth}

A recent and surprising result by Abboud, Hansen, Virginia Williams and Williams~\cite{AbboudHVWW16} proves conditional lower bounds for LCS and related problems under a natural weaker variant of SETH, called BP-SETH. In this variant, the role of CNF formulas in SETH is replaced by branching programs, providing a much more expressive class of Boolean functions -- intuitively, the corresponding satisfiability problem becomes much harder. As a consequence, refuting a conditional lower bound based on BP-SETH would yield stronger algorithmic consequences, strengthening the conditional lower bound significantly. Furthermore, Abboud et al.\ show that even a sufficiently strong \emph{polylogarithmic} improvement for LCS would imply faster (formula) satisfiability algorithms than currently known.

In this section, we show how to adapt the proofs of our conditional lower bounds to also hold under BP-SETH. To this end, we first introduce this weaker assumption, briefly state the main construction of Abboud et al.\ (specialized to LCS) and then describe the necessary modifications to our proofs.

\paragraph{Branching Programs and BP-SETH.} 
Branching programs provide a popular model for non-uniform computation. Formally, a \emph{branching program $P$ on $N$ variables, length $T$ and width $W$} consists of a directed graph $G$, whose vertex set $V$ is divided into $T$ disjoint layers $V_1, \dots, V_T$. Each layer contains at most $W$ vertices. Every edge in $G$ starts in some layer $V_j$ and ends at the following layer $V_{j+1}$. Each edge is annotated with a constraint $X_i = b$, where $X_1, \dots, X_N$ are the Boolean input variables and $b\in \{0,1\}$. The constraints of all edges starting in layer $V_j$ must use the same variable $X_i$, but may have potentially different values of $b$. There are two distinguished nodes, namely some \emph{start node} $v_0\in V_1$, and an \emph{accept node} $v^* \in V_T$.
Given an assignment $X\in \{0,1\}^N$ to the variables, we say that $P$ accepts $x$ if and only if there is a path from $v_0$ to $v^*$ such that the constraints on all edges on the path are satisfied by $X$.  

The corresponding satisfiability problem BP-SAT asks, given a branching program $P$ on $N$ variables, length $T$ and width $W$, to determine whether there exists an assignment $X\in \{0,1\}^N$ that is accepted by $P$. The results in~\cite{AbboudHVWW16} rely on the following hypothesis.

\begin{hypo}
For any $\varepsilon>0$, BP-SAT with $(\log W)(\log T) = o(N)$ cannot be solved in time $O((2-\varepsilon)^N)$.
\end{hypo}

Note that this hypothesis is much weaker than SETH: By Barrington's theorem, already branching programs of constant width and $T=\mathrm{polylog}(N)$ can simulate any \textsf{NC} circuit, and thus any CNF formula.

\paragraph{Central Construction of Abboud et al.}

We present an intermediate step of the construction of Abboud et al.\ in a convenient form for our purposes. 

\begin{lem}[{Normalized BP-Vector Gadgets~\cite[Implicit lemma in the proof of Theorem 6, specialized to LCS]{AbboudHVWW16}}]
\label{lem:BP-NVG}
Let $P$ be a branching program on $N_1+N_2$ variables, width $W$ and length~$T$. Let $\sA = \{a_1,\dots,a_A\} \subseteq \{0,1\}^{N_1}, \sB = \{b_1,\dots,b_B\} \subseteq \{0, 1\}^{N_2}$ be sets of partial assignments. In linear time in the output size, we can construct binary strings  $x_1, x_2,\dots, x_A$ of length $\ell_\x$ and $y_1,y_2,\dots, y_B$ of length $\ell_\y$ (called \emph{normalized vector gadgets}) with the following properties.
\begin{enumerate}
\item There is some function $f(W,T) = T^{\Oh(\log W)}$ with $\ell_x, \ell_\y \le f(W,T)$,
\item We can compute $\rho_0, \rho_1$ with $\rho_0 > \rho_1$ such that $L(x_i, y_j) \ge \rho_0$ if the assignment given by assigning $a_i$ to the variables $X_1,\dots,X_{N_1}$ and assigning $b_j$ to $X_{N_1+1},\dots,X_{N_1+N_2}$ is accepted by $P$. Otherwise, we have $L(x_i,y_j) = \rho_1$.
\end{enumerate}
\end{lem}

\paragraph{Proof modifications.}

Equipped with the above lemma, we sketch how to modify our lower bounds to also hold under BP-SETH. 
Consider the following problem \emph{LCS-Pair with Guarantees}: Given binary strings $x_1, x_2,\dots, x_A$ of length $\ell_\x$ and $y_1,y_2,\dots, y_B$ of length $\ell_\y$, where $\ell_\x, \ell_\y = (A+B)^{o(1)}$, and given integers $\rho_0 > \rho_1$, decide whether either (1) $L(x_i,y_j) \ge \rho_0$ for some $i,j$, or (2) $L(x_i,y_j) = \rho_1$ for all $i,j$. The naive solution for this problem runs in time $(AB)^{1+o(1)}$. 

\lemref{BP-NVG} shows that any $\Oh((AB)^{1-\eps})$-time algorithm for LCS-Pair with Guarantees would break BP-SETH. Indeed, given a branching program $P$ on $N$ variables, length $T$ and width $W$, we first split the variables into two sets of size $N/2$, and let $\sA, \sB$ be the sets of all assignments to these two sets of variables. Then testing whether there are partial assignments $a_i \in \sA$ and $b_j \in \sB$ that together are accepted by $P$ is equivalent to deciding satisfiability of $P$. Since $A = |\sA| = B = |\sB| = 2^{N/2}$, any $\Oh((AB)^{1-\eps})$-time algorithm would break BP-SETH. Moreover, in BP-SETH we can assume $\log T \cdot \log W = o(N)$ and thus $\ell_\x,\ell_\y = T^{\Oh(\log W)} = 2^{o(N)} = (A+B)^{o(1)}$, as required in the definition of LCS-Pair with Guarantees. Furthermore, the same $(AB)^{1-o(1)}$ lower bound under BP-SETH also holds restricted to instances with $B = A^{\beta \pm o(1)}$ for any $0 < \beta \le 1$. To see this, in the above reduction instead of splitting the $N$ variables into   equal halves $N_1=N_2=N/2$, we choose $N_1,N_2$ such that $N_2 \approx \beta N_1$, so that the sets of partial assignments $\sA = \{0,1\}^{N_1}$ and  $\sB = \{0,1\}^{N_2}$ satisfy $B = |\sB| = |\sA|^{\beta\pm o(1)} = A^{\beta \pm o(1)}$.

Observe that the normalized vector gadgets constructed in \lemref{coreSingleStrings} show the similar claim that any $\Oh((AB)^{1-\eps})$-time algorithm for LCS-Pair with Guarantees would break the OV hypothesis, and thus SETH. Since all reductions presented in this paper use normalized vector gadgets either directly via \lemref{coreSingleStrings} or indirectly via \lemref{core}, they all implicitly go via LCS-Pair with Guarantees. Hence, we can easily replace the first part of our reductions, i.e., the reduction from SAT/OV to LCS-Pair with Guarantees, by the reduction from BP-SAT to LCS-Pair with Guarantees given in \lemref{BP-NVG}. This yields a conditional lower bound based on BP-SETH. 

There are two steps where we have to be careful: First, LCS-Pair with Guarantees does not immediately give properties (i) and (iv) of \lemref{coreSingleStrings}, however, they can be ensured easily as in the proof of \lemref{coreSingleStrings}. Second, \lemref{core} is the construction from~\cite{BringmannK15}, and thus to check that it goes via LCS-Pair with Guarantees one needs to check that the proof in~\cite{BringmannK15} indeed only uses \lemref{coreSingleStrings}. 
Along these lines we obtain the following strengthening of our main result.

\begin{thm}
\thmrefs{main1}{main2} also hold after replacing SETH by BP-SETH.
\end{thm}

\bibliographystyle{plain}
\bibliography{lcsAndEdit}

\end{document}